\documentclass[11pt]{article}

\usepackage{amssymb,amsmath,bm}
\usepackage{amsthm}

\usepackage{textcomp}
\usepackage{enumerate}      
\usepackage{graphicx}        
\usepackage{caption}

\usepackage{tabu}

\usepackage[export]{adjustbox}

\usepackage{mathrsfs}

\usepackage{url}

\usepackage{array}
\newcommand{\PreserveBackslash}[1]{\let\temp=\\#1\let\\=\temp}
\newcolumntype{C}[1]{>{\PreserveBackslash\centering}p{#1}}
\newcolumntype{R}[1]{>{\PreserveBackslash\raggedleft}p{#1}}
\newcolumntype{L}[1]{>{\PreserveBackslash\raggedright}p{#1}}

\renewcommand{\setminus}{{\smallsetminus}}

\usepackage{amssymb,amsmath,bm}
\usepackage{amsthm}
\usepackage{hyperref}
\usepackage{mathrsfs}
\usepackage{textcomp}
\usepackage{enumerate}
\usepackage{graphicx}
\usepackage{tikz}
\usepackage{verbatim}

\newtheorem{theorem}{Theorem}[section]
\newtheorem{lemma}[theorem]{Lemma}
\newtheorem{proposition}[theorem]{Proposition}
\newtheorem{definition}[theorem]{Definition}
\newtheorem{corollary}[theorem]{Corollary}

\theoremstyle{remark}
\newtheorem{remark}[theorem]{Remark}

\theoremstyle{remark}

\numberwithin{equation}{section}
\newcommand{\ads}{\mathbb A\mathrm d\mathbb S}

\begin{document}
\title{\bf Asymptotics of $b$-$6j$ symbols and anti-de Sitter tetrahedra}

\author{Tianyue Liu, Shuang Ming, Xin Sun, Baojun Wu and Tian Yang}

\date{}

\maketitle

\begin{abstract}
In this paper, we study the asymptotics of the $6j$-symbols for the principal series of the modular double of $\mathrm U_q\mathfrak{sl}(2;\mathbb R)$, and of their analytic extension -- what we call the $b$-$6j$ symbols, relating them in various cases to the volume of truncated hyperideal  tetrahedra in the hyperbolic and the anti-de Sitter geometry. To the best of our knowledge, this is the first time that the  anti-de Sitter geometry appears in
the asymptotics of quantum invariants.
In addition, based on the connection to conformal field theory,  we reveal a correspondence between the edge lengths and the dihedral angles of truncated hyperideal anti-de Sitter tetrahedra and the Fenchel-Nielsen coordinates of hyperbolic four-holed spheres. We also provide a concrete instance of $3D/2D$ holography, 
in the spirit of the  AdS/CFT correspondence.
\end{abstract}

\tableofcontents

\section{Introduction}

The relationship between the asymptotics of  $6j$-symbols and the volume of tetrahedra is an extensively studied topic, whose significance is largely driven by its profound implications for 
three-dimensional quantum gravity. Inspired by the work of Wigner \cite{W}~, Ponzano and Regge \cite{PR} proposed a relation between the {Wigner} $6j$-symbols of $\mathfrak{sl}(2;\mathbb C)$ and the Euclidean tetrahedra, which was rigorously proved in \cite{R}. The $6j$-symbols of the quantum group  $\mathrm U_q\mathfrak{sl}(2;\mathbb C)$  was considered in \cite{TW,C,CM,BY2}, which in various cases were related to the tetrahedra in the spherical, Euclidean and hyperbolic geometry. In our previous work \cite{LMSWY}, we considered the $6j$-symbols associated to the principal series (also known as the positive representations) of the modular double of $\mathrm U_{q}\mathfrak{sl}(2;\mathbb R)$, denoted by 
$\mathrm U_{q\tilde{q}}\mathfrak{sl}(2;\mathbb R)$, which is a quantum group of significant interest in mathematical physics \cite{Fad95, Fad16, PT99,TesVar} and representation theory \cite{FI14}. When the six parameters of a $\mathrm U_{q\tilde{q}}\mathfrak{sl}(2;\mathbb R)$  $6j$-symbol scale according to the edge lengths of a truncated hyperideal hyperbolic tetrahedron or a flat tetrahedron, we established the connection between the semi-classical limit of the $6j$-symbol 
and the hyperbolic volume of the tetrahedron. These $6j$-symbols were recently studied in the physics literature \cite{CHJL, CEZ, Hartman1, Hartman2} as the building blocks of a $(2+1)$-dimensional topological quantum field theory (TQFT), called the Virasoro TQFT, which is closely related to the $3D$ quantum  anti-de Sitter (AdS) gravity. The Virasoro TQFT is also believed to be equivalent to the Teichm\"uller TQFT considered by Andersen-Kashaev \cite{AK}. Our previous work \cite{LMSWY} initiated a mathematical study of the Virasoro TQFT by defining the Turaev-Viro type invariants for hyperbolic $3$-manifolds with totally geodesic boundary, and obtained a thorough understanding of their semi-classical limit and its relationship with the hyperbolic volume and the adjoint twisted Reidemeister torsion of the manifolds. 
\medskip

In this paper, we complete the study of the asymptotics of the $\mathrm U_{q\tilde{q}}\mathfrak{sl}(2;\mathbb R)$ $6j$-symbols. As an intermediate step, we showed that if the six parameters of a $\mathrm U_{q\tilde{q}}\mathfrak{sl}(2;\mathbb R)$ $6j$-symbol do not correspond to the edge lengths of a truncated hyperideal hyperbolic tetrahedron nor a flat tetrahedron, then they must correspond to the edge lengths of a truncated hyperideal anti-de Sitter tetrahedron (see Theorem \ref{classification}). Then we established the connection between the semi-classical limit of the $6j$-symbol 
and the anti-de Sitter volume of the tetrahedron in this case (see Theorem \ref{cov}). 
The $\mathrm U_{q\tilde{q}}\mathfrak{sl}(2;\mathbb R)$ $6j$-symbol, as a function of its six parameters, has a meromorphic extension out of the spectrum of the principal series of $\mathrm U_{q\tilde{q}}\mathfrak{sl}(2;\mathbb R)$, which we call the $b$-$6j$ symbol (see Definition \ref{def: b-6j}).
We also study the asymptotics of these  $b$-$6j$ symbols when the six parameters scale according to the dihedral angles of a truncated hyperideal hyperbolic tetrahedron or a truncated hyperideal anti-de Sitter tetrahedron (both beyond the spectrum of the principal series of $\mathrm U_{q\tilde{q}}\mathfrak{sl}(2;\mathbb R)$), relating them to the volume of the tetrahedron in the corresponding geometry (see Theorem \ref{vol3} and Theorem \ref{vol2}). The asymptotics of the $b$-$6j$ symbols was previously considered in the physics literature \cite{CHJL, TesVar, Hartman2}. However, to the best of our knowledge, the volume of anti-de Sitter tetrahedra has never appeared prior to our work. To prove the main results, we establish several  basic properties of truncated hyperideal anti-de Sitter tetrahedra, including criteria of the edge lengths and the dihedral angles, and the  volume formulae, which are also of independent interest (see Subsection \ref{sec1.3}). 

The $b$-$6j$ symbols play an important role in conformal field theory (CFT). On the one hand, they describe the fusion kernel of the  Virasoro conformal blocks,  which is intimately related to the quantization of the Teichm\"uller space. In light of this, we find an explicit correspondence between the edge lengths and the dihedral angles of truncated hyperideal anti-de Sitter tetrahedra and the Fenchel-Nielsen coordinates of hyperbolic four-hole spheres (see Subsection \ref{dig1}), realizing a proposal from the physics literature~\cite{NRS}.
On the other hand, the $b$-$6j$ symbols describe the boundary three-point function of the Liouville CFT. From the path integral perspective, the classical Liouville action  arises in the semi-classical limit of the $b$-$6j$ symbols, which is connected to $3D$ hyperbolic and $3D$ anti-de Sitter geometry  through our Theorem \ref{Thm1.1} and Theorem \ref{cov}. We view this correspondence as an instance of the $3D/2D$ holography, in the spirit of the AdS/CFT correspondence. See Subsection \ref{CFT} for more discussion.


\subsection{Definition and known results}
Let $S_b$ be the  \emph{double sine function} defined for $z\in \mathbb C$ with $0<\mathrm{Re}(z)<Q$ by
\begin{equation*}
S_b(z)=\exp\Bigg(\int_\Omega\frac{\sinh\Big(\big(\frac{Q}{2}-z\big)t\Big)}{4t\sinh(\frac{bt}{2})\sinh(\frac{t}{2b})}dt\Bigg),
\end{equation*}
where $b\in (0,1)$ and $Q=b+b^{-1}$, 
and the contour $\Omega$ goes along the real line and passes above the pole of the integrand at the origin.

We call a six-tuple $(a_1,\dots,a_6)$ of complex numbers  \emph{$b$-admissible} if for all $i,j\in\{1,2,3,4\},$ 
$$0<\mathrm{Re}q_j-\mathrm{Re}t_i<Q,$$
 where 
\begin{equation*}
\begin{split}
   & t_1=a_1+a_2+a_3,\quad  t_2=a_1+a_5+a_6,\quad t_3=a_2+a_4+a_6,\quad  t_4=a_3+a_4+a_5,\\
   & q_1=a_1+a_2+a_4+a_5,\quad q_2=a_1+a_3+a_4+a_6,\quad  
q_3=a_2+a_3+a_5+a_6\quad\text{and}\quad q_4=2Q.
\end{split}
\end{equation*}
In particular, if $(a_1,\dots,a_6)$ is $b$-admissible, then we have
$$\max\{\mathrm{Re}t_1,\mathrm{Re}t_2,\mathrm{Re}t_3,\mathrm{Re}t_4\}<\min\{\mathrm{Re}q_1,\mathrm{Re}q_2,\mathrm{Re}q_3,\mathrm{Re}q_4\}.$$

\begin{definition}[$b$-$6j$ symbols]\label{def: b-6j}
The \emph{$b$-$6j$ symbol} for a $b$-admissible six-tuple $(a_1,\dots,a_6)$ is given by 
\begin{equation*}
\bigg\{\begin{matrix} a_1 & a_2 & a_3 \\ a_4 & a_5 & a_6 \end{matrix} \bigg\}_b=\Bigg(\frac{1}{\prod_{i=1}^4\prod_{j=1}^4S_b(q_j-t_i)}\Bigg)^{\frac{1}{2}}\int_\Gamma \prod_{i=1}^4S_b(u-t_i)\prod_{j=1}^4S_b(q_j-u)d u,
\end{equation*}
where the contour $\Gamma$ is any vertical line passing the interval $(\max\{\mathrm{Re}t_i\},\min\{\mathrm{Re}q_j\})$.
\end{definition}

It is proved in Proposition \ref{lem: ab conver} that this integral converges absolutely, and is independent of the choice of the vertical line $\Gamma.$


\begin{remark}
When the six-tuples $(a_1,\dots,a_6)$ are of the form $\Big(\frac{Q}{2}+\mathbf i\frac{l_1}{2\pi b},\dots,\frac{Q}{2}+\mathbf i\frac{l_6}{2\pi b}\Big) $ for some $(l_1,\dots,l_6)\in\mathbb R^6_{>0},$  the $b$-$6j$ symbols are the $6j$-symbols for the principal series of the quantum group $\mathrm U_{q\tilde q}\mathfrak{sl}(2;\mathbb R)$, the modular double of $\mathrm U_{q}\mathfrak{sl}(2;\mathbb R).$ See e.g. \cite[Section 5.1]{LMSWY} for a brief overview. Although conceptually important, the technical part of the current paper does not rely on this  quantum group interpretation.
\end{remark}

In the previous work \cite{LMSWY}, we studied the asymptotics of the $\mathrm U_{q\tilde q}\mathfrak{sl}(2;\mathbb R)$ $6j$-symbols, and obtained the following results. 

\begin{theorem}\label{Thm1.1}\cite[Theorem 1.2 and Theorem 1.3]{LMSWY}
\begin{enumerate}[(1)]
\item If $(l_1,\dots, l_6)\in{\mathbb R_{>0 }^6}$ are the lengths of the edges of a truncated hyperideal tetrahedron $\Delta$, then as $b\to 0,$
$$\bigg\{\begin{matrix} \frac{Q}{2}{+}\mathbf i\frac{l_1}{2\pi b} & \frac{Q}{2}{+} \mathbf i\frac{l_2}{2\pi b} & \frac{Q}{2}{+} \mathbf i\frac{l_3}{2\pi b} \\ \frac{Q}{2}{+} \mathbf i\frac{l_4}{2\pi b} & \frac{Q}{2}{+} \mathbf i\frac{l_5}{2\pi b} & \frac{Q}{2}{+} \mathbf i\frac{l_6}{2\pi b} \end{matrix} \bigg\}_b=\frac{e^{\frac{-\mathrm{Cov}(\Delta)}{\pi b^2}}}{\sqrt[4]{-\det\mathrm{Gram}(\Delta)}}\Big(1+O\big(b^2\big)\Big),$$
where $\mathrm{Cov}(\Delta)$ is the co-volume of $\Delta$, and  $\mathrm{Gram}(\Delta)$ is the Gram matrix of $\Delta$ in the edge lengths.
\item If  $(l_1,\dots, l_6)\in{\mathbb R_{>0}^6}$ are the lengths of the edges of a flat tetrahedron, then as $b\to 0,$
$$\lim_{b\to 0}\pi b^2\log\bigg\{\begin{matrix} \frac{Q}{2}{+} \mathbf i\frac{l_1}{2\pi b} & \frac{Q}{2}{+} \mathbf i\frac{l_2}{2\pi b} & \frac{Q}{2}{+} \mathbf i\frac{l_3}{2\pi b} \\ \frac{Q}{2}{+} \mathbf i\frac{l_4}{2\pi b} & \frac{Q}{2}{+} \mathbf i\frac{l_5}{2\pi b} & \frac{Q}{2}{+} \mathbf i\frac{l_6}{2\pi b} \end{matrix} \bigg\}_b=-\widetilde{\mathrm{Cov}}(l_1,\dots,l_6),$$
where $\widetilde{\mathrm{Cov}}(l_1,\dots,l_6)$ is the extended co-volume function of $(l_1,\dots,l_6)$ defined in \cite{LY}.
\item  If $(l_1,\dots, l_6)\in{\mathbb R_{>0}^6}$ are neither the lengths of the edges of a truncated hyperideal tetrahedron nor the lengths of the edges a flat tetrahedron, then
$$\limsup_{b\to 0}\pi b^2\log\Bigg|\bigg\{\begin{matrix} \frac{Q}{2}{+} \mathbf i\frac{l_1}{2\pi b} & \frac{Q}{2}{+} \mathbf i\frac{l_2}{2\pi b} & \frac{Q}{2}{+} \mathbf i\frac{l_3}{2\pi b} \\ \frac{Q}{2}{+} \mathbf i\frac{l_4}{2\pi b} & \frac{Q}{2}{+} \mathbf i\frac{l_5}{2\pi b} & \frac{Q}{2}{+} \mathbf i\frac{l_6}{2\pi b} \end{matrix} \bigg\}_b\Bigg|= -\widetilde{\mathrm{Cov}}(l_1,\dots,l_6),$$
where $\widetilde{\mathrm{Cov}}(l_1,\dots,l_6)$ is the extended co-volume function of $(l_1,\dots,l_6)$ defined in \cite{LY}.
\end{enumerate}
\end{theorem}

The main results of the current  article are listed in Subsections \ref{sec1.1}, \ref{sec1.2} and \ref{sec1.3} below. 

\subsection{Asymptotics of  $\mathrm U_{q\tilde{q}}\mathfrak{sl}(2;\mathbb R)$ $6j$-symbols} \label{sec1.1}

In Section \ref{sec2}, we first give the following classification of six-tuples $(l_1,\dots,l_6)$ in $\mathbb R_{>0}^6$.
\bigskip

\noindent {\bf Theorem  \ref{classification}.}
{\em Any  $(l_1,\dots, l_6)\in \mathbb R^6_{>0}$ is exclusively the six-tuple of the edge lengths of:
\begin{enumerate}[(1)]
\item a truncated hyperideal hyperbolic tetrahedron, 
\item a (truncated hyperideal) flat tetrahedron, or
\item a truncated hyperideal anti-de Sitter tetrahedron.  
\end{enumerate}}
\medskip

Based on this classification then, we in Subsection \ref{el} obtain  a refinement of Theorem \ref{Thm1.1} relating the semi-classical limit of the $b$-$6j$ symbols in case (3) to geometric quantities  of a truncated hyperideal tetrahedron in the anti-de Sitter space. 

\begin{theorem}\label{cov} Let $(l_1,\dots, l_6)\in\mathbb R_{>0 }^6$ be the lengths of the edges of a truncated  hyperidel  anti-de Sittertetrahedron $\Delta.$ Then  as $b\to 0,$
\begin{equation*}\label{AdS/CFT}
\bigg\{\begin{matrix} \frac{Q}{2} +  \mathbf i\frac{l_1}{2\pi b} & \frac{Q}{2}+\mathbf i\frac{l_2}{2\pi b} & \frac{Q}{2}+\mathbf i\frac{l_3}{2\pi b} \\ \frac{Q}{2}+\mathbf i\frac{l_4}{2\pi b} & \frac{Q}{2}+ \mathbf i\frac{l_5}{2\pi b} & \frac{Q}{2}+\mathbf i\frac{l_6}{2\pi b} \end{matrix} \bigg\}_b=\frac{e^{\frac{-\widetilde{\mathrm{Cov}}(l_1,\dots,l_6)}{\pi b^2}} }{\sqrt[4]{\det\mathrm{Gram}(\Delta)}} \Bigg(2\cos\bigg( {\frac{\mathrm{Cov}(\Delta)}{\pi b^2}}+\frac{\pi}{4}\bigg) +O\big(b^2\big)\Bigg),
\end{equation*}
where $\widetilde{\mathrm{Cov}}(l_1,\dots,l_6)$ is the extended co-volume function of $(l_1,\dots,l_6)$ defined in \cite{LY}, $\mathrm{Cov}(\Delta)$ is the anti-de Sitter co-volume of $\Delta$, and $\mathrm{Gram}(\Delta)$ is the Gram matrix of $\Delta$ in the edge lengths.
\end{theorem}

See also Theorem \ref{covpm}, and see Definition \ref{vol-cov} and Definition \ref{Gramel} for the relevant definitions.

\subsection{Asymptotics of $b$-$6j$ symbols}\label{sec1.2}

In Subsection \ref{da}, we also obtain the semi-classical limit of the $b$-$6j$ symbols when the parameters are rescaled  according to the dihedral angles of either a truncated hyperideal hyperbolic tetrahedron, or a truncated hyperideal anti-de Sitter tetrahedron.

\begin{theorem}\label{vol3} Let $(\theta_1,\dots, \theta_6)$ be the dihedral angles of a truncated hyperideal hyperbolic tetrahedron $\Delta.$ Then as $b\to0,$
$$\bigg\{\begin{matrix} \frac{Q}{2} + \frac{\theta_4}{2\pi b} & \frac{Q}{2} +  \frac{\theta_5}{2\pi b} & \frac{Q}{2} + \frac{\theta_6}{2\pi b} \\ \frac{Q}{2} +  \frac{\theta_1}{2\pi b} & \frac{Q}{2} +  \frac{\theta_2}{2\pi b} & \frac{Q}{2} +  \frac{\theta_3}{2\pi b} \end{matrix} \bigg\}_b=\frac{e^{\frac{-\mathrm{Vol}(\Delta)}{\pi b^2}}}{\sqrt[4]{-\det\mathrm{Gram}(\Delta)}} \Big(1 +O\big(b^2\big)\Big),$$
where $\mathrm{Vol}(\Delta)$ is the hyperbolic volume of $\Delta$, and $\mathrm{Gram}(\Delta)$ is the Gram matrix of $\Delta$ in the dihedral angles. 
\end{theorem} 

See also Theorem \ref{vol}.

\begin{theorem}\label{vol2}
Let $(\theta_1,\dots, \theta_6)$ be the dihedral angles of  a truncated hyperideal anti-de Sitter tetrahedron $\Delta.$ Then  as $b\to 0,$
\begin{equation*}\label{AdS/CFT}
\bigg\{\begin{matrix}  \frac{Q}{2}+\frac{\theta_4}{2\pi b} & \frac{Q}{2}+ \frac{\theta_5}{2\pi b} & \frac{Q}{2}+\frac{\theta_6}{2\pi b}\\\frac{Q}{2} + \frac{\theta_1}{2\pi b} & \frac{Q}{2}+\frac{\theta_2}{2\pi b} & \frac{Q}{2}+\frac{\theta_3}{2\pi b}  \end{matrix} \bigg\}_b=\frac{e^{\frac{-\mathrm{Vol}(\Delta)\cdot \mathbf i}{\pi b^2} }}{\sqrt[4]{-\det\mathrm{Gram}(\Delta)}} \Big(1 +O\big(b^2\big)\Big),
\end{equation*}
where $\mathrm{Vol}(\Delta)$ is the anti-de Sitter volume of $\Delta$, and $\mathrm{Gram}(\Delta)$ is the Gram matrix of $\Delta$ in the dihedral angles.
\end{theorem}

See also Theorem \ref{vol2pm}. See Definition \ref{vol-cov} and Definition \ref{Gramda} for the relevant definitions, and see also Remark \ref{ccov}.

\subsection{Hyperideal tetrahedra in the anti-de Sitter space}\label{sec1.3}

As a preparation for the proof of Theorem \ref{cov} and Theorem \ref{vol2}, we recall in Section \ref{sec2} the basics of truncated hyperideal tetrahedra in the anti-de Sitter space (Definition \ref{truncated}), including the edge lengths (Definition \ref{edgelength}), the dihedral angles (Definition \ref{dihedralangle}), the Gram matrices in both the edge lengths and the dihedral angles (Definitions \ref{Gramel} and \ref{Gramda}), the volume and the co-volume (Definition \ref{vol-cov}), and the Schl\"afli formulae (Propositions \ref{Sch} and \ref{CoSch}). 

We also obtain below the criteria of both the edge lengths and the dihedral angles of a truncated hyperideal anti-de Sitter tetrahedron.
\medskip

\noindent{\bf Theorem \ref{criterion}.}  {\em For $(l_{1},\dots, l_{6})\in \mathbb R^6_{>0},$ the following are equivalent:
\begin{enumerate}[(1)]
\item $(l_{1},\dots, l_{6})$ are the edge lengths of a truncated hyperideal anti-de Sitter tetrahedron.

\item The Gram matrix of $(l_{1},\dots, l_{6})$ has signature $(2,2).$ 
\end{enumerate}}
\medskip

\noindent{\bf Theorem \ref{criterion2}.} 
{\em For a six-tuple $(\theta_1,\dots,\theta_6)$ in $\mathbb A$ (see Eq. (\ref{A}) for the definition), the following are equivalent:
\begin{enumerate}[(1)]
\item $(\theta_1,\dots,\theta_6)$ are the dihedral angles of a truncated hyperideal anti-de Sitter  tetrahedron. 

\item The Gram matrix $G$ of $(\theta_1,\dots,\theta_6)$ satisfies the following three conditions:
\begin{enumerate}[(a)]
\item $G$ has signature $(2,2).$
\item For each $i\in\{1,2,3,4\},$ $G_{ii}>0.$
\item For each $\{i,j\}\subset\{1,2,3,4\},$ $G_{ij}<0.$
\end{enumerate}
\end{enumerate}}

In Section \ref{sec3}, we  obtain explicit formulae of the volume and the co-volume of a truncated hyperideal anti-di Sitter tetrahedron,
 respectively in terms of the edge lengths and in terms of the dihedral angles.
 \medskip

\noindent {\bf Theorem \ref{ab}} and {\bf Theorem  \ref{volume}.}
{\em Let $\Delta$ be a truncated hyperideal anti-de Sitter tetrahedron  with edge lengths $(l_1,\dots ,l_6),$  then the co-volume of $\Delta$ can be computed by
$$\mathrm{Cov}(\Delta)=\frac{1}{2}\cdot \mathrm{Re} W(l_1,\dots,l_6),$$
and the  volume of $\Delta$ can be computed by
$$\mathrm{Vol}(\Delta)= {\frac{1}{2}}\Bigg(\mathrm{Re}W(l_1,\dots,l_6)-\sum_{k=1}^6\frac{\partial \mathrm{Re}W}{\partial l_k}(l_1,\dots,l_6)\cdot l_k\Bigg),$$
where $W$ is the function defined therein.
}
\medskip

\noindent{\bf Theorem \ref{volume2}.} {\em Let $\Delta$ be a truncated hyperideal anti-de Sitter tetrahedron with dihedral angles $(\theta_1,\dots ,\theta_6),$  then the volume of $\Delta$ can be computed by
$$\mathrm{Vol}(\Delta)=\frac{1}{2} \cdot W(\theta_1,\dots,\theta_6),$$
where $W$ is the function defined therein.}
\bigskip

\subsection{Connection to conformal field theory}\label{CFT}

In Subsection \ref{dig1}, we  reveal a correspondence between the edge lengths and the dihedral angles of truncated hyperideal anti-de Sitter tetrahedra and the Fenchel-Nielsen coordinates of hyperbolic four-holed spheres
and oriented hyperbolic four-hole spheres with geodesic boundary. See Formulae (\ref{lengths=}) and (\ref{angles=}).
In this subsection, we provide the CFT background leading to this purely geometric observation. 

The starting point of our observation is that the $b$-6j symbol agrees with the fusion kernel for the Virasoro conformal blocks. $2D$ CFT has a rich algebraic structure coming from its symmetry algebra - the Virasoro algebra. 
According to the conformal bootstrap program \cite{BPZ84}, the correlation functions of a $2D$ CFT can be expressed as a family of special functions, determined by the Virasoro algebra, called conformal blocks, together with model-dependent data, including the spectrum and structure constants. This expression depends on a pants decomposition of the surface associated with the correlation function. 
The most basic setting for the conformal bootstrap is the case of the four-puncture sphere, and the fusion kernel describes how the conformal blocks are transformed when switching between pants decompositions. The fusion kernel, together with the modular kernel coming from the one-puncture torus, governs the transformation rule of the conformal blocks associated with the general surfaces.  See \cite{Teschner03,CEZ}for the physics account of this picture.  A central assertion in \cite{Teschner03,CEZ}, originally due to Ponsot and Teschner \cite{PT99}, is that the fusion kernel equals the $b$-6j symbol up to a simple normalization factor. Furthermore, the spectrum for the principal series of $\mathrm U_{q\tilde{q}}\mathfrak{sl}(2;\mathbb R)$ agrees with the spectrum of a CFT of central importance called the Liouville CFT. Recently, substantial progress has been made in the mathematical understanding of Virasoro conformal blocks based on the probabilistic approach to Liouville CFT \cite{RV25}. In particular, a proof of Ponsot and Teschner's conjecture for the fusion kernel was given in \cite{GRSSW}.
Based on the connection to $4D$ supersymmetric gauge theory~\cite{AGT2010}, 
it was suggested in \cite{NRS} that the semi-classical limits of the fusion kernels generate the Fenchel-Nielsen coordinates of hyperbolic four-hole spheres. Our Theorem \ref{cov}, combined with the result in Subsection \ref{dig1}, provides a mathematical realization of this proposal.

On the other hand, it was conjectured in~\cite{PT02}, and recently rigorously proved  in \cite{ARSZ23}, that the boundary three-point function for the Liouville CFT agrees with the $b$-$6j$ symbol after a suitable reparameterization.
From the path integral perspective, the semi-classical limit of Liouville correlation functions can be expressed in terms of the classical Liouville action. For closed surfaces, this was rigorously established in~\cite{LRV2019}; and the boundary analog of it, which includes the case of the boundary three-point function, was recently announced in \cite{Cer25}. In light of this, our Theorem \ref{Thm1.1} and Theorem \ref{cov} yield a relation between the classical Liouville action and the $3D$ hyperbolic and the $3D$ anti-de Sitter geometry.  This may be viewed as an instance of the $3D/2D$ holography in the sense of the conjectural AdS/CFT correspondence \cite{Ma1998,Witten1998}. 
For other instances of holography, a notion of $3D$ hyperbolic volume called the renormalized volume was related to
the classical Liouville action on the boundary  of Schottky and quasi-Fuchsian $3$-manifolds \cite{TT2003,KS2008}.
See also the relation between the universal Liouville action and the renormalized volume \cite{BBPW2025}, which can be interpreted as a holographic description of the Schramm-Loewner Evolution (SLE). To the best of our knowledge, an explicit AdS/Liouville relation as the one we established in this article has never appeared before, even in the physics literature. 

We believe that  all the discussion above is only tip of an iceberg connecting the $3D$ anti-de Sitter geometry with the $2D$ conformal field theory, and plan to pursue the exploration along this direction.

\bigskip

\noindent\textbf{Acknowledgments.} We thank Qiyu Chen, Pengyu Le, Feng Luo, Sara Maloni, Jean-Marc Schlenker, and  Yilin Wang for inspiring discussions. We also thank Igor Frenkel for showing interest to this work and encouraging the authors to write it down. T.L., X.S., and B.W.\ are supported by National Key R\&D Program of China (No.\ 2023YFA1010700). S.M.\ is supported by National Natural Science Foundation of China (No.12371124). T.Y.\ is supported by NSF Grant  DMS-2505908. 

\section{Hyperideal tetrahedra in the anti-de Sitter space}\label{sec2}

\subsection{$3$-dimensional anti-de Sitter space}

The Lorentzian space $\mathbb E^{2,2}$ is the vector space $\mathbb R^4$ with the inner product $\langle,\rangle$ defined for  $\mathbf x = (x_1,x_2,x_3,x_4)$ and $\mathbf y= (y_1,y_2,y_3,y_4)$ by
$$\langle \mathbf x, \mathbf y\rangle=x_1y_1+x_2y_2-x_3y_3-x_4y_4.$$
The \emph{anti-de Sitter space} is
 $$\mathbb A\mathrm d\mathbb S^3= \{\mathbf v\in\mathbb E^{2,2}\ |\ \langle \mathbf v,\mathbf v \rangle =-1 \},$$ 
 and the \emph{dual anti-de Sitter space} is
 $$\mathbb A\mathrm d\mathbb  S^{3^*}= \{\mathbf v\in\mathbb E^{2,2}\ |\ \langle \mathbf v,\mathbf v \rangle =1 \}.$$ 
 The inner product $\langle,\rangle$ on $\mathbb E^{2,2}$ restricts to a semi-Riemannian metric on the tangent space of each point of $\mathbb A\mathrm d\mathrm S^3$ with signature $(2,1)$ and constant curvature $-1.$  We call $\ads^3$ with this induced metric the \emph{Lorentzian model} of the $3$-dimensional anti-de Sitter space.

The geodesics  in $\ads^3$ are the  intersections of $\ads^3$ with $2$-dimensional subspaces of $\mathbb E^{2,2};$ and the totally geodesic planes in $\ads^3$ are the  intersections of $\ads^3$ with $3$-dimensional subspaces of $\mathbb E^{2,2}.$  For $\mathbf v\in\mathbb A\mathrm d\mathbb S^{3}\cup \mathbb A\mathrm d\mathbb S^{3^*},$ let 
$\boldsymbol \Pi_{\mathbf v}=\{ \mathbf w \in \mathbb E^{2,2}\ | \langle \mathbf v, \mathbf w \rangle =0\}$
be the hyperplane containing all the vectors perpendicular to $\mathbf v,$ and let 
$$\Pi_{\mathbf v} = \boldsymbol \Pi_{\mathbf v}\cap \mathbb A\mathrm d\mathbb S^3.$$ 
The totally geodesic plane $\Pi_{\mathbf v}\subset \mathbb A\mathrm d\mathbb S^3$ is \emph{space-like} if $ \mathbf v\in\mathbb A\mathrm d\mathbb S^{3},$ and is  \emph{time-like}  if $\mathbf v\in\mathbb A\mathrm d\mathbb S^{3^*}.$ The restriction of the inner product $\langle, \rangle$ to each space-like totally geodesic plane is positive definite, and to each time-like totally geodesic plane has signature $(1,1).$ 

The group of isometries of $\ads^3$ is the orthogonal group $\mathrm O(2,2),$ which  includes transformations that can reverse spatial orientation or time orientation. The subgroup that preserves both spatial and time orientation is the identity component of $\mathrm O(2,2)$, which is the special orthogonal group $\mathrm{SO}(2,2).$
 \bigskip

\begin{remark}Let the affine hyperplane $\mathbb P^3_1=\{\mathbf x\in\mathbb E^{2,2}\ | x_4=1\},$ and let 
$\mathrm{proj}: \mathbb E^{2,2} \setminus \{\mathbf x\in\mathbb E^{2,2}\ | x_4=0\}\to \mathbb P^3_1$
be the radial projection along the ray from the origin $\mathbf 0.$ Then $\mathrm{proj}$ continuously extends to
$\mathrm{proj}:\mathbb E^{2,2}\setminus\{\mathbf 0\}\to\mathbb P^3_1\cup\mathbb P^3_{\infty},$
where  $\mathbb P^3_\infty$ is the set of lines in the linear subspace $\{\mathbf x\in\mathbb E^{2,2}\ | x_4=0\}$ passing through the origin $\mathbf 0.$ Let 
$\mathbb P^3=\mathbb P^3_1\cup \mathbb P^3_{\infty},$
then the radial projection $\mathrm{proj}$ restricts to a two-to-one map from $\ads^3$ to $\mathbb P^3,$ and the image 
$\mathrm{AdS}^3\doteq \mathrm{proj} (\mathbb A\mathrm d\mathbb S^3)$
with the metric  induced from $\langle,\rangle$ via $\mathrm{proj}$ is the \emph{projective model} of the $3$-dimensional anti-de Sitter space. 
\end{remark}

\subsection{Truncated hyperideal tetrahedra} 

Let 
$$\mathbb B^{2,2}=\{\mathbf v\in\mathbb E^{2,2}\ |\ \langle \mathbf v,\mathbf v \rangle <0 \},$$
which is the pre-image  of the projective model $\mathrm{AdS}^3$ of the anti-de Sitter space under the radial projection $\mathrm{proj}:\mathbb E^{2,2}\setminus \mathbf \{0\}\to\mathbb P^3.$

A \emph{hyperideal tetrahedron} in $\ads^3$ is defined to be  a quadruple  of linearly independent vectors $\{\mathbf v_1,$ $\mathbf v_2,$ $\mathbf v_3,\mathbf v_4\}$ in  $\mathbb A\mathrm d\mathbb S^{3^*},$   such that  for each pair $\{i,j\}\subset\{1,2,3,4\},$ the straight  line  segment $L_{ij}$ connecting $\mathbf v_i$ and $\mathbf v_j$ intersects $\mathbb B^{2,2}.$ The vectors $\mathbf v_1,\dots, \mathbf v_4$ are the \emph{hyperideal vertices}, or simply the \emph{vertices}, of the hyperideal tetrahedron. Two hyperideal tetrahedra $\{\mathbf v_1,$ $\mathbf v_2,$ $\mathbf v_3,\mathbf v_4\}$ and $\{\mathbf v'_1,$ $\mathbf v'_2,$ $\mathbf v'_3,\mathbf v'_4\}$ are \emph{isometric} if there is a spatial-orientation and time-orientation preserving isometry of $\mathbb E^{2,2}$ sending $\mathbf v_i$ to $\mathbf v_i'$, for each $i\in\{1,2,3,4\}.$

\begin{remark} The radial projection of the convex hull of the vertices of a hyperideal tetrahedron in $\ads^3$  is then a tetrahedron in an affine chart of $\mathbb P^3$ whose vertices are all outside of the projective model $\mathrm{AdS}^3$ and whose edges all intersect $\mathrm{AdS}^3,$ which coincides with the definition of the hyperideal polyhedron in \cite{CS}.
 \end{remark} 

For each $i\in\{1,\dots,4\},$  the \emph{face} of the hyperideal tetrahedron  opposite to $\mathbf v_i$ is the totally geodesic plane  $\mathbf F_i$  containing $\mathbf v_j,$ $\mathbf v_k$ and $\mathbf v_l,$ $\{j,k,l\}=\{1,2,3,4\}\setminus\{i\},$ i.e.,
$$\mathbf F_i=\{c_j\mathbf v_j+c_k\mathbf v_k+c_l\mathbf v_l\ |\ c_j,c_k,c_l\in\mathbb R\}\cap \mathbb A\mathrm d\mathbb S^3.$$
 The \emph{outward unit normal vector} of $\mathbf F_i$ is then the unique vector $\mathbf u_i$ satisfying:
\begin{enumerate}[(1)]
\item 
$\langle \mathbf u_i,\mathbf v_i\rangle >0,$ i.e., $\mathbf u_i$ is outward, 
\item $|\langle \mathbf u_i,\mathbf u_i\rangle| =1,$ i.e., $\mathbf u_i$ is unit, and 
\item 
$\langle \mathbf u_i,\mathbf v\rangle =0$ for any vector $\mathbf v\in \mathbf F_i,$ i.e., $\mathbf u_i$ is normal. 
\end{enumerate}
Condition (1) is to make sure that $\mathbf u_i$ and $\mathbf v_i$ are on the two different sides of $\mathbf F_i;$ and from Condition (3), we have $ \mathbf F_i=\Pi_{\mathbf u_i}.$ 

For each $i\in\{1,\dots,4\},$  the \emph{plane of truncation} at $\mathbf v_i$ is the plane  $\mathbf T_i$  containing the outward unit normal vectors  $\mathbf u_j,$ $\mathbf u_k$ and $\mathbf u_l,$ $\{j,k,l\}=\{1,2,3,4\}\setminus\{i\},$ i.e.,
$$\mathbf T_i=\{c_j\mathbf u_j+c_k\mathbf u_k+c_l\mathbf u_l\ |\ c_j,c_k,c_l\in\mathbb R\}\cap\mathbb A\mathrm d\mathbb S^3.$$
Then $\langle \mathbf u,\mathbf v_i\rangle =0$ for any vector $\mathbf u\in \mathbf T_i,$ and $\mathbf T_i= \Pi_{\mathbf v_i}.$ 

\begin{definition}[Truncated hyperideal tetrahedra in $\ads^3$]\label{truncated} A truncated hyperideal tetrahedron in $\ads^3$ is the subset of $\ads^3$ bounded by the four faces and the four  planes of truncation  of a hyperideal tetrahedron in $\mathbb A\mathrm d\mathbb S^3.$
\end{definition}

\begin{remark}\label{lift} In the projective model, a \emph{truncated hyperideal tetrahedron} is  the projection of a truncated hyperideal tetrahedron in $\mathbb A\mathrm d\mathbb S^3,$ which 
is a polyhedron in $\mathrm{AdS}^3$ bounded by the projections  of the faces and the projections of the planes of truncation at the vertices. A truncated hyperideal tetrahedron in the projective model has two lifts in the Lorentzian model, differing by the antipode map of $\mathbb E^{2,2}$.
\end{remark}

Let $\Delta$ be a truncated hyperideal tetrahedron in $\mathbb A\mathrm d\mathbb S^3.$  We still call the vectors $\mathbf v_1,\dots,\mathbf v_4$ the \emph{vertices} of $\Delta,$ and call the vectors  $\mathbf u_1,\dots,\mathbf u_4$ the \emph{outward unit normal vectors} of $\Delta.$  For each $i\in\{1,2,3,4\},$ we call
$$F_i\doteq \mathbf F_i \cap \Delta$$
the \emph{face} of $\Delta$ opposite to $\mathbf v_i,$  and call
$$T_i\doteq \mathbf T_i \cap \Delta$$
the \emph{triangle of truncation} at $\mathbf v_i.$ As all $\mathbf v_i$'s are in $\mathbb A\mathrm d\mathbb S^{3^*},$ we have the following proposition.

\begin{proposition}
All the four triangles of truncation of a truncated  hyperideal tetrahedron in $\ads^3$ are time-like.
\end{proposition}

 In Proposition \ref{space-like}, we will show that all the four faces of $\Delta$ are space-like. See also \cite[Lemma 2.1]{CS}.

\subsection{Edge lengths, Gram matrix, and a criterion}

\begin{definition}[Edge lengths]\label{edgelength}
Let $\Delta$ be a truncated hyperideal tetrahedron in  $\ads^3$ with vertices $\mathbf v_1,\dots,\mathbf v_4.$ For $\{i,j\}\subset\{1,2,3,4\},$ the \emph{edge} $e_{ij}$ of $\Delta$ is the intersection of $\Delta$ with the geodesic in $\ads^3$ connecting $\mathbf v_i$ and $\mathbf v_j,$ and the \emph{edge length} $l_{ij}$ is the geodesic length of $e_{ij}.$ 
\end{definition}
Observe that $e_{ij}$ is the shortest geodesic connecting the planes of truncation $\mathbf T_i$ and $\mathbf T_j,$  and $l_{ij}$ is then the distance between  $\mathbf T_i$ and $\mathbf T_j.$ As such, we have
\begin{equation}\label{cosh}
\langle \mathbf v_i,\mathbf v_j \rangle = -\cosh l_{ij}.
\end{equation}

\begin{definition}[Gram matrix in the edge lengths]\label{Gramel}
For a truncated hyperideal tetrahedron $\Delta$ in  $\ads^3$ with the edge lengths  $(l_{12},\dots,l_{34}),$  the \emph{Gram matrix} of  $\Delta$ \emph{in the edge lengths}, denoted by $\mathrm{Gram}(\Delta)$,  is defined by 
\begin{equation*}
\begin{bmatrix}
1 & -\cosh l_{12} & -\cosh l_{13}
    & -\cosh l_{14} \\
-\cosh l_{12} & 1 & -\cosh l_{23}
    & -\cosh l_{24}\\
    -\cosh l_{13} & -\cosh l_{23} & 1
    & -\cosh l_{34} \\
 -\cosh l_{14} & -\cosh l_{24} & -\cosh l_{34}
    & 1
  \end{bmatrix}.
\end{equation*}
More generally, for a six-tuple $(l_{12},\dots,l_{34})$ in $\mathbb R^6_{>0},$ we call the above matrix the \emph{Gram matrix} of the six-tuple, and denote it by $\mathrm{Gram}(l_{12},\dots,l_{34}).$
\end{definition}

In the rest of the paper, we will also use the labeling $l_1,\dots,l_6$ for the edge lengths with the identification $(l_{12},l_{13},l_{14},l_{23},l_{24},l_{34})=(l_1,l_2,l_6,l_3,l_5,l_4)$ as depicted in Figure \ref{convention1}, in which way the Gram matrix becomes 
\begin{equation*}
\begin{bmatrix}
1 & -\cosh l_{1} & -\cosh l_{2}
    & -\cosh l_{6} \\
-\cosh l_{1} & 1 & -\cosh l_{3}
    & -\cosh l_{5}\\
    -\cosh l_{2} & -\cosh l_{3} & 1
    & -\cosh l_{4} \\
 -\cosh l_{6} & -\cosh l_{5} & -\cosh l_{4}
    & 1
  \end{bmatrix}.
\end{equation*}

\begin{figure}[htbp]
\centering
\includegraphics[scale=0.3]{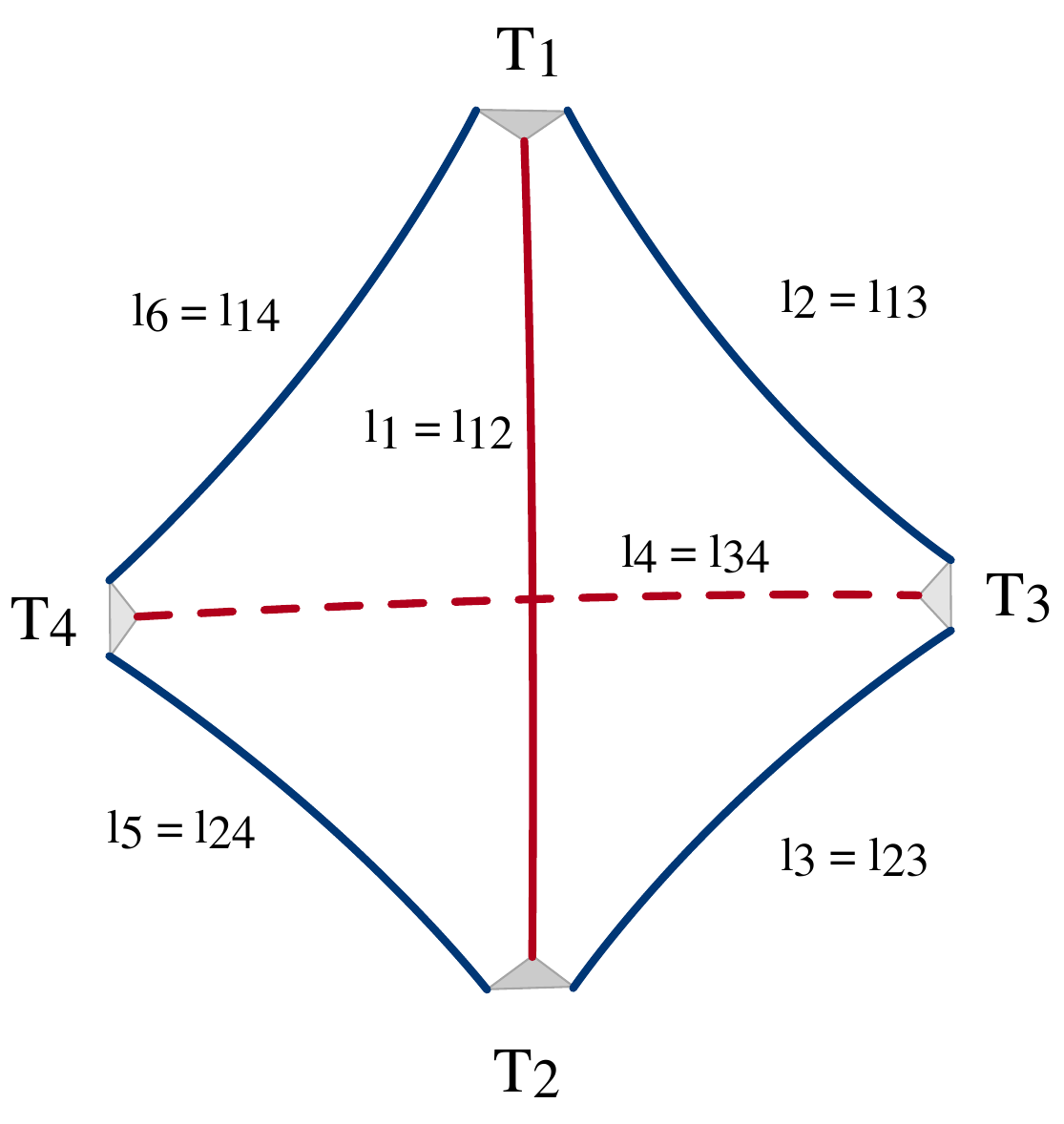}
\caption{Labeling for the edge lengths: in which way the edge lengths appearing in  row $i$ and in  column $i$ of the Gram matrix are adjacent to the triangle of truncation $T_i$ of $\Delta$.}
\label{convention1}
\end{figure}

\begin{theorem}\label{criterion}  For $(l_{1},\dots, l_{6})\in \mathbb R^6_{>0},$ the following are equivalent:
\begin{enumerate}[(1)]
\item $(l_{1},\dots, l_{6})$ are the edge lengths of a truncated hyperideal tetrahedron in $\mathbb A\mathrm d\mathbb S^3.$ 
\item The Gram matrix of $(l_{1},\dots, l_{6})$ has signature $(2,2).$ 
\end{enumerate}
\end{theorem}

\begin{proof}

Suppose (1) holds, and $\Delta$ is the truncated hyperideal tetrahedron in $\ads^3$ with edge lengths  $(l_{12},\dots, l_{34}).$ Let $G$ be the Gram matrix of $(l_{12},\dots, l_{34}).$ Then $G$ is the Gram matrix of $\Delta$ in the edge lengths. Let $\mathbf v_1, \dots, \mathbf v_4\in \mathbb A\mathrm d\mathbb S^{3^*}$ be the vertices of $\Delta,$ and let 
$$\mathrm V=[\mathbf v_1,\mathbf v_2,\mathbf v_3,\mathbf v_4]$$
be the $4\times 4$ matrix with $\mathbf v_1,\dots,\mathbf v_4$ the columns. Then by (\ref{cosh}),
$$G=\mathrm V^T\cdot \mathrm I_{2,2} \cdot \mathrm V,$$
where 
$$\mathrm I_{2,2}=\left[\begin{matrix}
1& 0& 0 & 0\\
0 & 1&0 &0\\
0& 0& -1&0 \\
0&0 &0 &  -1\\
 \end{matrix}\right].$$
Therefore, $G$ has signature $(2,2).$ 
\\

Suppose (2) holds, and $G$  is the Gram matrix of $(l_{12},\dots, l_{34}).$ Then by Sylvester's Law of Inertia,
$$G=\mathrm V^T\cdot \mathrm I_{2,2} \cdot \mathrm V$$
for an invertible matrix $\mathrm V,$ uniquely determined up to an action of $\mathrm O(2,2).$ Let $\mathbf v_1,\dots,\mathbf v_4$ be the columns of $\mathrm V,$ considered as vectors in $\mathbb E^{2,2}.$ We will show that $\mathbf v_1,\dots,\mathbf v_4$ are the vertices of a hyperideal tetrahedron in $\mathbb A\mathrm d\mathbb S^3,$ by verifying the following:
\begin{enumerate}[(a)]

\item   $\mathbf v_1,\dots,\mathbf v_4$ are linearly independent. 

\item For $i\in\{1,2,3,4\},$ $\mathbf v_i\in \mathbb A\mathrm d\mathbb S^{3^*}.$ 

\item  For $\{i,j\}\subset\{1,\dots, 4\},$ $\mathrm L_{ij}\cap \mathbb B^{2,2}\neq \emptyset,$ where
$L_{ij}$ is the line segment connecting $\mathbf v_i$ and $\mathbf v_j,$ and $\mathbb B^{2,2}=\{\mathbf v\in\mathbb E^{2,2}\ |\ \langle \mathbf v,\mathbf v \rangle <0 \}.$
\end{enumerate}

From the fact that $\mathrm V$ is invertible, (a) follows.

To see (b), we have for each $i\in \{1,2,3,4\}$ that $\langle \mathbf v_i,\mathbf v_i \rangle$ equals the $ii$-th entry of $G,$ which equals $1.$ Hence $\mathbf v_i\in \mathbb A\mathrm d\mathbb S^{3^*}.$ 

To see (c), we have 
$$\mathrm L_{ij} = \{ t\mathbf v_i + (1-t) \mathbf v_j\ |\ t\in [0,1]\},$$
and $\mathrm L_{ij}\cap \mathbb B^{2,2}\neq \emptyset$ if and only if the following inequality 
$$\langle  t\mathbf v_i + (1-t) \mathbf v_j,  t\mathbf v_i + (1-t) \mathbf  v_j \rangle <0 $$
in $t$ has a solution in $[0,1].$ As the left hand side of the inequality equals the following quadratic function
$$2(1-\langle \mathbf v_i,\mathbf v_j \rangle)t^2-2(1-\langle \mathbf v_i,\mathbf v_j \rangle)t+1$$
in $t,$ and $\langle \mathbf v_i,\mathbf v_j \rangle=-\cosh l_{ij} <-1,$  the inequality always has a solution in $(0,1).$ This completes the verification of (c).\end{proof}

As a consequence of Theorem \ref{criterion}, we have the following classification.

\begin{theorem}\label{classification}
Every $(l_1,\dots, l_6)\in \mathbb R^6_{>0}$ is exclusively the six-tuple of the edge lengths of:
\begin{enumerate}[(1)]
\item a truncated hyperideal tetrahedron in $\mathbb H^3,$ 
\item a (truncated hyperideal) flat tetrahedron, or
\item a truncated hyperideal tetrahedron in  $\mathbb A\mathrm d\mathbb S^3.$  
\end{enumerate}
\end{theorem}

\begin{proof}Let $G$ be the Gram matrix of $(l_{12},\dots, l_{34}).$ Then by \cite[Proposition 2.7]{LMSWY}, if $\det G<0,$ then $l_{12},\dots, l_{34}$ are the edge lengths of a a truncated hyperideal tetrahedron in $\mathbb H^3,$ and if $\det G=0,$ then $l_{12},\dots, l_{34}$ are the edge lengths of a flat tetrahedron; and by Lemma \ref{22} below and Theorem \ref{criterion}, if $\det G>0,$ then $l_{12},\dots, l_{34}$ are the edge lengths of a  truncated hyperideal tetrahedron in  $\mathbb A\mathrm d\mathbb S^3.$  
\end{proof}

\begin{lemma}\label{22} Let $G$ be the Gram matrix of $(l_{12},\dots, l_{34})\in \mathbb R^6_{>0}.$ If $\det G >0,$ then $G$ has signature $(2,2).$ 
\end{lemma}

\begin{proof} Since $\det G>0,$ the signature of $G$ has the possibilities $(4,0),$ $(2,2)$ and $(0,4).$ Since the first diagonal entry of $G$ equals $1,$ which is positive, $G$ cannot be negative definite and the signature of $G$ cannot be $(0,4);$ and since  $\det \left[\begin{matrix}
1& -\cosh l_{12} \\
-\cosh l_{12} & 1\\
 \end{matrix}\right] = 1 -\cosh^2l_{12}<0,$ $G$ cannot be positive definite and the signature of $G$ cannot be $(4,0).$ Therefore, the signature of $G$ is $(2,2).$
\end{proof}

The next result will be needed to define the dihedral angles of a hyperideal tetrahedron in $\ads^3$.

\begin{proposition}\label{cos} Let  $(l_{12},\dots, l_{34})\in \mathbb R^6_{>0}$ be the edge lengths of a  truncated hyperideal tetrahedron  in  $\mathbb A\mathrm d\mathbb S^3.$ For $i\in\{1,2,3,4\},$ let $\mathbf u_i$ be the outward unit normal vector of the face opposite to the vertex $\mathbf v_i.$ For $i,j\in\{1,\dots, 4\},$ let $G_{ij}$ be the $ij$-th cofactor of the Gram matrix of $(l_{12},\dots, l_{34}).$ Then 
$$\langle \mathbf u_i,\mathbf u_j \rangle = \frac{G_{ij}}{\sqrt{G_{ii}G_{jj}}}.$$
\end{proposition}

\begin{proof} Let $\Delta$ be the  truncated hyperideal tetrahedron  in  $\ads^3$ with edge lengths $(l_{12},\dots, l_{34})$ and vertices $\mathbf v_1,\dots,\mathbf v_4.$ 

For $i\in\{1,2,3,4\},$ let 
$$\mathbf w_i=\sum_{j=1}^4 G_{ij}\mathbf v_j.$$
The we have
$$\langle \mathbf w_i, \mathbf v_j \rangle = \delta_{ij}\det G,$$
where $\delta_{ij}$ is the Kronecker symbol; and as a consequence, 
$$\langle \mathbf w_i, \mathbf w_j \rangle = G_{ij}\det G.$$
Here we recall  for each $i\in\{1,2,3,4\}$ that $\det G_{ii} = 1 - \cosh^2l_j-\cosh^2l_k-\cosh^2l_l -2\cosh l_j\cosh l_k\cosh l_l <0,$ where $\{i,j,k,l\}=\{1,2,3,4\};$  and $\det G>0.$ For $i\in\{1,2,3,4\},$ we define 
$$\mathbf u_i =  \frac{\mathbf w_i}{\sqrt{-G_{ii}\det G}}.$$
Then we have 
$$\langle \mathbf u_i,\mathbf u_j \rangle =  \frac{\langle \mathbf w_i,\mathbf w_j \rangle}{\sqrt{G_{ii}G_{jj}}\det G} = \frac{G_{ij}}{\sqrt{G_{ii}G_{jj}}}.$$

We are left to show that $\mathbf u_1,\dots,\mathbf u_4$ are the outward unit normal vectors of $\Delta,$ and we have the following verifications:
\begin{enumerate}[(a)]
\item For $i\in\{1,2,3,4\},$ $$\langle \mathbf u_i,\mathbf v_i \rangle  =   \frac{\langle \mathbf w_i, \mathbf v_i\rangle }{\sqrt{-G_{ii}\det G}}=  \frac{\det G}{\sqrt{-G_{ii}\det G}}= \sqrt{\frac{\det G}{-G_{ii}}}>0,$$
hence $\mathbf u_i$ is outward. 

\item For  $i\in\{1,2,3,4\},$ 
\begin{equation}\label{uu1}
    \langle \mathbf u_i,\mathbf u_i \rangle = \frac{\langle \mathbf w_i,\mathbf w_i \rangle }{-G_{ii}\det G}=\frac{G_{ii}\det G}{-G_{ii}\det G}= -1,
\end{equation}
hence $\mathbf u_i$ is unit. 

\item For $\{i,j\}\subset \{1,2,3,4\},$ $$\langle \mathbf u_i,\mathbf v_j \rangle = \frac{\langle \mathbf w_i, \mathbf v_j\rangle }{\sqrt{-G_{ii}\det G}} =  \frac{\delta_{ij}\det G}{\sqrt{-G_{ii}\det G}}= 0,$$
hence $\mathbf u_i$ is normal. 
\end{enumerate}
This completes the proof.
\end{proof}

As an immediate consequence of (\ref{uu1}), we have the following proposition.

\begin{proposition}\label{space-like}
All the four faces of a truncated hyperideal tetrahedron in $\ads^3$
 are space-like.
\end{proposition}

 \begin{remark}
 The same result was first proved in \cite[Lemma 2.1]{CS} for hyperideal polyhedra in $\mathbb A\mathrm d\mathbb S^3.$
 \end{remark}

\subsection{Dihedral angles, Gram matrix, and another criterion}

Let $\Delta$ be a truncated hyperideal tetrahedron  in  $\ads^3$ with  vertices $\mathbf v_1,\dots,\mathbf v_4$ and the corresponding outward unit normal vectors of the faces $\mathbf u_1,\dots, \mathbf u_4.$  For $\{i,j\}\subset \{1,2,3,4\},$ the dihedral angle $\theta_{ij}$ at the edge $e_{ij}$ should satisfy 
\begin{equation}\label{cos2}
 {\langle \mathbf u_k, \mathbf u_l \rangle = \cos \theta_{ij}},
\end{equation}
where $\{k,l\}=\{1,2,3,4\}\setminus\{i,j\}.$  {The idea behind this requirement is that the dihedral angle between two faces should equal $\pi$ minus the angle between the corresponding normal vectors, and the angle between a vector $\mathbf u$ in $\ads^3$ with itself should be defined to be $0.$ }

 As an immediate consequence of  Proposition \ref{cos} and \cite[Proposition 4.5]{LY} (see also \cite[Proposition 2.6]{LMSWY}), we have the following corollary.

\begin{corollary}\label{>1} Let $\Delta$ be a truncated hyperideal tetrahedron  in  $\ads^3$ with outward unit normal vectors $\mathbf u_1,\dots,\mathbf u_4.$ Then for $\{i,j\}\subset\{1,2,3,4\},$ either $\langle \mathbf u_i,\mathbf u_j\rangle <-1$ or $\langle \mathbf u_i,\mathbf u_j\rangle >1.$
\end{corollary}

 Guaranteed by Corollary \ref{>1}, we make the following definition. 
 
\begin{definition}[Dihedral angles]\label{dihedralangle} The \emph{dihedral angle} $\theta_{ij}$ at the edge $e_{ij}$ is defined as follows.  If $\langle \mathbf u_k, \mathbf u_l \rangle<-1,$ then we define 
$$ {\theta_{ij}\doteq \pi- \mathbf i \cosh^{-1} \big(-\langle \mathbf u_k, \mathbf u_l \rangle\big)};$$
and if $\langle \mathbf u_k, \mathbf u_l \rangle>1,$ then we define 
$$ {\theta_{ij}\doteq \mathbf i \cosh^{-1} \big(\langle \mathbf u_k, \mathbf u_l \rangle\big)},$$
where $\{i,j,k,l\}=\{1,2,3,4\}.$
\end{definition}

In the rest of this paper, we let 
\begin{equation}\label{A}
    \mathbb A=\mathrm A_1\cup \mathrm A_2\cup \mathrm A_3,
\end{equation}
where
$$\mathrm A_1=(\pi-\mathbf i\mathbb R_{>0}) \times \mathbf i\mathbb R_{>0} \times  \mathbf i\mathbb R_{>0}  \times \mathbf i\mathbb R_{>0} \times  \mathbf i\mathbb R_{>0} \times (\pi-\mathbf i\mathbb R_{>0}) ,$$
$$\mathrm A_2= \mathbf i\mathbb R_{>0} \times (\pi-\mathbf i\mathbb R_{>0}) \times \mathbf i\mathbb R_{>0} \times  \mathbf i\mathbb R_{>0} \times (\pi-\mathbf i\mathbb R_{>0}) \times \mathbf i\mathbb R_{>0},$$
$$\mathrm A_3=\mathbf i\mathbb R_{>0} \times  \mathbf i\mathbb R_{>0} \times (\pi-\mathbf i\mathbb R_{>0}) \times (\pi-\mathbf i\mathbb R_{>0}) \times \mathbf i\mathbb R_{>0} \times  \mathbf i\mathbb R_{>0}.$$
Then by \cite[Proposition 4.5]{LY} and \cite[Proposition 2.6]{LMSWY}, the set of all the possible dihedral angles of a truncated hyperideal tetrahedra in $\ads^3$ is a subset of $\mathbb A.$

\begin{remark}\label{extcov}
We observe that, if $(l_{12},\dots,l_{34})$ are the edge lengths of a truncated  hyperideal tetrahedron in $\ads^3,$ then the extended co-volume function defined in \cite{LY} can be computed by
$$\widetilde{\mathrm{Cov}}(l_{12},\dots,l_{34})=\frac{\pi}{2}\cdot(l_{ij}+l_{kl}),$$
where $\{i,j,k,l\}=\{1,2,3,4\}$ with  $\cos\theta_{ij}<-1$ and $\cos\theta_{kl}<-1,$ or equivalently, with $\mathrm{Re}\theta_{ij}=\mathrm{Re}\theta_{kl}=\pi.$
\end{remark}

\begin{remark}
The imaginary part of $\theta_{ij}$ coincides with the negative of the exterior dihedral angle defined in \cite{CS}. 
\end{remark}

\begin{remark}\label{rm} It was proved as a special case of  \cite[Theorem 1.3]{CS} that
a six-tuple $(\theta_{12},\dots,\theta_{34})$ in $\mathbb A$ is the six-tuple of dihedral angles of a hyperideal tetrahedron in $\ads^3$, with $\mathrm{Re}\theta_{ij}=\mathrm{Re}\theta_{kl}=\pi,$ if and only if:
\begin{enumerate}[(1)]
\item For each $i\in\{1,2,3,4\}$ and $\{j,k,l\}=\{1,2,3,4\}\setminus\{i\},$
$$\mathrm{Im}\theta_{ij}+\mathrm{Im}\theta_{ik}+\mathrm{Im}\theta_{il}<0.$$

\item For $\{i,j,k,l\}=\{1,2,3,4\}$, $$\mathrm{Im}\theta_{ij}+\mathrm{Im}\theta_{ik}+\mathrm{Im}\theta_{jl}+\mathrm{Im}\theta_{kl}<0.$$
\item For $k\in\{1,2,3,4\}\setminus \{i,j\},$
$$\mathrm{Im}\theta_{ij}+\mathrm{Im}\theta_{ik}<0;$$
 and for $i\in\{1,2,3,4\}\setminus\{k,l\},$
$$\mathrm{Im}\theta_{ik}+\mathrm{Im}\theta_{kl}<0.$$
\end{enumerate}
It is easy to see that Condition (1) implies Conditions (2) and (3).
\end{remark}

In Theorem \ref{criterion2} below, we provide another criterion of the dihedral angles of a truncated hyperideal tetrahedron in $\mathbb A\mathrm d\mathbb S^3,$ in terms of the Gram matrix, which will be needed in the proof of Theorem \ref{vol2}. 

\begin{definition}[Gram matrix in the dihedral angles]\label{Gramda}
For a truncated hyperideal tetrahedron $\Delta$ in  $\ads^3$ with the dehidral angles  $(\theta_{12},\dots,\theta_{34}),$  the \emph{Gram matrix} of  $\Delta$ \emph{in the dihedral angles}, denoted by $\mathrm{Gram}(\Delta)$,  is defined by 
\begin{equation*}
\begin{bmatrix}
-1 & \cos \theta_{34} & \cos \theta_{24}
    & \cos \theta_{23} \\
\cos \theta_{34} & -1 & \cos \theta_{14}
    & \cos \theta_{13}\\
    \cos \theta_{24} & \cos \theta_{14} & -1
    & \cos \theta_{12} \\
 \cos \theta_{23} & \cos \theta_{13} & \cos \theta_{12}
    & -1
  \end{bmatrix}.
\end{equation*}
For a six-tuple $(\theta_{12},\dots,\theta_{34})$ in $\mathbb A,$ we call the above matrix the \emph{Gram matrix} of the six-tuple, and denote it by $\mathrm{Gram}(\theta_{12},\dots,\theta_{34}).$
\end{definition}

In the rest of the paper, we will also use the labeling $\theta_1,\dots,\theta_6$ for the dihedral angles with the identification $(\theta_{12},\theta_{13},\theta_{14},\theta_{23},\theta_{24},\theta_{34})=(\theta_1,\theta_2,\theta_6,\theta_3,\theta_5,\theta_4)$ as depicted in Figure \ref{convention2}, in which way the Gram matrix becomes 
\begin{equation*}
\begin{bmatrix}
-1 & \cos \theta_{4} & \cos \theta_{5}
    & \cos \theta_{3} \\
\cos \theta_{4} & -1 & \cos \theta_{6}
    & \cos \theta_{2}\\
    \cos \theta_{5} & \cos \theta_{6} & -1
    & \cos \theta_{1} \\
 \cos \theta_{3} & \cos \theta_{2} & \cos \theta_{1}
    & -1
  \end{bmatrix}.
\end{equation*}
In the rest of the paper, when we say ``a six-tuple $(\theta_1,\dots,\theta_6)$ in $\mathbb A$'', we mean that the corresponding six-tuple $(\theta_{12},\theta_{13},\theta_{14},\theta_{23},\theta_{24},\theta_{34})$ is in $\mathbb A.$

\begin{figure}[htbp]
\centering
\includegraphics[scale=0.3]{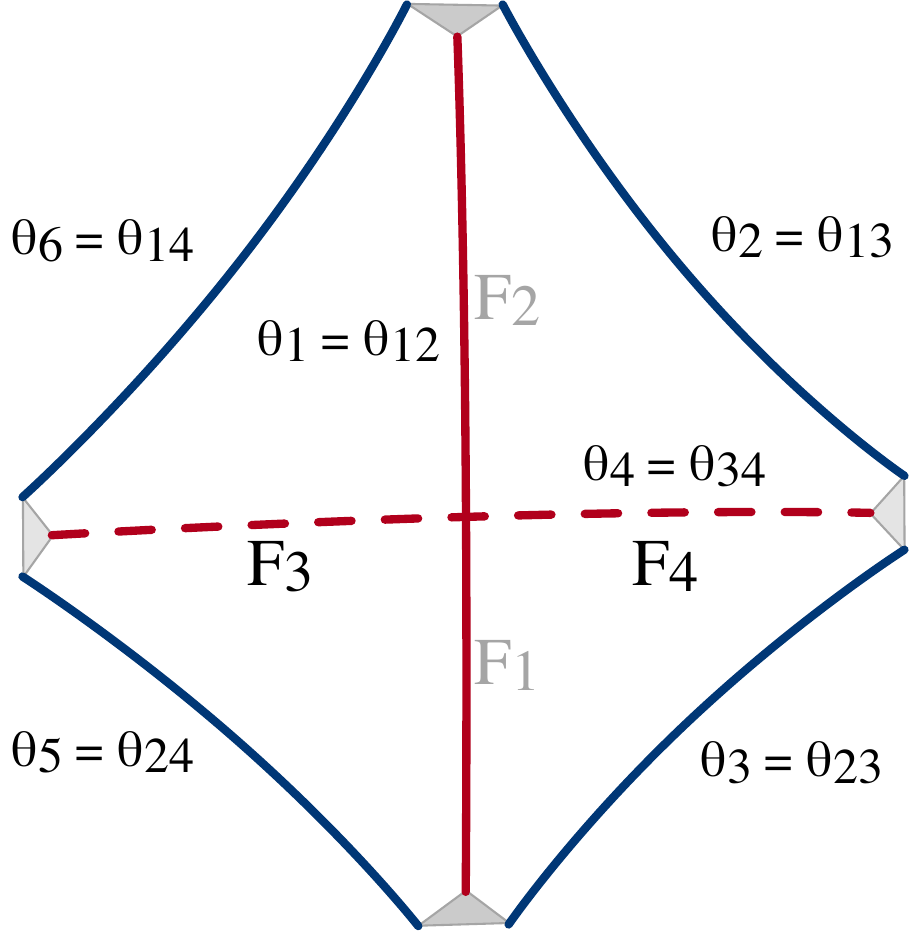}
\caption{Labeling for the dihedral angles: in which way the dihedral angles appearing in row $i$ and in column $i$ of the Gram matrix are adjacent to the face $F_i$ of $\Delta$.}
\label{convention2}
\end{figure}

\begin{theorem}\label{criterion2} For a six-tuple $(\theta_1,\dots,\theta_6)$ in $\mathbb A,$ the following are equivalent:
\begin{enumerate}[(1)]
\item $(\theta_1,\dots,\theta_6)$ are the dihedral angles of a truncated hyperideal tetrahedron in $\mathbb A\mathrm d\mathbb S^3.$ 
\item The Gram matrix $G$ of $(\theta_1,\dots,\theta_6)$ satisfies the following three conditions:
\begin{enumerate}[(a)]
\item $G$ has signature $(2,2).$
\item For each $i\in\{1,2,3,4\},$ $G_{ii}>0.$
\item For each $\{i,j\}\subset\{1,2,3,4\},$ $G_{ij}<0.$
\end{enumerate}
\end{enumerate}
\end{theorem}

\begin{remark}\label{rm2} By a direct computation, it is easy to check that Condition (2)(b) implies condition (1) of \cite[Theorem 1.3]{CS} mentioned in Remark \ref{rm}. It can also be seen that Condition (2)(b) implies Conditions (2)(a) and (2)(c), but we will not use this fact in this paper. 
\end{remark}

\begin{proof}[Proof of Theorem \ref{criterion2}]  The proof is similar to that of Theorem \ref{criterion}. 
\medskip

Namely, suppose (1) holds that $\theta_1,\dots,\theta_6$ are the dihedral angles of a truncated hyperideal tetrahedron $\Delta$ in $\mathbb A\mathrm d\mathrm S^3$. Then the Gram matrix $G$ of $(\theta_1,\dots,\theta_6)$ is the Gram matrix of $\Delta$ in the dihedral angles. Let $\mathbf u_1,\dots,\mathbf u_4$ be the outward unit normal vectors of $\Delta,$ and let 
$$\mathrm U=[\mathbf u_1,\mathbf u_2,\mathbf u_3,\mathbf u_4]$$
be the $4\times 4$ matrix with $\mathbf u_i$'s the columns. Then 
$$G=\mathrm U^T\cdot \mathrm I_{2,2}\cdot \mathrm U.$$
Hence $G$ has signature $(2,2),$ and (a) of (2) holds, provided that $\mathrm U$ is invertible. To see the invertibility of $\mathrm U$, we let $\mathbf v_1,\dots,\mathbf v_4$ be the vertices of $\Delta$ considered as vectors of $\mathbb E^{2,2}.$ Suppose a linear combination $\sum_{k=1}^4a_k\mathbf u_k=0$ for some $a_1,\dots,a_4\in\mathbb R.$ Then for each $i\in\{1,2,3,4\},$ we have 
$$a_i \langle \mathbf u_i, \mathbf v_i\rangle=\Big\langle \sum_{k=1}^4a_k\mathbf u_k, \mathbf v_i\Big\rangle = 0;$$ and since $\langle \mathbf u_i, \mathbf v_i\rangle>0,$ we have  $a_i=0.$ This implies that  $\mathbf u_1,\dots,\mathbf u_4$ are linearly independent, and $\mathrm U$ is invertible. 

\smallskip

To see (b) of (2), as a consequence of (a), we have
\begin{equation}\label{det>0}
\det G>0.
\end{equation}
Let 
$$\mathbf w_i=\sum_{j=1}^4 G_{ij}\mathbf u_j.$$
Then we have
\begin{equation}\label{wu}
\langle \mathbf w_i,\mathbf u_j\rangle =\delta_{ij}\det G,
\end{equation}
and as a consequence,
\begin{equation}\label{ww}
\langle \mathbf w_i,\mathbf w_j\rangle =G_{ij}\det G.
\end{equation}
Let $\mathbf v_1,\dots,\mathbf v_4$ be the vertices of $\Delta$  considered as vectors of $\mathbb E^{2,2}.$ Then for each $i\in\{1,2,3,4\},$ (\ref{wu}) implies that $\mathbf w_i$ is a scalar multiple of $\mathbf v_i.$ Since $\mathbf v_i\in\mathbb A\mathrm d\mathbb S^{3^*},$ we have
\begin{equation}\label{www}
\langle \mathbf w_i,\mathbf w_i \rangle>0.
\end{equation}
As a consequence of (\ref{det>0}),(\ref{ww}) and (\ref{www}), we have
$$G_{ii}=\frac{\langle \mathbf w_i,\mathbf w_i\rangle}{\det G}>0.$$
This proves (b) of (2).
\smallskip

To see (c) of (2), for  $i\in\{1,2,3,4\},$ let 
\begin{equation}\label{v}
\mathbf v_i' =\frac{\mathbf w_i}{\sqrt{G_{ii}\det G}}. 
\end{equation}
We first claim that 
\begin{equation}\label{v=v}
\mathbf v_i' = \mathbf v_i,
\end{equation}
i.e., $\mathbf v_1',\dots,\mathbf v_4'$ are the vertices of $\Delta.$
Indeed, by (\ref{ww}), 
\begin{equation}\label{vvv}
\langle \mathbf v_i', \mathbf v_i' \rangle = \frac{\langle \mathbf w_i, \mathbf w_i \rangle }{G_{ii}\det G}=1,
\end{equation}
hence $\mathbf v_i' \in \mathbb A\mathrm d\mathbb S^{3^*}$; and by (\ref{wu}), we have for each $i\in \{1,2,3,4\}$ that 
$$\langle \mathbf v_i', \mathbf u_i \rangle = \frac{\langle \mathbf w_i, \mathbf u_i \rangle }{\sqrt{G_{ii}\det G}}=\sqrt{\frac{\det G}{G_{ii}}}>0,$$
and for each $\{i,j\}\subset\{1,2,3,4\}$ that 
$$\langle \mathbf v_i', \mathbf u_j \rangle = \frac{\langle \mathbf w_i, \mathbf u_j \rangle }{\sqrt{G_{ii}\det G}}=0.$$
This implies that $\mathbf u_1,\dots, \mathbf u_4$ are the outward unit normal vectors of the tetrahedron $\Delta'$ in $\mathbb A\mathbf d\mathbb S^3$ with vertices $\mathbf v_1',\dots,\mathbf v_4'.$ As $\mathbf u_1,\dots, \mathbf u_4$ are also the outward unit normal vectors of $\Delta,$ we have $\Delta'=\Delta,$ and $\mathbf v_1',\dots,\mathbf v_4'$ are respectively the vertices $\mathbf v_1,\dots,\mathbf v_4$ of $\Delta.$

Next, as the line segment $L_{ij}=\{t\mathbf v_i + (1-t)\mathbf v_j\ |\ t\in [0,1]\}$ connecting $\mathbf v_i$ and $\mathbf v_j$ intersects $\mathbb B^{2,2}=\{\mathbf v \in \mathbb E^{2,2}\ |\ \langle \mathbf v, \mathbf v\rangle<0\},$ we have that the quadratic inequality 
$$\langle t\mathbf v_i + (1-t)\mathbf v_j, t\mathbf v_i + (1-t)\mathbf v_j \rangle  = 2(1-\langle \mathbf v_i,\mathbf v_j \rangle)t^2-2(1-\langle \mathbf v_i,\mathbf v_j \rangle)t+1<0$$
in $t$ has a solution in $[0,1],$ where the equality comes from (\ref{vvv}).  This implies that 
\begin{equation}\label{vv}
\langle \mathbf v_i, \mathbf v_j \rangle <-1. 
\end{equation}
By (\ref{v}), (\ref{v=v}) and (\ref{ww}, we  have
\begin{equation}\label{inner}
\langle \mathbf v_i,\mathbf v_j\rangle = \frac{\langle \mathbf w_i,\mathbf w_j\rangle }{\sqrt{G_{ii}\det G}\sqrt{G_{jj}\det G}}=\frac{G_{ij}}{\sqrt{G_{ii}G_{jj}}},
\end{equation}
which, together with (\ref{vv}), implies that 
$$G_{ij}<0.$$
This proves (c) of (2).
\medskip

Now, suppose (2) holds. Then by (a), $G$ has signature $(2,2),$ and hence
\begin{equation}\label{det>02}
\det G>0.
\end{equation}
By Sylvester's Law of Inertia,
$$G=\mathrm U^T\cdot \mathrm I_{2,2}\cdot \mathrm U$$
for an invertible matrix $\mathrm U,$ uniquely determined up to an action of  elements of $O(2,2).$  Let $\mathbf u_1,\dots, \mathbf u_4$ be the columns of $\mathrm U$ considered as vectors in $\mathbb E^{2,2}$ and let
$$\mathbf w_i=\sum_{j=1}^4G_{ij}\mathbf u_j.$$
Then we have
\begin{equation}\label{wu2}
\langle \mathbf w_i,\mathbf u_j\rangle =\delta_{ij}\det G,
\end{equation}
and as a consequence 
\begin{equation}\label{ww2}
\langle \mathbf w_i,\mathbf w_j\rangle =G_{ij}\det G.
\end{equation}
By (b), we have $G_{ii}\neq 0,$ hence we can define 
\begin{equation}\label{v2}
\mathbf v_i= \frac{\mathbf w_i}{\sqrt{G_{ii}\det G}}.
\end{equation}

Next we show that 
\begin{enumerate}[(i)]
\item $\mathbf v_1,\dots, \mathbf v_4$ are the vertices of a hyperideal tetrahedron $\Delta$ in $\mathbb A\mathrm d\mathbb S^3.$ 
\item $\mathbf u_1,\dots,\mathbf u_4$ are the outward unit normal vectors of $\Delta$, so that $G$ is the Gram matrix of $\Delta$ in the dihedral angles.  
\end{enumerate}
From (i) and  (ii), it follows that $(\theta_{12},\dots,\theta_{34})$ is the six-tuple of dihedral angles of $\Delta,$ and (1) holds.
\smallskip

To prove (i), we need following three steps. 
\begin{enumerate}[\text{Step} 1.]
\item We show that $\mathbf v_1,\dots, \mathbf v_4$ are linearly independent. Indeed, let 
$$\mathrm W=[\mathbf w_1,\mathbf w_2,\mathbf w_3,\mathbf w_4]$$
be the $4\times 4$ matrix containing $\mathbf w_i$'s as the columns. Then
$$\mathrm W=\mathrm U\cdot \mathrm{Ad}(G),$$
where $\mathrm{Ad}(G)$ is the adjugate matrix of $G.$ By Cramer's rule,
$$\mathrm{Ad}(G)=\det G\cdot G^{-1}.$$
Therefore, $\det \mathrm{Ad}(G)\neq 0,$ and as a consequence, 
$$\det\mathrm W=\det \mathrm U\cdot \det \mathrm{Ad}(G)\neq 0, $$
and the columns $\mathbf w_1,\dots, \mathbf w_4$ of $W$ are linearly independent. Since $\mathbf v_1,\dots,\mathbf v_4$ are non-zero scalar multiples of $\mathbf w_1,\dots, \mathbf w_4,$ they are linearly independent.

\item  Since by (b) that $G_{ii}>0$ for each $i\in \{1,2,3,4\},$ by (\ref{v2}) and (\ref{ww2}), we have
\begin{equation}\label{vv2}
\langle \mathbf v_i,\mathbf v_i\rangle =   \frac{\langle \mathbf w_i, \mathbf w_i\rangle}{G_{ii}\det G}=1.
\end{equation}
This shows that $\mathbf v_1,\dots, \mathbf v_4$ are in $\mathbb A\mathrm d\mathbb S^{3^*}.$ 

\item We show that for each $\{i,j\}\subset\{1,2,3,4\},$ the line segment  $L_{ij}$ connecting $\mathbf v_i$ and $\mathbf v_j$ intersects $\mathbb B^{2,2}.$ To see this, by Jacobi's Theorem (see \cite[2.5.1. Theorem]{P}) and (\ref{det>02}),
we have
$$G_{ij}^2-G_{ii}G_{jj}=(\cos^2\theta_{kl}-1)\det G=\big(\cosh^2(\mathrm{Im}\theta_{kl})-1\big)\det G>0,$$
where $\{i,j,k,l\}=\{1,2,3,4\}.$
Then either $G_{ij}> \sqrt{G_{ii}G_{jj}}$ or $G_{ij}< -\sqrt{G_{ii}G_{jj}};$ and by (c) that $G_{ij}<0$, the latter is the only possibility. This, together with (\ref{v2}) and (\ref{ww2}), implies that 
$$\langle \mathbf v_i,\mathbf v_j \rangle =    \frac{\langle \mathbf w_i, \mathbf w_j\rangle}{\sqrt{G_{ii}\det G}\sqrt{G_{jj}\det G}}=\frac{G_{ij}}{\sqrt{G_{ii}G_{jj}}} < -1,$$
which, together with (\ref{vv2}), further implies that the quadratic inequality 
$$\langle t\mathbf v_i + (1-t)\mathbf v_j, t\mathbf v_i + (1-t)\mathbf v_j \rangle  = 2(1-\langle \mathbf v_i,\mathbf v_j \rangle)t^2-2(1-\langle \mathbf v_i,\mathbf v_j \rangle)t+1<0$$
has a solution in $(0,1)$, where the equality comes from (\ref{vv2}). This is equivalent to that $L_{ij}$ intersects $\mathbb B^{2,2}.$
\end{enumerate}
This completes the proof of (i).
\smallskip

For (ii), we verify the conditions of an outward normal vector in the following steps.
\begin{enumerate}[\text{Step} 1.]
\item Since $\langle \mathbf u_i,\mathbf u_i\rangle$ equals the $i$-th diagonal entry of $G$ which equals $-1,$ $\mathbf u_i\in\mathbb A\mathrm d\mathbb S^{3^*}$ is a unit vector.

\item For $j\neq i,$ by (\ref{v2}) and (\ref{wu2}), we have
$$\langle \mathbf u_i, \mathbf v_j \rangle = \frac{\langle \mathbf u_i, \mathbf w_j \rangle}{\sqrt{G_{jj}\det G}} =0.$$
Therefore, for any vector $\mathbf v$ in the plane $F_i$ containing  $\mathbf v_j,$ $\mathbf v_k$ and $\mathbf v_l,$ $\{j,k,l\}=\{1,2,3,4\}\setminus \{i\},$ 
$$\langle \mathbf u_i, \mathbf v \rangle =0,$$
and $\mathbf u_i$ is a unit normal vector of $\Delta$. 

\item For each $i\in\{1,2,3,4\},$ by (\ref{v2}) and (\ref{wu2}),  we have
$$\langle \mathbf u_i, \mathbf v_i \rangle = \frac{\langle \mathbf u_i, \mathbf w_i \rangle}{\sqrt{G_{ii}\det G}}= \sqrt{\frac{\det G}{G_{ii}}} >0,$$
and $\mathbf u_i$ is a outward unit normal vector of $\Delta.$
\end{enumerate}
This completes the proof of (ii).
\end{proof}

\subsection{Volume, co-volume and Schl\"afli formula}

The volume and co-volume of a truncated hyperideal tetrahedron in $\ads^3$ are defined as follows. Let $(x_1,x_2,x_3)$ be a local coordinates of $\ads^3,$ and let $g_{\ads^3}=\sum_{i,j} g_{ij}dx_i\otimes dx_j$ be the semi-Riemannian metric induced from the inner product $\langle,\rangle$ of $\mathbb E^{2,2}.$ Then the  $3\times 3$ symmetric matrix $G_{\ads^3}=[g_{ij}]$ has signature $(2,1),$ and hence  $\det G_{\ads^3}<0.$ Define the volume form of $\ads^3$ as
$$d\omega_{\ads^3}=\sqrt{\det G_{\ads^3}}dx_1\wedge dx_2\wedge dx_3.$$

\begin{definition}[Volume and Co-volume]\label{vol-cov}
The \emph{volume} of a truncated hyperideal tetrahedron $\Delta$ in $\ads^3$ is defined by 
$$\mathrm{Vol}(\Delta)\doteq\int_{\Delta} |d\omega_{\ads^3}|,$$
and the \emph{co-volume} of $\Delta$ is defined 
by
$$\mathrm{Cov}(\Delta)\doteq\mathrm{Vol}(\Delta)+\frac{1}{2}\sum_{k=1}^6 \mathrm{Im}\theta_k\cdot l_i,$$
where $(\theta_1,\dots,\theta_6)$ and $(l_1,\dots,l_6)$ are respectively the dihedral angles and the edge lengths of $\Delta.$ 
\end{definition}

In Theorem \ref{ab}, Theorem \ref{volume} and Theorem \ref{volume2}, we will respectively give explicit formulae of $\mathrm{Vol}(\Delta)$ and $\mathrm{Cov}(\Delta)$ in terms of the edge lengths $(l_1,\dots,l_6)$ and of the dihedral angles $(\theta_1,\dots,\theta_6)$.  
\medskip

The Schl\"afli formula for the volume of a truncated hyperideal tetrahedron in $\ads^3$ was first stated in \cite{S} in a seemingly different form.  We re-state the result in Theorem \ref{thm:souam} below making it better adapted for our purpose.

\begin{theorem} \cite[Theorem 2]{S}
\label{thm:souam}
Let $\Delta(t)$ be a smooth deformation of a hyperideal tetrahedron $\Delta$ in $\ads^{3}$. Then the dihedral angles $(\theta_1(t),\dots,\theta_6(t))$ of $\Delta(t)$ vary smoothly with
$$\frac{d\mathrm{Vol}(\Delta(t))}{dt}=-\sum_{k=1}^6\frac{d\mathrm{Im}\theta_k(t)}{dt}\cdot l_k(t),$$
where $(l_1(t),\dots,l_6(t))$ are the edge lengths of $\Delta(t).$
\end{theorem}

\begin{remark}
 \cite[Theorem 2]{S} was originally stated for polyhedra in $\ads^3$ satisfying the following two conditions:
\begin{enumerate}[(1)]
    \item The metrics restricted to each of the edges and faces is non-degenerate.
    \item If the outward unit normal vectors  $\mathbf v_{1}$ and $\mathbf v_{2}$  of two adjacent faces span a definite metric space, then $\mathbf v_{1}+\mathbf v_{2}\ne 0.$
\end{enumerate}
These conditions are automatically satisfied by a truncated hyperideal tetrahedron considered as a polyhedron in $\ads^{3}$.
\end{remark}

\begin{remark} 
In \cite{S}, the definition of the dihedral angles differ from ours by a multiplication by $-\mathbf{i}.$ 
\end{remark}

As immediate consequences of Theorem \ref{thm:souam}, we have 

\begin{proposition}\label{Sch}
Let $\mathrm{Vol}(\Delta)=\mathrm{Vol}(\theta_1,\dots,\theta_6)$ be the volume of the truncated hyperideal tetrahedron $\Delta$ in $\ads^{3}$ considered as a function of the dihedral angles $(\theta_{1}, \ldots \theta_{6})$. Then 
\begin{equation}\label{sch}
\frac{\partial \mathrm{Vol}(\Delta)}{\partial \mathrm{Im}\theta_k}=-\frac{l_k}{2},
\end{equation}
where $(l_1,\dots,l_6)$ are the edge lengths of $\Delta.$
\end{proposition}

\begin{proposition}\label{CoSch}
Let $\mathrm{Cov}(\Delta)=\mathrm{Cov}(l_1,\dots,l_6)$ be the co-volume of the truncated  hyperideal tetrahedron $\Delta$ in $\ads^{3}$ considered as a function of the  edge lengths $(l_{1}, \ldots l_{6})$. Then 
\begin{equation}\label{cosch}\frac{\partial \mathrm{Cov}(\Delta)}{\partial l_k}=\frac{\mathrm{Im}\theta_k}{2},
\end{equation}
where $(\theta_1,\dots,\theta_6)$ are the dihedral angles of $\Delta.$
\end{proposition}

\begin{proof}
Let $\boldsymbol l=(l_1,\dots,l_6)$ and $\boldsymbol\theta=(\theta_1,\dots,\theta_6),$ and write  the co-volume function as
$$\mathrm{Cov}(\boldsymbol l)=\mathrm{Vol}(\theta_{1}(\boldsymbol l),\ldots,\theta_{6}(\boldsymbol l))+\frac{1}{2}\sum_{i=1}^{6}\mathrm{Im}\theta_i(\boldsymbol l)\cdot l_i.$$
Then we have
\begin{equation*}
\begin{split}
\frac{\partial \mathrm{Cov}(\boldsymbol l)}{\partial l_{k}}&=\sum_{i=1}^{6}\frac{\partial \mathrm{Vol}(\boldsymbol\theta)}{\partial \mathrm{Im}\theta_{i}}\cdot \frac{\partial \mathrm{Im}\theta_{i}}{\partial l_{k}}+\frac{\mathrm{Im}\theta_{k}}{2}+\frac{1}{2}\sum_{i=1}^{6}\frac{\partial \mathrm{Im}\theta_{i}}{\partial l_{k}}\cdot l_{k}\\
&=\sum_{i=1}^{6} \left(-\frac{l_{k}}{2}\right)\cdot \frac{\partial\mathrm{Im} \theta_{i}}{\partial l_{k}}+\frac{\mathrm{Im}\theta_{k}}{2}+\frac{1}{2}\sum_{i=1}^{6}\frac{\partial \mathrm{Im}\theta_{i}}{\partial l_{k}}\cdot l_{k}=\frac{\mathrm{Im}\theta_{k}}{2},
\end{split}
\end{equation*}
where the second equality comes from (\ref{sch}).
\end{proof}

The following Proposition \ref{cont} will be needed in the proof of Proposition \ref{critical2}, which plays a key role in the proof of Theorem \ref{vol2}.

\begin{proposition}\label{cont}
Let $\Delta$ be a truncated Euclidean tetrahedron and let $\Phi: \Delta\times [0,1]\rightarrow \ads^3$ be a continuous map such that for each $t\in (0,1]$ the image $\Phi(\Delta\times\{t\})$ is a truncated hyperideal tetrahedron in $\ads^3,$ and the image $\Phi(\Delta\times\{0\})$ is a flat (truncated hyperideal) tetrahedron in a space-like totally geodesic plane in $\ads^3$.
Then $$\lim_{t\to 0}\mathrm{Vol}(\Phi(\Delta\times\{t\}))=\mathrm{Vol}(\Phi(\Delta\times\{0\}))=0.$$
\end{proposition}

\begin{proof} The Euclidean metric $ds^2=dx_1^2+dx_2^2+dx^2_3+dx_4^2$ on $\mathbb{R}^4$ induces a Riemannian metric $g_{\mathbb E}$ on $\ads^3\subset \mathbb E^{2,2}\cong \mathbb R^4$, which further induces the (subspace) topology on $\ads^3.$
Since $\Phi$ is continuous and $\Delta$ is compact, for each $\epsilon>0$ there exists a $\delta>0$ such that the image  $\Phi(\Delta\times [0, \delta))$ lies in the region 
$$D_{\epsilon}(\Phi(\Delta\times\{0\}))\doteq \big\{x\in \ads^{3}\ \big|\ d_E(x, \Phi(\Delta\times\{ 0\}))\leqslant  \epsilon\big\},$$
where $d_{\mathbb E}$ is the distance on $\ads^3$ induced by $g_{\mathbb E}$. On the other hand, since $D_{\epsilon}(\Phi(\Delta, 0))$ lies in $\ads^3$ which has a positive Euclidean distance from the origin, the ratio $\frac{|\sqrt{\det(g_{\ads^3})}|}{\sqrt{\det(g_{\mathbb E})}}=\frac{1}{\sqrt{x_1^2+x_2^2+x_3^2+x_4^2}}$ is bounded from above by some constant $C>0$ on it.
 As a consequence, letting $\mathrm{Vol}_{\mathbb E}$ be the volume in the Riemannian metric $g_{\mathbb E},$ we have 
 $\mathrm{Vol}(\Delta(t))\leqslant  C\cdot\mathrm{Vol}_{\mathbb E}(D_{\epsilon}(\Phi(\Delta\times\{ 0\}))),$ which converges to $0$ as $\epsilon\to 0$. 
\end{proof}

\begin{remark}\label{ccov} It is in fact more nature to consider the \emph{genuine  volume} of $\Delta$ defined by
$$\mathrm{Vol}_{\mathbb C}(\Delta)\doteq\int _{\Delta} d\mathrm{\omega_{\ads^3}}=\mathrm{Vol}(\Delta)\cdot \mathbf i,$$ and the \emph{genuine co-volume} of $\Delta$ defined by
$$\mathrm{Cov}_{\mathbb C}(\Delta)\doteq\mathrm{Vol}_{\mathbb C}(\Delta)+\frac{1}{2}\sum_{k=1}^6 \theta_k\cdot l_i.$$
Then by Remark \ref{extcov},
$$\mathrm{Re}\mathrm{Cov}_{\mathbb C}(\Delta)=\frac{\pi}{2} \cdot (l_i+l_{i+3})=\widetilde{\mathrm{Cov}}(l_1,\dots,l_6),$$
where $i\in\{1,2,3\}$ is such that $\mathrm{Re}\theta_i=\mathrm{Re}\theta_{i+3}=\pi;$ and 
$$\mathrm{Im}\mathrm{Cov}_{\mathbb C}(\Delta)=\mathrm{Vol}(\Delta)+\frac{1}{2}\sum_{k=1}^6 \mathrm{Im}\theta_k\cdot l_i=\mathrm{Cov}(\Delta).$$
In this way, (\ref{sch}) and  (\ref{cosch}) can be re-written as
$$\frac{\partial \mathrm{Vol}_{\mathbb C}(\Delta)}{\partial \theta_k}=-\frac{l_k}{2}\quad\text{and}\quad\frac{\partial \mathrm{Cov}_{\mathbb C}(\Delta)}{\partial l_k}=\frac{\theta_k}{2}.$$
Moreover, the formula in Theorem \ref{vol2} can be re-written as 
$$\bigg\{\begin{matrix}  \frac{Q}{2}+\frac{\theta_4}{2\pi b} & \frac{Q}{2}+ \frac{\theta_5}{2\pi b} & \frac{Q}{2}+\frac{\theta_6}{2\pi b}\\\frac{Q}{2} + \frac{\theta_1}{2\pi b} & \frac{Q}{2}+\frac{\theta_2}{2\pi b} & \frac{Q}{2}+\frac{\theta_3}{2\pi b}  \end{matrix} \bigg\}_b=\frac{e^{\frac{-\mathrm{Vol}_{\mathbb C}(\Delta)}{\pi b^2} }}{\sqrt[4]{-\det\mathrm{Gram}(\Delta)}} \Big(1 +O\big(b^2\big)\Big).$$
Compare with Theorem \ref{vol3}, and see also Remark \ref{3.4} and Remark \ref{3.20}.
\end{remark}

\subsection{Relationship with hyperbolic four-hole spheres}\label{dig1}

In this subsection, we  describe a natural  one-to-one correspondence between truncated hyperideal tetrahedra in $\ads^3$ and oriented hyperbolic four-hole spheres with geodesic boundary, relating the edge lengths and the dihedral angles of the former to the Fenchel-Nielsen coordinates of the latter. See Formulae (\ref{lengths=}) and (\ref{angles=}). This correspondence plays a key role in the relationship between the asymptotics of $b$-$6j$ symbols and of the Liouville fusion kernel. The idea is inspired by a construction  in \cite{NRS}.
\medskip

Let $\mathbb M(2,2)$ be the space of $2\times 2$ matrices with real entries, and consider the isomorphism $\Phi: \mathbb M(2,2)\to \mathbb E^{2,2}$ defined by 
$$\Phi\bigg(\left[\begin{matrix}
a & b \\
c & d\\
 \end{matrix}\right]\bigg)=\bigg(\frac{a+d}{2},\frac{b-c}{2},\frac{a-d}{2},\frac{b+c}{2}\bigg).$$
By abuse of notation, let $\langle,\rangle$ also be the inner products on $\mathbb M(2,2)$ induced  from  $\mathbb E^{2,2}$ via $\Phi$. Then for $A,B\in \mathbb M(2,2),$  we have 
$$\langle A, B\rangle=\frac{1}{2}\mathrm{Tr}\big(A\cdot\mathrm{Ad}(B)\big),$$
where $\mathrm{Ad}(B)$ is the adjugate matrix of $B.$ As a consequence, we have
$$\Phi^{-1}(\ads^{3^*})=\{A\in\mathbb M(2,2)\ |\ \det(A)=1\}=\mathrm{SL}(2;\mathbb R),$$
which identifies the Lie group $\mathrm{SL}(2;\mathbb R)$ with $\ads^{3^*}.$ 
\medskip

To see the correspondence, for one direction, let $S$ be an oriented hyperbolic four-hole sphere with four  boundary geodesics $\gamma_1,\dots,\gamma_4$ with the induced orientation. Let $\gamma_5$ be a simple closed geodesic that separates $S$ into two hyperbolic pairs of pants $P_1$ containing $\gamma_1,\gamma_2$ and $P_2$ containing $\gamma_3,\gamma_4$, with the orientation induces from that of $P_1;$ and let $\gamma_6$ be a simple closed geodesic that intersects $\gamma_5$ at two points, and that separates $S$ into two hyperbolic pairs of pants $P_3$ containing $\gamma_2,\gamma_3$ and $P_4$ containing $\gamma_1,\gamma_4$, with the orientation induces from that of $P_3.$  Let $\rho:\pi_1(S)\to \mathrm{PSL}(2;\mathbb R)$ be the holonomy representation of the hyperbolic structure of $S,$ and for each $k\in \{1,2,3\},$ let $A_k\in \mathrm{SL}(2;\mathbb R)$ be the lift of $\rho([\gamma_k])\in\mathrm{PSL}(2;\mathbb R)$ with  $\mathrm{Tr(A_k)}<-2.$ Let $A_4=(A_3A_2A_1)^{-1}$, $A_5=(A_2A_1)^{-1}$ and $A_6=(A_3A_2)^{-1}.$ Then for each $k\in\{3,4,5\},$ $A_k\in \mathrm{SL}(2;\mathbb R)$ is a lift of $\rho([\gamma_k])\in\mathrm{PSL}(2;\mathbb R)$; and due to the fact that the holonomy representation for each hyperbolic pair of pants has the relative Euler class $-1$, we have $\mathrm{Tr(A_k)}<-2.$ (See \cite[Section 3.2 and Proposition 3.5]{G} for the definition and more details.)
Recall the identification $\Phi:\mathrm{SL}(2;\mathbb R)\cong \ads^{3^*}$, and let 
$$\mathbf v_1=\Phi(\mathrm{I}),\quad \mathbf v_2=\Phi(\mathrm{A_2A_1}),\quad \mathbf v_3=\Phi(\mathrm{A_1}) \quad \text{and} \quad \mathbf v_4=\Phi(\mathrm{A_3A_2A_1}),$$
where $\mathrm I$ is the $2\times 2$ identity matrix.
Then we claim that $\{\mathbf v_1,\mathbf v_2,\mathbf v_3,\mathbf v_4\}$ determine a hyperideal tetrahedron $\Delta$ in $\ads^3$ with them as the vertices. Indeed, we have 
\begin{equation}\label{boundarylengths}
\begin{split}
\langle\mathbf v_1,\mathbf v_2\rangle &= \langle \mathrm{I},A_2A_1\rangle=\frac{1}{2}\mathrm{Tr}\big((A_2A_1)^{-1}\big)=\frac{1}{2}\mathrm{Tr}A_5< -1,\\
\langle\mathbf v_1,\mathbf v_3\rangle  &= \langle \mathrm{I},A_1\rangle=\frac{1}{2}\mathrm{Tr}(A_1^{-1})=\frac{1}{2}\mathrm{Tr}A_1< -1,\\
\langle\mathbf v_1,\mathbf v_4\rangle & = \langle \mathrm{I},A_3A_2A_1\rangle =\frac{1}{2}\mathrm{Tr}\big((A_3A_2A_1)^{-1}\big)=\frac{1}{2}\mathrm{Tr}A_4< -1,\\
\langle\mathbf v_2,\mathbf v_3\rangle &= \langle A_2A_1,A_1\rangle =\frac{1}{2}\mathrm{Tr}A_2< -1,\\
\langle\mathbf v_2,\mathbf v_4\rangle & = \langle A_2A_1,A_3A_2A_1\rangle=\frac{1}{2}\mathrm{Tr}(A_3^{-1})=\frac{1}{2}\mathrm{Tr}A_3< -1,\\
\langle\mathbf v_3,\mathbf v_4\rangle & = \langle A_1,A_3A_2A_1\rangle =\frac{1}{2}\mathrm{Tr}\big((A_3A_2)^{-1}\big)=\frac{1}{2}\mathrm{Tr}A_6< -1.
\end{split}
\end{equation}
As a consequence, for each pair $\{i,j\}\subset \{1,2,3,4\},$  the straight line segment $L_{ij}$ connecting $\mathbf v_i$ and $\mathbf v_j$ intersects $\mathbb B^{2,2},$ and hence $\{\mathbf v_1,\mathbf v_2,\mathbf v_3,\mathbf v_4\}$ are the vertices of a truncated hyperideal tetrahedron $\Delta$ in $\ads^3.$ 
\smallskip

Conversely, let $\Delta$ be a truncated hyperideal tetrahedron in $\ads^3$ with vertices $\mathbf v_1,\mathbf v_2,\mathbf v_3,\mathbf v_4\in \ads^{3^*}.$ For $i\in\{1,2,3,4\},$ let 
$$B_i=\Phi^{-1}(\mathbf v_i)\in\mathrm{SL}(2;\mathbb R),$$
and let 
$$A_1=B_1B_3^{-1},\ \ 
A_2=B_2B_3^{-1},\ \ 
A_3=B_2B_4^{-1},\ \ 
A_4=B_1B_4^{-1},\ \ 
A_5=B_1B_2^{-1}\ \ \text{and}\ \ 
A_6=B_3B_4^{-1}.$$
Then by the equalities in (\ref{boundarylengths}), we have 
$$\mathrm{Tr}A_k=2\langle\mathbf v_i,\mathbf v_j\rangle$$
for all the triples $(k,i,j)\in\{(1,1,3),(2,2,3),(3,2,4),(4,1,4),(5,1,2),(6,3,4)\}$; and by (\ref{cosh}), we have
\begin{equation}\label{<-2}
    \mathrm{Tr}A_k<-2
\end{equation}
for each $k\in\{1,\dots, 6\}.$ Now for an oriented four-hole sphere $S$ with boundary curves $\gamma_1,\dots,\gamma_4$ with the induced orientation, let $\gamma_5$ be a simple closed curve that separates $S$ into two pants $P_1$ containing $\gamma_1,\gamma_2$ and $P_2$ containing $\gamma_3,\gamma_4$ with the orientation induces from that of $P_1,$ and let $\gamma_6$ be a simple closed curve that intersects $\gamma_5$ at points and that separates $S$ into two pants $P_3$ containing $\gamma_2,\gamma_3$ and $P_4$ containing $\gamma_1,\gamma_4$ with the orientation induces from that of $P_3.$
We define a representation $\rho:\pi_1(S)\to\mathrm{PSL}(2;\mathbb R)$ by
letting $\rho([\gamma_k])=\pm A_k$
for $k\in\{1,2,3\}.$
 Then we have $\rho([\gamma_k])=\pm A_k$ for $k\in\{4,5,6\}.$ By (\ref{<-2}) and \cite[Proposition 3.5, Proposition 3.7 and Theorem 3.4]{G}, the restriction of $\rho$ to each of $P_1$ and $P_2$ has relative Euler class $-1$, hence $\rho$ itself has relative Euler class $-2,$ which is the hononomy representation of a hyperbolic structure on $S$ with geodesic boundary. 
\medskip

Observe that if we orient $S$ oppositely, then the induced orientation on $\gamma_i$ is reversed for each $i\in\{1,\dots,4\}$, and the holonomy representation $\rho$ sends each $[\gamma_i]\in\pi_1(S)$ to the inverse element in $\mathrm{PSL}(2;\mathbb R)$ which lifts to $A_i^{-1}\in\mathrm{SL}(2;\mathbb R).$ As a consequence, the corresponding hyperideal tetrahedron $\overline{\Delta}$ in $\ads^3$ has vertices $\Phi(A_i^{-1})=\overline {\mathbf v}_i,$ $i\in\{1,2,3,4\},$ where for $\mathbf v=(v_1,v_2,v_3,v_4)\in\mathbb E^{2,2}$, $\overline{\mathbf v}=(v_1,-v_2,-v_3,-v_4)$. Therefore, $\overline{\Delta}$ is isometric to the mirror image of $\Delta$.
\medskip

Now, under this correspondence, if for $k\in\{1,\dots,6\}$ we let $l(\gamma_k)$ be the length of the simple closed geodesic $\gamma_k$ on $S,$ and for $\{i,j\}\subset\{1,2,3,4\}$ we let $l_{ij}$ be the length of the edge $e_{ij}$ of $\Delta$ connecting the vertices $\mathbf v_i$ and $\mathbf v_j,$ then by the equalities in (\ref{boundarylengths}) and (\ref{cosh}), we have 
\begin{equation}\label{lengths=}
\big(l(\gamma_1),l(\gamma_2),l(\gamma_3),l(\gamma_4),l(\gamma_5),l(\gamma_6)\big)=2\big(l_{13},l_{23},l_{24},l_{14},l_{12},l_{34}\big).
    \end{equation}

Next, we show that in the Fenchel-Nielsen coordinates of $S$ with respect to the pants decomposition  given by the system of curves $\{\gamma_1,\gamma_2,\gamma_3,\gamma_4,\gamma_5\}$,
if we let $t(\gamma_5)$ be the twist parameter at $\gamma_5$, then 
$$t(\gamma_5)=\pm \mathrm{Im}\theta_{12},$$
where $\theta_{12}$ is the dihedral angle of $\Delta$ at the edge $e_{12}.$ Indeed, using the Law of Cosine for right-angled hyperbolic  hexagons and for  twisted right-angled hyperbolic hexagons (see Figure \ref{folded} and \cite[Appendix A, Case 6]{RY}), 
one has 
\begin{equation}\label{cosh3}
 {\cosh t(\gamma_5)=-\frac{G_{34}}{\sqrt{G_{33}G_{44}}}},
\end{equation}
where for $i,j\in\{1,2,3,4\},$ $G_{ij}$ is the $ij$-th co-factor of the matrix 
\begin{equation*}
\begin{bmatrix}
1 & -\cosh \frac{l(\gamma_5)}{2} & -\cosh \frac{l(\gamma_1)}{2}
    & -\cosh \frac{l(\gamma_4)}{2} \\
-\cosh \frac{l(\gamma_5)}{2} & 1 & -\cosh \frac{l(\gamma_2)}{2}
    & -\cosh \frac{l(\gamma_3)}{2}\\
    -\cosh \frac{l(\gamma_1)}{2} & -\cosh \frac{l(\gamma_2)}{2} & 1
    & -\cosh \frac{l(\gamma_6)}{2} \\
 -\cosh \frac{l(\gamma_4)}{2}& -\cosh \frac{l(\gamma_3)}{2}& -\cosh \frac{l(\gamma_6)}{2}
    & 1
  \end{bmatrix},
\end{equation*}
which by (\ref{lengths=}) coincides with the Gram matrix of $\Delta$ in the edge lengths. On the other hand, by Proposition \ref{cos} and (\ref{cos2}),  we have 
\begin{equation}\label{cos3}
 {\cos\theta_{12}=\frac{G_{34}}{\sqrt{G_{33}G_{44}}}}.
\end{equation}
 Comparing (\ref{cosh3}) and (\ref{cos3}), we have the result.

\begin{figure}[htbp]
\centering
\includegraphics[scale=0.2]{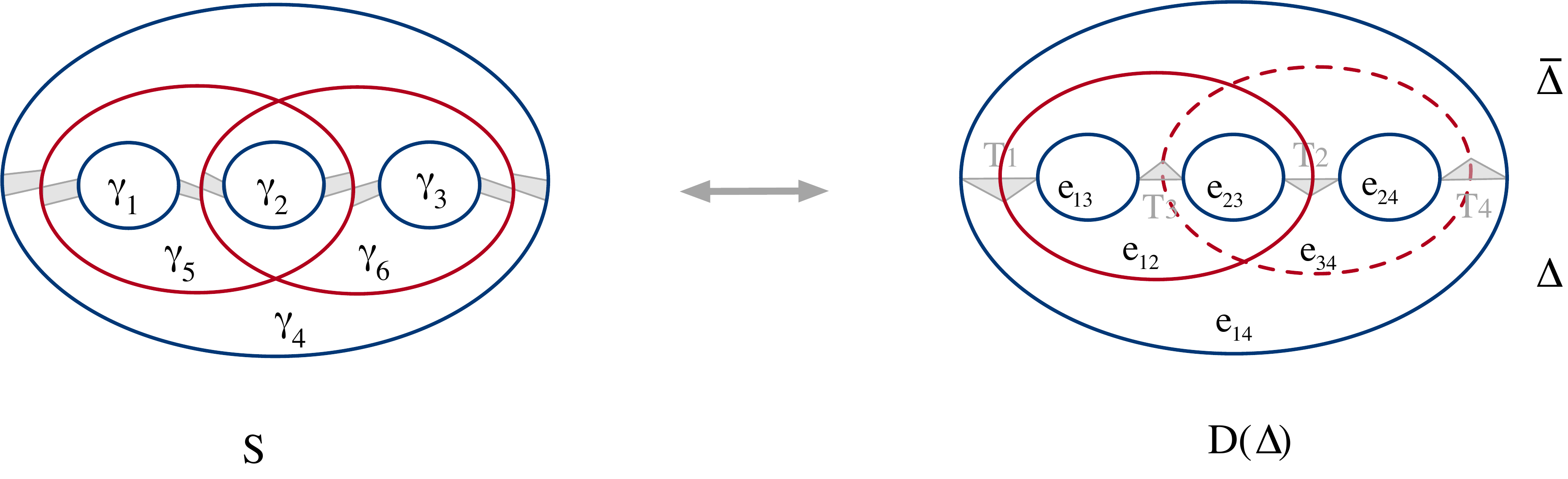}
\caption{On the left we have the hyperbolic four-hole sphere $S,$ where the shaded regions are bounded by $\gamma_i$'s and the shortest geodesic arcs between them, hence are twisted right-angled hyperbolic hexagons as in \cite[Appendix A, Case 6]{RY}. On the right we have the double $D(\Delta)$ of the truncated hyperideal tetrahedron $\Delta$ in $\ads^3$ obtained by gluing $\Delta$ and its mirror image $\overline\Delta$ together along the shaded triangles of truncation $T_1,\dots, T_4$, which is topologically a handlebody  whose boundary is a closed surface of genus $3$. }
\label{folded}
\end{figure}

Finally, if we also consider the pants decomposition given by the system of curves $\{\gamma_1,\gamma_2,\gamma_3,\gamma_4,\gamma_6\}$, combine it with the previous pants decompositions, and consider them together as a pants decomposition of the closed hyperbolic surface of genus $3$ pleated along the curves $\gamma_1,\dots,\gamma_6$ with bending angles $0,0,0,0,\pi,\pi$, which folds in a two-to-one manner onto $S$ as depicted in Figure \ref{folded}, then we can also talk about the twist parameters $t(\gamma_1),\dots,t(\gamma_6)$ on the curves $\gamma_1,\dots,\gamma_6;$ and by the same computations, we have 
\begin{equation}\label{angles=}
\big(t(\gamma_1),t(\gamma_2),t(\gamma_3),t(\gamma_4),t(\gamma_5),t(\gamma_6)\big)=\pm\big(\mathrm{Im}\theta_{13},\mathrm{Im}\theta_{23},\mathrm{Im}\theta_{24},\mathrm{Im}\theta_{14},\mathrm{Im}\theta_{12},\mathrm{Im}\theta_{34}\big),
\end{equation}
where $\theta_{ij}$ is the dihedral angle of $\Delta$ at the edge $e_{ij}$ connecting the vertices  $\mathbf v_i$ and $\mathbf v_j,$  for $\{i,j\}\subset\{1,2,3,4\}$, and the $\pm$ is an overall sign depending on the orientations of $S.$

\section{Asymptotics of $b$-$6j$ symbols}\label{sec3}

\subsection{Definition and basic properties}

Let $S_b$ be the  \emph{double sine function} defined for $z\in \mathbb C$ with $0<\mathrm{Re}(z)<Q$ by
\begin{equation}\label{eq:def-S}
S_b(z)=\exp\Bigg(\int_\Omega\frac{\sinh\Big(\big(\frac{Q}{2}-z\big)t\Big)}{4t\sinh(\frac{bt}{2})\sinh(\frac{t}{2b})}dt\Bigg),
\end{equation}
where $b\in (0,1)$ and $Q=b+b^{-1}$, 
and the contour $\Omega$ goes along the real line and passes above the pole at the origin. By the functional equation (see e.g \cite[A.15]{TesVar})
\begin{equation}\label{FE1}
S_b(z+b^{\pm 1})=2\sin (\pi b ^{\pm 1} z)S_b(z),
\end{equation}
$S_b$ extended to a meromorphic function on $\mathbb C$ with the set of poles $\{ -nb-mb^{-1} \ |\ m, n\in \mathbb Z_{\geqslant 0}\}$ and the set of  zeros  $\{ Q + nb + mb^{-1} \ |\ m, n\in \mathbb Z_{\geqslant 0}\}.$

For a six-tuple of complex numbers  $(a_1,\dots,a_6)$, we let in the rest of this article \begin{equation*}
\begin{split}
   & t_1=a_1+a_2+a_3,\quad  t_2=a_1+a_5+a_6,\quad t_3=a_2+a_4+a_6,\quad  t_4=a_3+a_4+a_5,\\
   & q_1=a_1+a_2+a_4+a_5,\quad q_2=a_1+a_3+a_4+a_6,\quad  
q_3=a_2+a_3+a_5+a_6\quad\text{and}\quad q_4=2Q.
\end{split}
\end{equation*}
Then the  six-tuple $(a_1,\dots,a_6)\in\mathbb C^6$ is called \emph{$b$-admissible} if 
$$0<\mathrm{Re}q_j-\mathrm{Re}t_i<Q$$
for all $i,j\in\{1,2,3,4\}.$ 
In particular, if $(a_1,\dots,a_6)$ is $b$-admissible, then 
$$\max\{\mathrm{Re}t_1,\mathrm{Re}t_2,\mathrm{Re}t_3,\mathrm{Re}t_4\}<\min\{\mathrm{Re}q_1,\mathrm{Re}q_2,\mathrm{Re}q_3,\mathrm{Re}q_4\}.$$

\noindent{\bf Definition \ref{def: b-6j}} ($b$-$6j$ symbols).
The \emph{$b$-$6j$ symbol} for a $b$-admissible six-tuple $(a_1,\dots,a_6)\in\mathbb C^6$ is given by 
\begin{equation}\label{b6j}
\bigg\{\begin{matrix} a_1 & a_2 & a_3 \\ a_4 & a_5 & a_6 \end{matrix} \bigg\}_b=\Bigg(\frac{1}{\prod_{i=1}^4\prod_{j=1}^4S_b(q_j-t_i)}\Bigg)^{\frac{1}{2}}\int_\Gamma \prod_{i=1}^4S_b(u-t_i)\prod_{j=1}^4S_b(q_j-u)d u,
\end{equation}
where the contour $\Gamma$ is any vertical line passing the interval $(\max\{\mathrm{Re}t_i\},\min\{\mathrm{Re}q_j\})$.  See Figure \ref{Gamma} (a).
\\

\begin{figure}[htbp]
\centering
\includegraphics[scale=0.33]{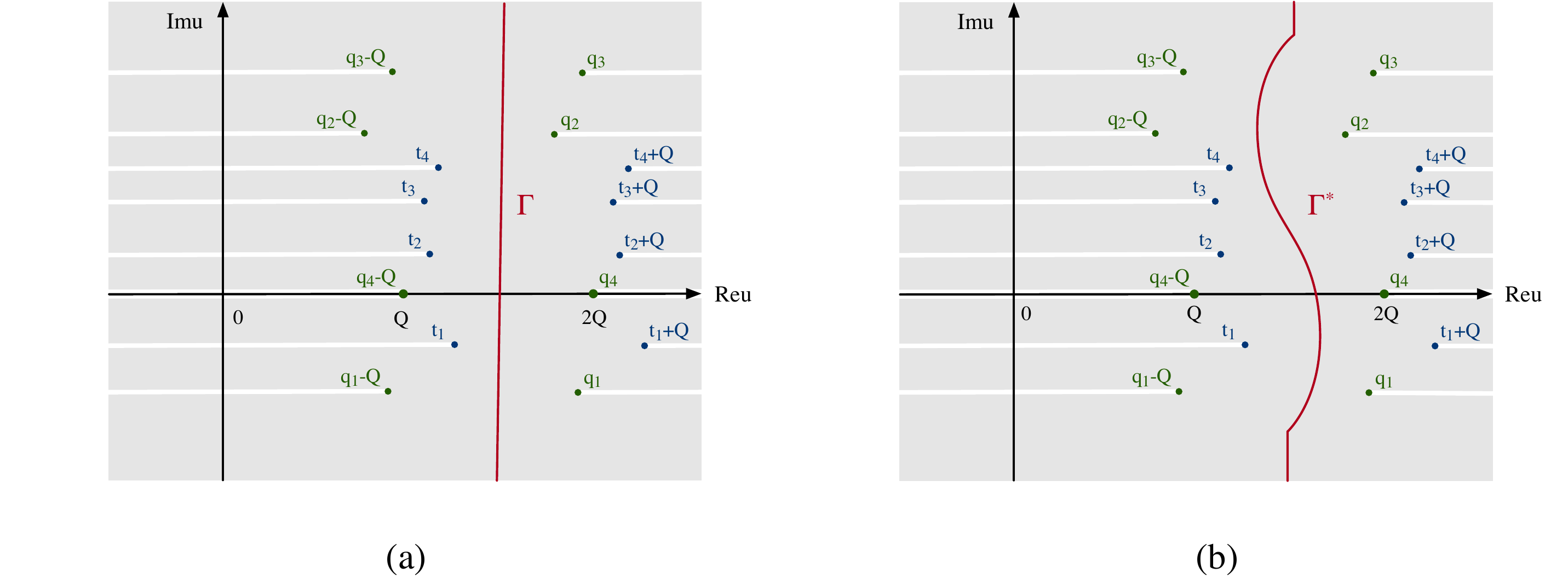}
\caption{Contours $\Gamma$, $\Gamma^*$ and possible zeros and poles (located in the white rays) of the integrand in~\eqref{b6j}.}
\label{Gamma}
\end{figure}

\begin{proposition}\label{lem: ab conver}
   The integral in (\ref{b6j}) converges absolutely, and is independent of the choice of the vertical line $\Gamma.$
\end{proposition}

\begin{proof} We show the absolute convergence first. 
Recall from \cite[Formulae (B.32), (B.52)]{Eber} that 
\begin{equation}\label{eq: reflection}
    S_b(z)S_b(Q-z)=1,
\end{equation}
    and  
   \begin{equation}\label{eq: asy}
       \lim_{|\mathrm{Im}z|\to\infty}e^{\text{sgn}(\mathrm{Im}z)(\frac{\pi \mathbf{i} }{2}z(z-Q)+\frac{\pi\mathbf{i}}{12}(Q^2+1))}S_b(z)=1.
   \end{equation}  
Then we have 
\begin{equation}\label{est}
\begin{split}
\Bigg|
 \prod_{i=1}^4S_b(u-t_i)\prod_{j=1}^4S_b(q_j-u)\Bigg|=&\Bigg|
 \prod_{i=1}^4\frac{S_b(u-t_i)}{S_b(u+Q-q_i) } \Bigg|
 \\ \sim
&\begin{cases}
C\prod_{i=1}^4e^{\pi \mathrm{Im}u\mathrm{Re}(q_i-t_i-Q)+\pi\mathrm{Re}u\mathrm{Im}(q_i-t_i-Q)}& \text{as} \quad \textnormal{Im}u \to \infty,\\
C^{-1} \prod_{i=1}^4e^{-\pi \mathrm{Im}u\mathrm{Re}(q_i-t_i-Q)-\pi\mathrm{Re}u\mathrm{Im}(q_i-t_i-Q)} &\text{as} \quad \textnormal{Im}u\to -\infty,
\end{cases}\\
=
&\begin{cases}
Ce^{-2\pi Q\mathrm{Im}u} & \text{as} \quad \textnormal{Im}u \to \infty, \\
C^{-1} e^{2\pi Q\mathrm{Im}u} & \text{as} \quad \textnormal{Im}u \to -\infty,
\end{cases}
\end{split}
\end{equation}
where $C=\prod_{i=1}^4e^{\frac{\pi}{2}\mathrm{Im}(t_i^2-q_i^2+t_iQ+q_iQ)}$ is a quantity independent of $u,$
and the last equality comes from that $\sum_{i=1}^4(q_i-t_i-Q)=-2Q$. Therefore, the integrand in (\ref{b6j}) decays exponentially at infinity, and the integral absolutely converges. 

Next, we show that the integral does not depend on the choice of the contour. For a vertical line $\Gamma$ and any $\lambda>0$, we let $\Gamma_\lambda=\{z\in \Gamma\mid -\lambda\leqslant \mathrm{Im}z\leqslant \lambda\}$; and to simplify the notation we denote the integrand of (\ref{b6j}) by $F(u)$. Then for two vertical lines   $\Gamma_1$ and $\Gamma_2$ as described in the statement of the proposition, we have
$$\int_{\Gamma_1}F(u)du-\int_{\Gamma_2}F(u)du=\lim_{\lambda\to\infty}\bigg(\int_{\Gamma_{1,\lambda}}F(u)du-\int_{\Gamma_{2,\lambda}}F(u)du\bigg),$$
which vanishes due to the analyticity of $F(u)$, the Residue Theorem and the estimates in (\ref{est}).
\end{proof}

By Proposition \ref{lem: ab conver} and the analyticity of the integrand in (\ref{b6j}), we have the following Proposition \ref{contour}, which gives us more flexibility in the computation of the $b$-$6j$ symbols. In particular, it will be used in the proof of Theorem \ref{vol2}.

\begin{proposition}\label{contour}
Given $L>0$ and $c\in [\max\{\mathrm{Re}t_i\},\min\{\mathrm{Re}q_j\}]$, let 
$\Gamma^*$ be a contour satisfying the following condition: for $u\in \Gamma^*$ with $|\mathrm {Im} u|\geqslant L$, $\mathrm {Re} u=c$; and for $u\in \Gamma^*$ with $|\mathrm {Im} u|\leqslant L,$ if $\mathrm{Im}u=\mathrm{Im}t_i$ for some $i\in\{1,2,3,4\},$ then 
$$\mathrm{Re}t_i<\mathrm{Re}u<\mathrm{Re}t_i+Q,$$
and if $\mathrm{Im}u=\mathrm{Im}q_j$ for some $j\in\{1,2,3,4\},$ then 
$$\mathrm{Re}q_j-Q<\mathrm{Re}u<\mathrm{Re}q_j.$$ 
See Figure \ref{Gamma} (b).
Then the $b$-$6j$ symbol can be computed by
\begin{equation*}
\bigg\{\begin{matrix} a_1 & a_2 & a_3 \\ a_4 & a_5 & a_6 \end{matrix} \bigg\}_b=\Bigg(\frac{1}{\prod_{i=1}^4\prod_{j=1}^4S_b(q_j-t_i)}\Bigg)^{\frac{1}{2}}\int_{\Gamma^*} \prod_{i=1}^4S_b(u-t_i)\prod_{j=1}^4S_b(q_j-u)du.
\end{equation*} 
\end{proposition}

From \eqref{b6j}, one sees immediately the following tetrahedral symmetry of the b-6j symbols.
\begin{proposition}[Tetrahedral symmetry]\label{tetra}
    \begin{align*}
\begin{Bmatrix} 
      a_1 & a_2 & a_3 \\
      a_4 & a_5 & a_6
   \end{Bmatrix}_b=\begin{Bmatrix} 
      a_2 & a_1 & a_3 \\
      a_5 & a_4 & a_6
   \end{Bmatrix}_b=\begin{Bmatrix} 
      a_1 & a_3 & a_2 \\
      a_4 & a_6 & a_5
   \end{Bmatrix}_b=\begin{Bmatrix} 
      a_1 & a_5 & a_6 \\
      a_4 & a_2 & a_3
   \end{Bmatrix}_b.
\end{align*}
\end{proposition}

The next proposition describes the reflection symmetry of the $b$-$6j$ symbols that if we change any $a_k\mapsto Q-a_k$, the value of the $b$-$6j$ symbol does not change.
\begin{proposition}[Reflection symmetry]\label{reflection}
    If $(a_1,\dots,a_6)$ is $b$-admissible, then $(Q-a_1,a_2,\dots,a_6)$ is also $b$-admissible, and   
$$\bigg\{\begin{matrix} a_1 & a_2 & a_3 \\ a_4 & a_5 & a_6 \end{matrix} \bigg\}_b=\bigg\{\begin{matrix} Q-a_1 & a_2 & a_3 \\ a_4 & a_5 & a_6 \end{matrix} \bigg\}_b.$$
Moreover, the same result holds if we change $a_k\mapsto Q-a_k$ for any $k\in\{1,\dots,6\}.$
\end{proposition}

\begin{proof}
To see the $b$-admissibility of $(Q-a_1,a_2,\cdots,a_6)$, we denote the $t_i$'s and  $q_j$'s for this new six-tuples by $t_i'$'s and $q_j'$'s, and reserve the $t_i$'s and $q_j$'s for the original six-tuple $(a_1,\dots,a_6).$  By a direction computation, we have:
\begin{enumerate}[(a)]
\item If $i,j\in\{1,2\}$ or $i,j\in\{3,4\},$ then 
$$q_j'-t_i'=q_j-t_i.$$
\item 
If $i\in\{1,2\}$ and $j\in\{3,4\},$ or $i\in\{3,4\}$ and $j\in\{1,2\},$ then 
$$q_j'-t_i'= Q-(q_l-t_k),$$
 where $\big\{\{i,k\},\{j,l\}\big\}=\big\{\{1,2\},\{3,4\}\big\}.$
\end{enumerate}
Then the $b$-admissibility of the six-tuple $(Q-a_1,a_2,\cdots,a_6)$ follows form (a), (b) and the $b$-admissibility of the six-tuple $(a_1,\dots,a_6).$

  To see the reflection symmetry, we recall in  \cite[Formulae (2.9),(2.16), (2.17), (2.24)]{TesVar} that the $b$-$6j$ symbol admits an alternative integral representation, in which 
   (specifically in (2.16)) the integral part is invariant under the substitution $a_1\mapsto Q-a_1$. The invariance of the factor in front of the integral in \cite[(2.24)]{TesVar} can be verified by (a), (b), (\ref{eq: reflection}) and a direct computation.

   The last statement follows immediately from Proposition \ref{tetra}, the tetrahedral symmetry.
\end{proof}

To prove Theorems \ref{cov}, \ref{vol3} and \ref{vol2}, following \cite{LMSWY}, we introduce the following a new set of variables $\alpha_k$, $\xi$, $\tau_i$, $\eta_j$:
\begin{align}
&\alpha_k=\pi b a_k-\frac{\pi b^2}{2} \textrm{ for }k\in\{1,\dots, 6\},\quad \xi=\pi b u, \label{alpha-a}\\
 &\tau_i=\pi b t_i-\frac{3\pi b^2}{2}\textrm{ for }i\in\{1,2,3,4\},\quad\text{and}\quad \eta_j=\pi b q_j-2\pi b^2 \textrm{ for }j\in\{1,2,3,4\}.\label{tau-t}
\end{align}
Then 
\begin{equation*}
\begin{split}
&\tau_1=\alpha_1+\alpha_2+\alpha_3,\quad\tau_2=\alpha_1+\alpha_5+\alpha_6,\quad
\tau_3=\alpha_2+\alpha_4+\alpha_6,\quad\tau_4=\alpha_3+\alpha_4+\alpha_5,\\
&\eta_1=\alpha_1+\alpha_2+\alpha_4+\alpha_5,\quad \eta_2=\alpha_1+\alpha_3+\alpha_4+\alpha_6,\quad 
\eta_3=\alpha_2+\alpha_3+\alpha_5+\alpha_6\quad\text{and}\quad \eta_4=2\pi;
\end{split}
\end{equation*}
and for $i,j\in\{1,2,3,4\},$ 
$$0< \mathrm{Re}\eta_j-\mathrm{Re}\tau_i < \pi,$$ and
$$\max\{\mathrm{Re}\tau_1,\mathrm{Re}\tau_2,\mathrm{Re}\tau_3,\mathrm{Re}\tau_4\}<\min\{\mathrm{Re}\eta_1,\mathrm{Re}\eta_2,\mathrm{Re}\eta_3,\mathrm{Re}\eta_4\}.$$

Let $\boldsymbol\alpha=(\alpha_1,\dots,\alpha_6)$ and let
\begin{equation*}
\begin{split}
U_{\boldsymbol \alpha,b}(\xi)= & -\pi \mathbf i b^2 \sum_{i=1}^4\sum_{j=1}^4 \log S_b(q_j-t_i) + 2\pi \mathbf i b^2 \sum_{i=1}^4  \log S_b(u-t_i) + 2\pi \mathbf i b^2 \sum _{j=1}^4\log S_b(q_j-u).
\end{split}
\end{equation*}
Then we have 
\begin{equation}\label{fk}
\bigg\{\begin{matrix} a_1 & a_2 & a_3 \\ a_4 & a_5 & a_6 \end{matrix} \bigg\}_b=\frac{1}{\pi b}\int_{\Gamma}\exp\bigg(\frac{U_{\boldsymbol \alpha,b}(\xi)}{2\pi \mathbf ib^2} \bigg)d\xi,
\end{equation}
where $\Gamma$ is a vertical line passing  the real axis in the interval $(\max\{\mathrm{Re}\tau_i\},\min\{\mathrm{Re}\eta_j\})$.

Define the function $U_{\boldsymbol \alpha}$ by 
\begin{equation}\label{U}
U_{\boldsymbol \alpha}(\xi)=-\frac{1}{2}\sum_{i=1}^4\sum_{j=1}^4L(\eta_j-\tau_i)+\sum_{i=1}^4L(\xi-\tau_i)+\sum_{j=1}^4L(\eta_j-\xi),
\end{equation}
where 
\begin{equation}\label{eq:Lx}
L(x)= x^2-\pi x +\frac{\pi^2}{6}-\mathrm{Li}_2\big(e^{2\mathbf ix}\big)
\end{equation}
and $\mathrm{Li}_2$ is the dilogarithm function defined on $\mathbb C\setminus (1,+\infty)$ by
$$\mathrm{Li}_2(z)=-\int_0^z\frac{\log (1-u)}{u}du.$$
Here the integral is along any path in $\mathbb C\setminus (1,\infty)$ connecting $0$ and $z.$ It is proved in \cite[Section 2.2]{LMSWY} that the function $L(x)$ extends continuously on $\mathbb C\setminus (-\infty,0)\cup(\pi,\infty)$ and  holomorphically on $\mathbb C\setminus (-\infty,0]\cup[\pi,\infty),$ by the following functional equations:
 If $x\in\mathbb C\setminus (-\infty,0]\cup[\pi,\infty)$ with $\mathrm{Im}x>0$, then
\begin{equation*}\label{period3}
L(x+\pi)=L(x)+2\pi x;
\end{equation*}
and if $x\in \mathbb C\setminus (-\infty,0]\cup[\pi,\infty)$ with $\mathrm{Im}x< 0$, then
\begin{equation*}\label{period4}
L(x+\pi)=L(x)-2\pi x.
\end{equation*}
Define also
\begin{equation}\label{kappa}
\kappa_{\boldsymbol \alpha}(\xi)=8\pi^2+14\pi\sum_{k=1}^6\alpha_k-28\pi\xi-4\pi \mathbf i\sum_{i=1}^4\log\Big(1-e^{2\mathbf i(\xi-\tau_i)}\Big)+3\pi \mathbf i\sum_{j=1}^4\log\Big(1-e^{2\mathbf i(\eta_j-\xi)}\Big),
\end{equation}
and 
\begin{equation}\label{nuU}\nu_{\boldsymbol\alpha,b}(\xi)=\frac{U_{\boldsymbol \alpha,b}(\xi)-\kappa_{\boldsymbol \alpha}(\xi)b^2-U_{\boldsymbol \alpha}(\xi)}{b^4}.
\end{equation} 
Then by (\ref{fk}) and (\ref{nuU}), we have  
\begin{equation}\label{6jint}
\bigg\{\begin{matrix} a_1 & a_2 & a_3 \\ a_4 & a_5 & a_6 \end{matrix} \bigg\}_b=\frac{1}{\pi b}\int_{\Gamma}\exp\bigg(\frac{U_{\boldsymbol \alpha}(\xi)+\kappa_{\boldsymbol \alpha}(\xi)b^2+
\nu_{\boldsymbol \alpha,b}(\xi)b^4}{2\pi \mathbf ib^2} \bigg)d\xi,
\end{equation}
where $\Gamma$ is a vertical line passing the real axis in the interval $(\max\{\mathrm{Re}\tau_i\},\min\{\mathrm{Re}\eta_j\})$.

\subsection{Asymptotics in the edge lengths}\label{el}

The goal of this subsection is to prove Theorem \ref{cov}.

Let $(l_1,\dots, l_6)\in \mathbb R_{>0}^6;$  
and for $k\in\{1,\dots,6\}$, let $a_k =\frac{Q}{2} + \mathbf i \frac{l_k}{2\pi b}$.
Then the six-tuple $(a_1,\dots,a_6)$ is $b$-admissible, and by Definition \ref{b6j}, the $b$-$6j$ symbol of the six-tuple $(a_1,\dots,a_6)$ is computed by
\begin{equation}\label{b-6j}
\bigg\{\begin{matrix} a_1 & a_2 & a_3 \\ a_4 & a_5 & a_6 \end{matrix} \bigg\}_b=\Bigg(\frac{1}{\prod_{i=1}^4\prod_{j=1}^4S_b(q_j-t_i)}\Bigg)^{\frac{1}{2}}\int_\Gamma \prod_{i=1}^4S_b(u-t_i)\prod_{j=1}^4S_b(q_j-u)d u,
\end{equation}
where the contour $\Gamma$ is any vertical line passing the interval $(\frac{3Q}{2},2Q)$. See Figure \ref{Da}.

\begin{figure}[htbp]
\centering
\includegraphics[scale=0.4]{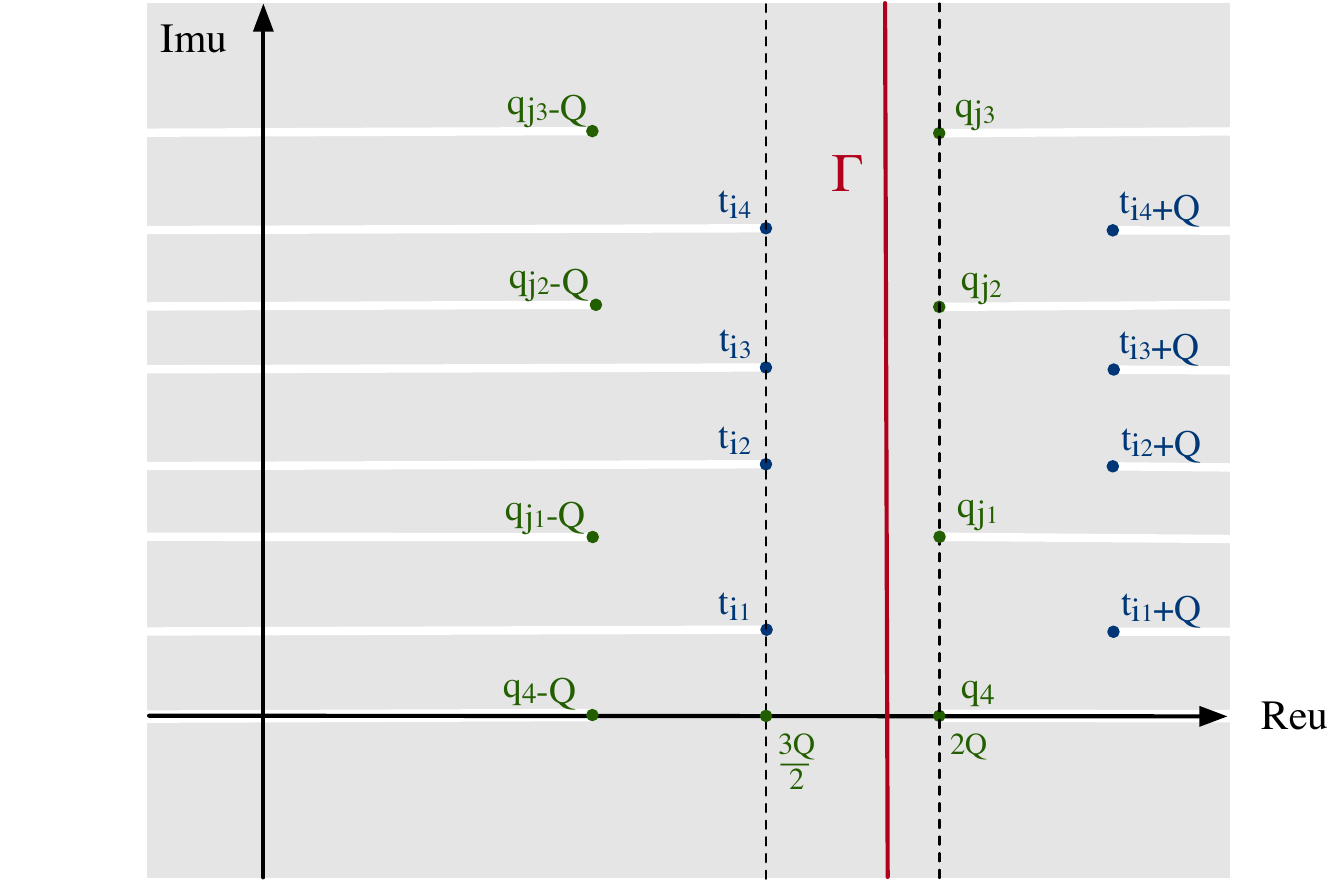}
\caption{Contour $\Gamma$ and possible zeros and poles (located in the white rays) of the integrand in~\eqref{b-6j}, where $\{i_1,i_2,i_3,i_4\}=\{1,2,3,4\}$  with $\mathrm{Im}t_{i_1}\leqslant \mathrm{Im}t_{i_2}\leqslant \mathrm{Im}t_{i_3}\leqslant \mathrm{Im}t_{i_4},$ and $\{j_1,j_2,j_3\}=\{1,2,3\}$ with $0=\mathrm{Im}q_4<\mathrm{Im}q_{j_1}\leqslant \mathrm{Im}q_{j_2}\leqslant \mathrm{Im}q_{j_3}.$}
\label{Da}
\end{figure}

Under the change of variables (\ref{alpha-a}) and (\ref{tau-t}), let $\boldsymbol \alpha = (\alpha_1,\dots,\alpha_6)$, and for $\delta >0$ sufficiently small consider the region $D_{\delta}=D^{\boldsymbol \alpha}_{\delta}$  consisting of $\xi\in \mathbb C$ either with $\mathrm {Re}\xi\in[\frac{3\pi}{2}+\delta, 2\pi-\delta],$ or with $|\mathrm{Im}(\xi-\tau_i)|\geqslant \delta,$ $|\mathrm{Im}(\eta_j-\xi)|\geqslant \delta$ for all $i, j\in\{1,2,3,4\}$.  See Figure \ref{Ddc}. Let $U_{\boldsymbol \alpha}$ be the function on $D_{\delta}$ defined by (\ref{U}). 

\begin{figure}[htbp]
\centering
\includegraphics[scale=0.5]{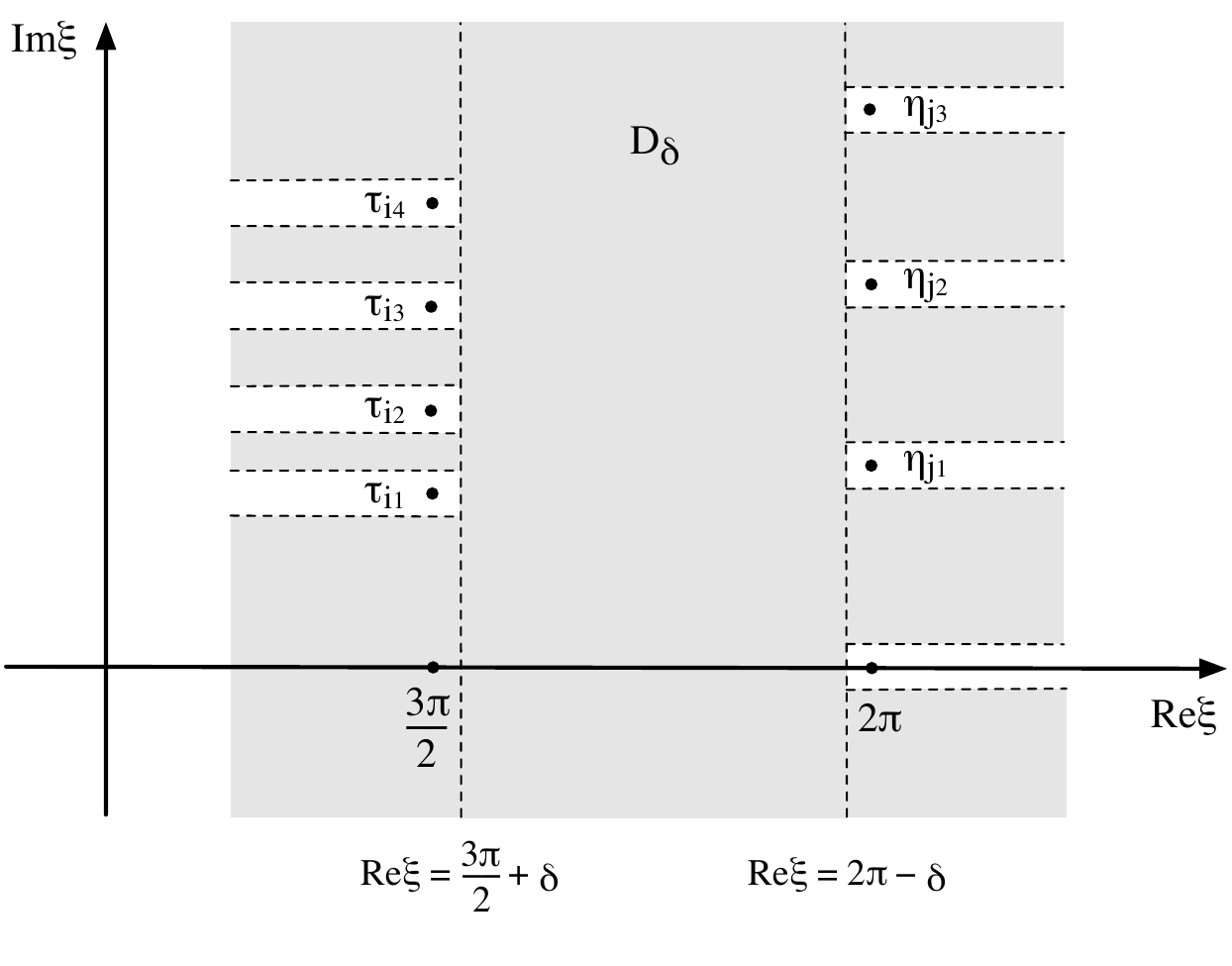}
\caption{Region $D_{\delta}$, where $\{i_1,i_2,i_3,i_4\}=\{1,2,3,4\}$ with $\mathrm{Im}\tau_{i_1}\leqslant \mathrm{Im}\tau_{i_2}\leqslant \mathrm{Im}\tau_{i_3}\leqslant \mathrm{Im}\tau_{i_4},$ and $\{j_1,j_2,j_3\}=\{1,2,3\}$ with $0=\mathrm{Im}\eta_4<\mathrm{Im}\eta_{j_1}\leqslant \mathrm{Im}\eta_{j_2}\leqslant \mathrm{Im}\eta_{j_3}.$}
\label{Ddc}
\end{figure}

\begin{proposition}\cite[Proposition 3.1 (3)]{LMSWY}\label{Prop3.1}  If $(l_1,\dots,l_6)$ are the edge lengths of a truncated hyperideal tetrahedron in $\mathbb A\mathrm d\mathbb S^3,$ then
there are two non-degenerate critical points $\xi_1^*$ and $\xi_2^*$ of $U_{\boldsymbol \alpha}$ in $D_{\delta}$ with 
$$\mathrm{Re}\xi_1^*=\mathrm{Re}\xi_2^*=2\pi.$$
 Moreover,
$$\mathrm{Im}U_{\boldsymbol \alpha}(\xi_1^*)=\mathrm{Im}U_{\boldsymbol \alpha}(\xi_2^*)=-2\widetilde{\mathrm{Cov}}(l_1,\dots,l_6),$$
where $\widetilde{\mathrm{Cov}}$ is the extended co-volume function defined in \cite{LY},
and 
$$\mathrm{Re}U_{\boldsymbol \alpha}(\xi_1^*)=-\mathrm{Re}U_{\boldsymbol \alpha}(\xi_2^*).$$
\end{proposition}

In \cite{LMSWY} the critical points $\xi_1^*$ and $\xi_2^*$ are computed as follows: For $k\in\{1,\dots,6\},$ let $u_k=e^{2\mathbf i\alpha_k}.$ Then consider the following quadratic equation 
\begin{equation*}
Az^2+Bz+C=0
\end{equation*}
with \begin{equation*}
\begin{split}
A=&u_1u_4+u_2u_5+u_3u_6-u_1u_2u_6-u_1u_3u_5-u_2u_3u_4-u_4u_5u_6+u_1u_2u_3u_4u_5u_6,\\
B=&-\Big(u_1-\frac{1}{u_1}\Big)\Big(u_4-\frac{1}{u_4}\Big)-\Big(u_2-\frac{1}{u_2}\Big)\Big(u_5-\frac{1}{u_5}\Big)-\Big(u_3-\frac{1}{u_3}\Big)\Big(u_6-\frac{1}{u_6}\Big),\\
C=&\frac{1}{u_1u_4}+\frac{1}{u_2u_5}+\frac{1}{u_3u_6}-\frac{1}{u_1u_2u_6}-\frac{1}{u_1u_3u_5}-\frac{1}{u_2u_3u_4}-\frac{1}{u_4u_5u_6}+\frac{1}{u_1u_2u_3u_4u_5u_6}.\\
\end{split}
\end{equation*}
This quadratic equation is under the identification $z=e^{-2\mathbf i\xi}$ equivalent to the equation 
$e^{U'_{\boldsymbol\alpha}(\xi)}=1.$
Since $(l_1,\dots,l_6)$ are the lengths of a truncated hyperideal tetrahedron in $\mathbb A\mathrm d\mathbb S^3,$ as explained in \cite[Proof of Proposition 3.1 (3)]{LMSWY} and by Theorem \ref{criterion}, $A>0,$ $B<0$ and  $B^2-4AC=16\mathrm{Gram}(l_1,\dots,l_6)>0.$ Hence 
$z_1^*=\frac{-B+\sqrt{B^2-4AC}}{2A}>0$ and $z_2^*=\frac{-B-\sqrt{B^2-4AC}}{2A}>0.$ 
Then $\xi_1^*$ and $\xi_2^*$ are respectively  the complex numbers with $\mathrm{Re}\xi_1^*=\mathrm{Re}\xi_2^*=2\pi$  satisfying 
$e^{-2\mathbf i\xi^*_1}=z_1^*\quad\text{and}\quad e^{-2\mathbf i{\xi^*_2}}=z_2^*.$

\begin{theorem}\label{ab} Let $\Delta$  be a truncated hyperideal tetrahedron in $\ads^3$ with edge lengths $(l_1,\dots,l_6)$. Let   $U_{\boldsymbol\alpha}$ and  $\xi_1^*$ be as in Proposition \ref{Prop3.1}, and let 
$$W(l_1,\dots,l_6)=U_{\boldsymbol\alpha}(\xi_1^*).$$
Then 
\begin{enumerate}[(1)]
\item $\mathrm{Re}W(l_1,\dots,l_6)= {2\mathrm{Cov}(\Delta)},$ the anti-de Sitter co-volume of $\Delta,$  and 
\item $\mathrm{Im}W(l_1,\dots,l_6)=-2\widetilde{\mathrm{Cov}}(l_1,\dots,l_6),$ the extended co-volume of $(l_1,\dots,l_6)$ defined in \cite{LY}.
\end{enumerate}
\end{theorem}

\begin{remark}\label{3.4}
By Remark \ref{ccov}, Theorem \ref{ab} can be re-stated in terms of the genuine co-volume as 
$$W(l_1,\dots,l_6)=-2\mathbf i\mathrm{Cov}_{\mathbb C}(\Delta).$$
Compare with \cite[Proposition 3.1 (1)]{LMSWY}, Proposition \ref{critical1} and Remark \ref{3.20}.
\end{remark}

As a direct consequence of Theorem \ref{ab} and Proposition \ref{CoSch}, the Schl\"afli formula, we have the following volume formula of a truncated hyperideal tetrahedron in $\ads^3$ in terms of the edge lengths.

\begin{theorem}\label{volume}
Let $\Delta$ be a truncated hyperideal tetrahedron in $\ads^3$ with edge lengths $(l_1,\dots ,l_6).$ Then
$$\mathrm{Vol}(\Delta)= {\frac{1}{2}}\Bigg(\mathrm{Re}W(l_1,\dots,l_6)-\sum_{k=1}^6\frac{\partial \mathrm{Re}W}{\partial l_k}(l_1,\dots,l_6)\cdot l_k\Bigg).$$
\end{theorem}

\begin{proof}[Proof of Theorem \ref{ab}] 
Let $(\theta_1,\dots,\theta_6)$ be the dihedral angles of $\Delta.$
We first prove that for each $k\in\{1,\dots, 6\},$ 
\begin{equation}\label{co-schlafli}
\frac{\partial W}{\partial l_k}=  {-}\mathbf i \theta_k.
\end{equation}
By the symmetry of $W$ in its variables, it suffices to prove the result for $k=1.$ From \cite[Eq. (3.14)]{LMSWY}, we have 
\begin{equation}\label{pW}
\frac{\partial W(l_1,\dots,l_6)}{\partial l_1}=-\frac{1}{2}\log\frac{G_{34}-\sqrt{G_{34}^2-G_{33}G_{44}}}{G_{34}+\sqrt{G_{34}^2-G_{33}G_{44}}};
\end{equation}
and we are left to show that the right hand side of (\ref{pW}) equals $-\mathbf i\theta_1.$ By Proposition \ref{cos} and (\ref{cos2}),  we have 
\begin{equation*}
 {\cos\theta_{ij}=\frac{G_{kl}}{\sqrt{G_{kk}G_{ll}}}},
\end{equation*}
where $\{k,l\}=\{1,2,3,4\}\setminus\{i,j\},$ which implies that 
$$\frac{G_{34}-\sqrt{G_{34}^2-G_{33}G_{44}}}{G_{34}+\sqrt{G_{34}^2-G_{33}G_{44}}}=\frac{\frac{G_{34}}{\sqrt{G_{33}G_{44}}}-\sqrt{\frac{G_{34}^2}{G_{33}G_{44}}-1}}{\frac{G_{34}}{\sqrt{G_{33}G_{44}}}+\sqrt{\frac{G_{34}^2}{G_{33}G_{44}}-1}}=\frac{\cos \theta_1-\sqrt{\cos^2\theta_1-1}}{\cos \theta_1+\sqrt{\cos^2\theta_1-1}}.$$
Then by a case-by-case discussion for  the case $\theta_1=\mathbf i\varphi$ with $\varphi>0,$ and the case $\theta_1=\pi-\mathbf i\varphi$ with $\varphi>0,$ we have that 
$$\frac{\cos \theta_1-\sqrt{\cos^2\theta_1-1}}{\cos \theta_1+\sqrt{\cos^2\theta_1-1}}=e^{2\mathbf i\theta_1},$$
from which (\ref{co-schlafli}) follows. 

This implies that (1) and (2) hold respectively up to a constant. By checking the values for the edge lengths $(l_1,\dots,l_6)$ of a flat tetrahedron $\Delta,$ we have 
$$\mathrm{Re}W(l_1,\dots,l_6)=0=2\mathrm{Cov}(\Delta),$$
and 
$$\mathrm{Im}W(l_1,\dots,l_6)=-2\mathrm{Cov}(l_1,\dots,l_6)=-2\widetilde{\mathrm{Cov}}(l_1,\dots,l_6).$$
This, together with Proposition \ref{cont},  implies that the constants for (1) and (2) are both $0,$ and the result follows.  
\end{proof}

\begin{proof}[Proof of Theorem \ref{cov}] By \cite[Eq. (3.54)]{LMSWY}, we have 
\begin{equation*}\label{eq:ads}
\bigg\{\begin{matrix} \frac{Q}{2} +  \mathbf i\frac{l_1}{2\pi b} & \frac{Q}{2}+\mathbf i\frac{l_2}{2\pi b} & \frac{Q}{2}+\mathbf i\frac{l_3}{2\pi b} \\ \frac{Q}{2}+\mathbf i\frac{l_4}{2\pi b} & \frac{Q}{2}+ \mathbf i\frac{l_5}{2\pi b} & \frac{Q}{2}+\mathbf i\frac{l_6}{2\pi b} \end{matrix} \bigg\}_b=\frac{e^{\frac{-\widetilde{\mathrm{Cov}}(\boldsymbol l)}{\pi b^2}} }{\sqrt[4]{\det\mathrm{Gram}(\Delta)}} \Bigg(2\cos\bigg(\frac{\mathrm{Re}U_{\boldsymbol\alpha}(\xi_1^*)}{2\pi b^2}+\frac{\pi}{4}\bigg) +O\big(b^2\big)\Bigg);
\end{equation*}
and by Theorem \ref{ab} (1), the result follows. 
\end{proof}

As an immediate consequence of Theorem \ref{cov} and Proposition \ref{reflection}, the reflection symmetry, we have 

\begin{theorem}\label{covpm} Let $(l_1,\dots, l_6)\in\mathbb R_{>0 }^6$ be the edge lengths of a truncated  hyperideal anti-de Sitter tetrahedron $\Delta.$ Then  as $b\to 0,$
\begin{equation*}\label{AdS/CFT}
\bigg\{\begin{matrix} \frac{Q}{2} \pm  \mathbf i\frac{l_1}{2\pi b} & \frac{Q}{2}\pm \mathbf i\frac{l_2}{2\pi b} & \frac{Q}{2}\pm \mathbf i\frac{l_3}{2\pi b} \\ \frac{Q}{2}\pm \mathbf i\frac{l_4}{2\pi b} & \frac{Q}{2}\pm  \mathbf i\frac{l_5}{2\pi b} & \frac{Q}{2}\pm \mathbf i\frac{l_6}{2\pi b} \end{matrix} \bigg\}_b=\frac{e^{\frac{-\widetilde{\mathrm{Cov}}(l_1,\dots,l_6)}{\pi b^2}} }{\sqrt[4]{\det\mathrm{Gram}(\Delta)}} \Bigg(2\cos\bigg( {\frac{\mathrm{Cov}(\Delta)}{\pi b^2}}+\frac{\pi}{4}\bigg) +O\big(b^2\big)\Bigg),
\end{equation*}
where in each entry of the $b$-$6j$ symbol the sign $+$ and $-$ can be chosen arbitrarily, $\widetilde{\mathrm{Cov}}(l_1,\dots,l_6)$ is the extended co-volume function of $(l_1,\dots,l_6)$ defined in \cite{LY}, $\mathrm{Cov}(\Delta)$ is the anti-de Sitter co-volume of $\Delta$, and $\mathrm{Gram}(\Delta)$ is the Gram matrix of $\Delta$ in the edge lengths.
\end{theorem}


\subsection{Asymptotics in the dihedral angles}\label{da}

The goal of this subsection is to prove Theorem \ref{vol3} and Theorem \ref{vol2}. The main tool is the following Saddle Point Approximation.

\begin{proposition}\cite[Proposition 6.1]{WY2}\label{saddle}
Let $D$ be a region in $\mathbb C$ and let $f(z)$ and $g(z)$ be holomorphic functions on $D$. Let $f_\hbar(z)$ be a holomorphic function of the form
$$ f_\hbar(z) = f(z) + \upsilon_\hbar(z)\hbar^2.$$
Let $S$ be a curve in $D$ and let $c$ be a point on $S.$ Suppose
\begin{enumerate}[(i)]
\item $c$ is a critical point of $f$ in $D,$
\item $\mathrm{Re}(f)(c) > \mathrm{Re}(f)(z)$ for all $z \in S\setminus \{c\},$
\item $f''(c)\neq 0$, 
\item $g(c) \neq 0,$ 
\item $|\upsilon_\hbar(z)|$ is bounded from above by a constant independent of $\hbar$ in $D,$ and
\item $S$ is smooth  around $c.$ 
\end{enumerate}
Then
\begin{equation*}
\int_S g(z) e^{\frac{f_\hbar(z)}{\hbar}} dz= \big(2\pi \hbar\big)^{\frac{1}{2}}\frac{g(c)}{\sqrt{-f''(c)}}e^{\frac{f(c)}{\hbar}} \Big(1 + O \big(\hbar\big)\Big).
 \end{equation*}
\end{proposition}

\subsubsection{The hyperbolic case}

Let $(\theta_1,\dots, \theta_6)$ be the dihedral angles of a truncated hyperideal hyperbolic tetrahedron $\Delta;$ 
and for $k\in\{1,\dots,6\}$, let $a_k =\frac{Q}{2} + \frac{\theta_k}{2\pi b}$.
Then by \cite{BB}, for $i,j\in\{1,2,3,4\},$ 
$$\frac{b}{2}< q_j-t_i < Q-\frac{b}{2},$$
and the six-tuple  $(a_1,\dots,a_6)$ is $b$-admissible. By Definition \ref{b6j}, the 
$b$-$6j$ symbol of $(a_1,\dots,a_6)$ is computed by
\begin{equation}\label{fk3}
\bigg\{\begin{matrix} a_1 & a_2 & a_3 \\ a_4 & a_5 & a_6 \end{matrix} \bigg\}_b=\Bigg(\frac{1}{\prod_{i=1}^4\prod_{j=1}^4S_b(q_j-t_i)}\Bigg)^{\frac{1}{2}}\int_\Gamma \prod_{i=1}^4S_b(u-t_i)\prod_{j=1}^4S_b(q_j-u)d u,
\end{equation}
where the contour $\Gamma$ is any vertical line passing the interval $(\max\{t_1,t_2,t_3,t_4\},\min\{q_1,q_2,q_3,q_4\})$. See Figure \ref{Da2}.

\begin{figure}[htbp]
\centering
\includegraphics[scale=0.4]{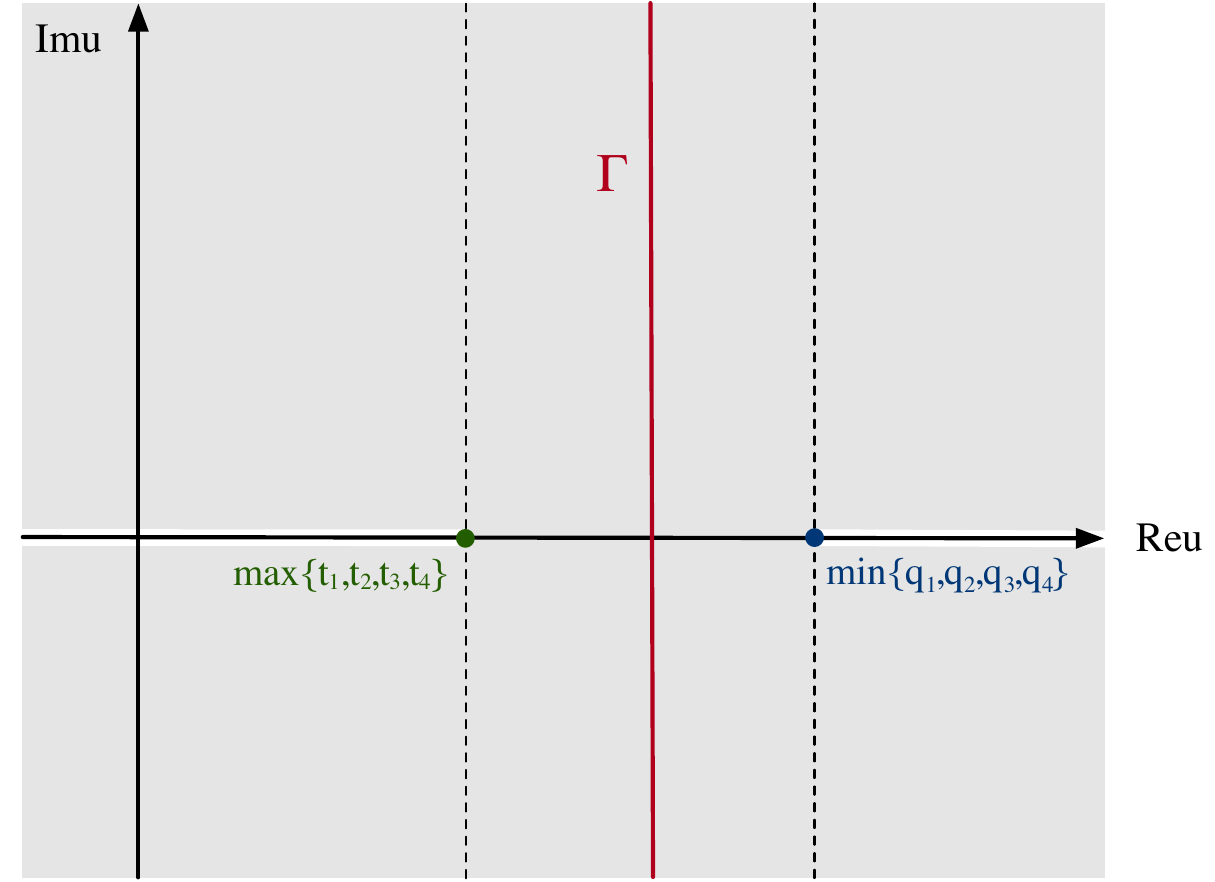}
\caption{Contour $\Gamma$ and possible zeros and poles (located in the white rays) of the integrand in~\eqref{fk3}.}
\label{Da2}
\end{figure}

To prove Theorem \ref{vol3},  we have by  (\ref{6jint}) that 
\begin{equation}
\bigg\{\begin{matrix} a_1 & a_2 & a_3 \\ a_4 & a_5 & a_6 \end{matrix} \bigg\}_b=\frac{1}{\pi b}\int_{\Gamma}\exp\bigg(\frac{U_{\boldsymbol \alpha}(\xi)+\kappa_{\boldsymbol \alpha}(\xi)b^2+
\nu_{\boldsymbol \alpha,b}(\xi)b^4}{2\pi \mathbf ib^2} \bigg)d\xi,
\end{equation}
where $\Gamma$ is a vertical line passing  the real axis in the interval $(\max\{\tau_1,\tau_2,\tau_3,\tau_4\},\min\{\eta_1,\eta_2,\eta_3,\eta_4\})$, and $U_{\boldsymbol\alpha}$, $\kappa_{\boldsymbol\alpha}$ and $\nu_{\boldsymbol\alpha,b}$ are respectively as defined in (\ref{U}), (\ref{kappa}) and (\ref{nuU}).

\begin{proposition}\label{critical1} Let $(\theta_1,\dots, \theta_6)$ be the dihedral angles of a truncated hyperideal hyperbolic  tetrahedron $\Delta$. Then the function $U_{\boldsymbol \alpha}(\xi)$ has a unique critical point $\xi^*$ in the domain
$$D=\Big\{ \xi\in\mathbb C\ \Big|\ \max\{\tau_1,\tau_2,\tau_3,\tau_4\}< \mathrm{Re}\xi< \min\{\eta_1,\eta_2,\eta_3,\eta_4\} \Big\}.$$
Moreover, $\xi^*\in\mathbb R$ and 
$$U_{\boldsymbol \alpha}(\xi^*)=-2\mathbf i\mathrm{Vol}(\Delta).$$
\end{proposition}

\begin{proof} We first show that $U_{\boldsymbol \alpha}$ has a unique critical point $\xi^*$ on the interval $(\min\{\tau_i\},\min\{\eta_j\})$ following the argument in \cite{C}. 

By the functional property of the dilogarithm function that 
$\mathrm{Li}_2\big(e^{2\mathbf i\theta}\big)=\frac{\pi^2}{6}+\theta(\theta-\pi)+2\mathbf i\Lambda(\theta)$  for $\theta\in(0,\pi),$ 
where $\Lambda$ is the Lobachevsky function defined by $\Lambda(\theta)=-\int_{0}^\theta\ln|2\sin t|dt,$ we have for $\theta\in(0,\pi),$
\begin{equation}\label{L}
L(\theta)=-2\mathbf i\Lambda(\theta).
\end{equation}
As a consequence, we have for $\xi\in(\max\{\tau_i\},\min\{\eta_j\}),$ 
$$U_{\boldsymbol \alpha}(\xi)=-2\mathbf i V_{\boldsymbol \alpha}(\xi),$$
where 
$$V_{\boldsymbol \alpha}(\xi)=-\frac{1}{2}\sum_{i=1}^4\sum_{j=1}^4\Lambda(\eta_j-\tau_i)+\sum_{i=1}^4\Lambda(\xi-\tau_i)+\sum_{j=1}^4\Lambda(\eta_j-\xi).$$
As
$$V_{\boldsymbol \alpha}'(\xi)=-\sum_{i=1}^4\ln|2\sin(\xi-\tau_i)|+\sum_{j=1}^4\ln|2\sin(\eta_j-\xi)|,$$
we have
\begin{equation}\label{r}
\lim_{\xi\to\min\{\eta_j\}}V_{\boldsymbol \alpha}'(\xi)=-\infty
\end{equation}
and
\begin{equation}\label{l}
\lim_{\xi\to\max\{\tau_i\}}V_{\boldsymbol \alpha}'(\xi)=+\infty;
\end{equation}
and as 
\begin{equation}\label{convex}
V_{\boldsymbol \alpha}''(\xi)=\sum_{i=1}^4 \big(\cot(\xi-\tau_i)+\cot(\eta_i-\xi)\big)>0,
\end{equation}
we have that $V_{\boldsymbol \alpha}$ is strictly concave up in $(\min\{\tau_i\},\min\{\eta_j\}).$ The inequality in (\ref{convex}) can be seen as follows: for $\alpha, \beta\in(0,\pi)$ with $\alpha+\beta<\pi,$ there is an Euclidean triangle with two of its angles $\alpha$ and $\beta.$ Let $C$ be the lengths of the side between $\alpha$ and $\beta$ and let $h_C$ be the height at this side, then $\cot \alpha+\cot\beta=\frac{h_C}{C}>0.$ Here we let $\alpha=\xi-\tau_i$ and $\beta=\eta_i-\xi,$ then $\alpha+\beta=\eta_i-\tau_i<\pi$ due to the fact that $(\theta_1,\dots,\theta_6)$ are the dihedral angles of a hyperideal tetrahedron.

Putting (\ref{L}), (\ref{r}), (\ref{l}) and (\ref{convex}) together, we have that $U_{\boldsymbol \alpha}$ has a unique critical point $\xi^*$ in the interval $(\min\{\tau_i\},\min\{\eta_j\}).$ 

Then Proposition \ref{decreasing} below implies that  $\xi^*$ is the unique critical point of $U_{\boldsymbol \alpha}$ on the entire domain $D.$

Finally, by (\ref{L}) and a direct computation, $\mathrm{Re}U_{\boldsymbol \alpha}(\xi^*)=0;$ and by  \cite{C}, $\mathrm{Im}U_{\boldsymbol \alpha}(\xi^*)=-2V_{\boldsymbol\alpha}(\xi^*)=-2\mathrm{Vol}(\Delta).$
 \end{proof}

\begin{proposition}\label{decreasing}
On the domain $D,$ $\mathrm{Im}U_{\boldsymbol \alpha}(\xi)$ is increasing in  $\mathrm{Im}\xi$ when $\mathrm{Im}\xi<0,$ and is decreasing in  $\mathrm{Im}\xi$ when $\mathrm{Im}\xi>0.$ As a consequence, $\xi^*$ achieves the maximum of $\mathrm{Im}U_{\boldsymbol \alpha}$ on the contour 
$$\Gamma^*=\Big\{\xi\in D\ \Big|\ \mathrm{Re}\xi = \mathrm{Re}\xi^*  \Big\}.$$ 
\end{proposition}

The proof of Proposition \ref{decreasing} relies on the following lemmas.

\begin{lemma}\label{d} Let $L$ be the function defined in (\ref{eq:Lx}). Then
\begin{enumerate}[(1)] 

\item for $\alpha,\beta\in(0,\pi),$ we have $$ \frac{\partial}{\partial l}\Big|_{l=0}\Big( \mathrm{Im}L(\alpha+\mathbf il)+ \mathrm{Im}L(\beta-\mathbf il)\Big)=0;$$

\item   for $\alpha,\beta\in(0,\pi)$ with $\alpha+\beta<\pi,$  we have $$\lim_{l\to +\infty} \frac{\partial}{\partial l}\Big( \mathrm{Im}L(\alpha+\mathbf il) +\mathrm{Im}L(\beta-\mathbf il)\Big)<0,$$
and
  $$ \lim_{l\to -\infty} \frac{\partial }{\partial l}\Big(\mathrm{Im}L(\alpha+\mathbf il)+\mathrm{Im}L(\beta-\mathbf il)\Big)>0.$$

\end{enumerate}
\end{lemma}

\begin{proof} By a direct computation, we have
\begin{equation}\label{limp}
\frac{\partial}{\partial l} \mathrm{Im}L(\alpha+\mathbf il)=2\alpha-\pi-2\arg\Big(1-e^{-2l+2\mathbf i\alpha}\Big),
\end{equation}
and
$$ \frac{\partial}{\partial l} \mathrm{Im}L(\alpha-\mathbf il)=-2\alpha+\pi+2\arg\Big(1-e^{2l+2\mathbf i\alpha}\Big).$$

For (1), we have 
$$ \frac{\partial}{\partial l}\Big|_{l=0} \mathrm{Im}L(\alpha+\mathbf il)=2\alpha-\pi-2\arg\Big(1-e^{2\mathbf i\alpha}\Big)=0$$
and
$$ \frac{\partial}{\partial l}\Big|_{l=0} \mathrm{Im}L(\alpha-\mathbf il)=-2\alpha+\pi+2\arg\Big(1-e^{2\mathbf i\alpha}\Big)=0.$$
Therefore, 
$$ \frac{\partial}{\partial l}\Big|_{l=0}\Big( \mathrm{Im}L(\alpha+\mathbf il)+ \mathrm{Im}L(\beta-\mathbf il)\Big)=0,$$
which proves (1).
\medskip

For (2), we have
$$\lim_{l\to +\infty} \frac{\partial}{\partial l} \mathrm{Im}L(\alpha+\mathbf il)=2\alpha-\pi-2\arg(1)=2\alpha-\pi,$$
$$\lim_{l\to -\infty} \frac{\partial}{\partial l} \mathrm{Im}L(\alpha+\mathbf il)=2\alpha-\pi-2\arg(-e^{2\mathbf i\alpha})=\pi-2\alpha,$$
$$\lim_{l\to +\infty} \frac{\partial}{\partial l} \mathrm{Im}L(\alpha-\mathbf il)=-2\alpha+\pi+2\arg(-e^{2\mathbf i\alpha})=2\alpha-\pi$$
and
$$\lim_{l\to -\infty} \frac{\partial}{\partial l} \mathrm{Im}L(\alpha-\mathbf il)=-2\alpha+\pi+2\arg(1)=-2\alpha+\pi.$$
As a consequence, if $\alpha+\beta<\pi,$ then
\begin{equation}\label{lim+}
\lim_{l\to +\infty} \frac{\partial}{\partial l}\Big( \mathrm{Im}L(\alpha+\mathbf il)+ \mathrm{Im}L(\beta-\mathbf il)\Big)=2(\alpha+\beta)-2\pi<0
\end{equation}
and
\begin{equation}\label{lim-}
\lim_{l\to -\infty} \frac{\partial}{\partial l}\Big( \mathrm{Im}L(\alpha+\mathbf il) +\mathrm{Im}L(\beta-\mathbf il)\Big)=2\pi-2(\alpha+\beta)>0,
\end{equation}
which proves (2).
\end{proof}

\begin{lemma}\label{dd} Let $\alpha,\beta\in (0,\pi)$ such that $\alpha+\beta<\pi.$ 
\begin{enumerate}[(1)]
\item If $|\alpha-\beta|\leqslant \frac{\pi}{2},$ then
 $$\frac{\partial ^2}{\partial l^2}\Big(\mathrm{Im}L(\alpha+\mathbf il)+\mathrm{Im}L(\beta-\mathbf il)\Big) <0.$$

\item  If $|\alpha-\beta|>\frac{\pi}{2},$ then there is a unique $l_0>0$ such that:
\begin{enumerate}[(i)]
\item For $-l_0<l<l_0,$ 
$$\frac{\partial ^2}{\partial l^2}\Big(\mathrm{Im}L(\alpha+\mathbf il)+\mathrm{Im}L(\beta-\mathbf il)\Big)<0.$$ 

\item For $l<-l_0\ \text{or}\ l>l_0,$ 
$$\frac{\partial ^2}{\partial l^2}\Big(\mathrm{Im}L(\alpha+\mathbf il)+\mathrm{Im}L(\beta-\mathbf il)\Big)>0.$$
\end{enumerate}

\end{enumerate}
\end{lemma}

\begin{proof}  Denote
\begin{equation*}
\phi_{\alpha\beta}(l)=\frac{\partial ^2}{\partial l^2}\Big(\mathrm{Im}L(\alpha+\mathbf il)+\mathrm{Im}L(\beta-\mathbf  il)\Big).
\end{equation*}By a direct computation, we have
$$\phi_{\alpha\beta}(l)=-\frac{4\sin(2\alpha)}{\cosh^2l-\cos^2\alpha}-\frac{4\sin(2\beta)}{\cosh^2l-\cos^2\beta},$$
and
$$\phi_{\alpha\beta}(0)=-8(\cot \alpha+\cot\beta)<0$$
as argued after (\ref{convex}).

Now by a further direct computation,
$$\phi_{\alpha\beta}(l)=\frac{8\sin(\alpha+\beta)\big(-\cos(\alpha-\beta)\cosh^2 l +\cos\alpha\cos\beta\big)}{(\cosh^2l-\cos^2\alpha)(\cosh^2l-\cos^2\beta)}.$$
We observe that the denominator and the scalar $8\sin(\alpha+\beta)$ in the numerator are always positive; and in the numerator we notice that: in case (1) that $|\alpha-\beta|\leqslant \frac{\pi}{2},$  the scalar $-\cos(\alpha-\beta)$ in front of $\cosh^2l$ is negative or zero, and in case (2) that  $|\alpha-\beta|>\frac{\pi}{2},$ the scalar $-\cos(\alpha-\beta)$ in front of $\cosh^2l$  is positive. As a consequence, in case (1) the value $\phi_{\alpha\beta}(l)$ stays negative for all $l\in\mathbb R,$ and in case (2) the value $\phi_{\alpha\beta}(l)$ stays negative when
$$|l|<l_0=\cosh^{-1}\bigg(\frac{\cos\alpha\cos\beta}{\cos(\alpha-\beta)}\bigg)^{\frac{1}{2}},$$
and becomes positive when $|l|>l_0.$
\end{proof}

\begin{proof}[Proof of Proposition \ref{decreasing}] For each $i\in\{1,2,3,4\},$ on each straight line in $D$ consisting of $\xi$  with a  fixed $\mathrm{Re}\xi,$ let 
$$\psi_i(\mathrm{Im}\xi)=\mathrm{Im}L(\xi-\tau_i)+\mathrm{Im}L(\eta_i-\xi).$$
Then
$$\mathrm{Im}U_{\boldsymbol \alpha}(\xi)=\sum_{i=1}^4\psi_i(\mathrm{Im}\xi),$$
 and it suffices to show that for each $i,$ 
$\psi_i$ is increasing for $\mathrm{Im}\xi<0,$ and  is decreasing for  $\mathrm{Im}\xi>0.$

Now for each $i,$ we have $\mathrm{Re}(\xi-\tau_i)+\mathrm{Re}(\eta_i-\xi)=\eta_i-\tau_i<\pi.$ Then by Lemma \ref{d} (1) and Lemma \ref{dd} with $\alpha=\mathrm{Re}(\xi-\tau_i)$ and $\beta=\mathrm{Re}(\eta_i-\xi),$ we have $\psi_i'(0)=0$ and $\psi_i''(\mathrm{Im}\xi)<0$ for $\mathrm{Im}\xi\in [-l_0,l_0],$ with the understanding that $l_0=+\infty$ if $|\alpha-\beta|\leqslant \frac{\pi}{2}.$ Therefore, $\psi_i'(\mathrm{Im}\xi)>0$ for $\mathrm{Im}\xi\in [-l_0,0)$ and $\psi_i'(\mathrm{Im}\xi)<0$ for $\mathrm{Im}\xi\in (0, l_0].$  

Next, by Lemma \ref{d} (2) and Lemma \ref{dd}, 
$\psi_i'(\mathrm{Im}\xi)>\lim_{\mathrm{Im}\xi\to-\infty}\psi_i'(\mathrm{Im}\xi)>0$ for $\mathrm{Im}\xi<-l_0;$  and 
$\psi_i'(\mathrm{Im}\xi)<\lim_{\mathrm{Im}\xi\to+\infty}\psi_i'(\mathrm{Im}\xi)<0$ for $\mathrm{Im}\xi>l_0.$ 

Putting all together, we have $\psi_i'(\mathrm{Im}\xi)>0$ and $\psi_i$ is increasing for $\mathrm{Im}\xi<0,$ and $\psi_i'(\mathrm{Im}\xi)<0$ and $\psi_i$ is decreasing for $\mathrm{Im}\xi>0.$ 
\end{proof}

\begin{proposition} \label{limder} For $c\in[\max\{\tau_1,\tau_2,\tau_3,\tau_4\},\min\{\eta_1,\eta_2,\eta_3,\eta_4\}],$ let 
$$\Gamma_c=\Big\{\xi\in \overline D\ \Big|\ \mathrm{Re}\xi =c  \Big\}.$$ Then on $\Gamma_c,$ we have
\item $$\lim_{\mathrm{Im}\xi\to +\infty}\frac{\partial \mathrm{Im}U_{\boldsymbol \alpha}(\xi)}{\partial \mathrm{Im}\xi} = -4\pi\quad\text{and}\quad \lim_{\mathrm{Im}\xi\to -\infty}\frac{\partial \mathrm{Im}U_{\boldsymbol \alpha}(\xi)}{\partial \mathrm{Im}\xi} = 4\pi.$$
\end{proposition}
\begin{proof}
By (\ref{lim+}) and (\ref{lim-}), we have
$$\lim_{\mathrm{Im}\xi\to +\infty}\frac{\partial \mathrm{Im}U_{\boldsymbol \alpha}(\xi)}{\partial \mathrm{Im}\xi} = \sum_{i=1}^4 \Big(2\big((c- \tau_i)-(\eta_i-c)\big)-2\pi\Big) =-4\pi,$$
and 
$$\lim_{\mathrm{Im}\xi\to -\infty}\frac{\partial \mathrm{Im}U_{\boldsymbol \alpha}(\xi)}{\partial \mathrm{Im}\xi} = \sum_{i=1}^4 \Big(2\pi-2\big((c-\tau_i)-(\eta_i-c)\big)\Big) =4\pi.$$
\end{proof}

 \begin{proposition}\label{Hess} At the critical point  $\xi^*$ of $U_{\boldsymbol\alpha}$, we have
 \begin{equation*}
\frac{-U''_{\boldsymbol \alpha}(\xi^*)}{\exp\big(\frac{\kappa_{\boldsymbol \alpha}(\xi^*)}{\pi \mathbf i}\big)}=16\sqrt{\det\mathrm{Gram}(\Delta)}.
\end{equation*}
As a consequence, $\xi^*$ is a non-degenerate critical point of $U_{\boldsymbol\alpha}.$
\end{proposition}

 \begin{proof} The proof follows verbatim from that of \cite[Proposition 3.2]{LMSWY} adapting the argument from that of \cite[Lemma 3]{CM}. Namely, from $U'_{\boldsymbol \alpha}(\xi^*)=0,$
we have 
$$\sum_{i=1}^4\log\big(1-e^{2\mathbf i(\xi^*-\tau_i)}\big)=\sum_{j=1}^4\log\big(1-e^{2\mathbf i(\eta_j-\xi^*)}\big)-8 \mathbf i\xi^*+4\mathbf i\sum_{k=1}^6\alpha_k+2\pi \mathbf i;$$
and plugging in  (\ref{kappa}), we have
$$\kappa_{\boldsymbol \alpha}(\xi^*)=2\pi^2+2\pi\sum_{k=1}^6\alpha_k-4\pi\xi^*-\pi \mathbf i\sum_{i=1}^4\log\big(1-e^{2\mathbf i(\xi^*-\tau_i)}\big).$$
For $k\in\{1,\dots,6\},$ let $u_k=e^{2\mathbf i\alpha_k}.$ Consider the following quadratic equation 
\begin{equation*}
Az^2+Bz+C=0
\end{equation*}
with \begin{equation*}
\begin{split}
A=&u_1u_4+u_2u_5+u_3u_6-u_1u_2u_6-u_1u_3u_5-u_2u_3u_4-u_4u_5u_6+u_1u_2u_3u_4u_5u_6,\\
B=&-\Big(u_1-\frac{1}{u_1}\Big)\Big(u_4-\frac{1}{u_4}\Big)-\Big(u_2-\frac{1}{u_2}\Big)\Big(u_5-\frac{1}{u_5}\Big)-\Big(u_3-\frac{1}{u_3}\Big)\Big(u_6-\frac{1}{u_6}\Big),\\
C=&\frac{1}{u_1u_4}+\frac{1}{u_2u_5}+\frac{1}{u_3u_6}-\frac{1}{u_1u_2u_6}-\frac{1}{u_1u_3u_5}-\frac{1}{u_2u_3u_4}-\frac{1}{u_4u_5u_6}+\frac{1}{u_1u_2u_3u_4u_5u_6},\\
\end{split}
\end{equation*}
and let 
\begin{equation*}
z^*=\frac{-B+\sqrt{B^2-4AC}}{2A}\quad\text{and}\quad z^{**}=\frac{-B-\sqrt{B^2-4AC}}{2A}
\end{equation*}
be the two roots.
Then $z^*=e^{-2\mathbf i\xi^*},$ and
\begin{equation*}
\begin{split}
-U''_{\boldsymbol \alpha}(\xi^*)=4\bigg(&\frac{z^*}{1-z^*}+\frac{z^*u_1u_2u_4u_5}{1-z^*u_1u_2u_4u_5}+\frac{z^*u_1u_3u_4u_6}{1-z^*u_1u_3u_4u_6}+\frac{z^*u_2u_3u_5u_6}{1-z^*u_2u_3u_5u_6}\\
& -\frac{z^*u_1u_2u_3}{1-z^*u_1u_2u_3}-\frac{z^*u_1u_5u_6}{1-z^*u_1u_5u_6}-\frac{z^*u_2u_4u_6}{1-z^*u_2u_4u_6}-\frac{z^*u_3u_4u_5}{1-z^*u_3u_4u_5}\bigg).
\end{split}
\end{equation*}
As a consequence, 
\begin{equation*}
\begin{split}
\frac{-U''_{\boldsymbol \alpha}(\xi^*)}{\exp\big({\frac{\kappa_{\boldsymbol \alpha}(\xi^*)}{\pi \mathbf i}}\big)}=&\frac{4(1-z^*u_1u_2u_3)(1-z^*u_1u_5u_6)(1-z^*u_2u_4u_6)(1-z^*u_3u_4u_5)}{z^{*2}u_1u_2u_3u_4u_5u_6}\\
&\bigg(\frac{z^*}{1-z^*}+\frac{z^*u_1u_2u_4u_5}{1-z^*u_1u_2u_4u_5}+\frac{z^*u_1u_3u_4u_6}{1-z^*u_1u_3u_4u_6}+\frac{z^*u_2u_3u_5u_6}{1-z^*u_2u_3u_5u_6}\\
&\quad-\frac{z^*u_1u_2u_3}{1-z^*u_1u_2u_3}-\frac{z^*u_1u_5u_6}{1-z^*u_1u_5u_6}-\frac{z^*u_2u_4u_6}{1-z^*u_2u_4u_6}-\frac{z^*u_3u_4u_5}{1-z^*u_3u_4u_5}\bigg);
\end{split}
\end{equation*}
and by a direct computation, this equals
$$4\bigg(3Az^*+2B+\frac{C}{z^*}\bigg)=4\bigg(Az^*-\frac{C}{z^*}\bigg),$$
which in turn equals
$$4A(z^*-{z^{**}})=4\sqrt{B^2-4AC}=16\sqrt{\det\mathrm{Gram}(\Delta)},$$
where the last inequality comes from the algebraic computation in \cite{U} as one can write 
$$
\mathrm{Gram}(\Delta)=\begin{bmatrix}
1 & \frac{u_1+u_1^{-1}}{2} & \frac{u_2+u_2^{-1}}{2}
    & \frac{u_6+u_6^{-1}}{2}\\
\frac{u_1+u_1^{-1}}{2} & 1 & \frac{u_3+u_3^{-1}}{2}
    & \frac{u_5+u_5^{-1}}{2}\\
   \frac{u_2+u_2^{-1}}{2} & \frac{u_3+u_3^{-1}}{2} & 1
    & \frac{u_4+u_4^{-1}}{2}\\
 \frac{u_6+u_6^{-1}}{2} & \frac{u_5+u_5^{-1}}{2} & \frac{u_4+u_4^{-1}}{2}
    & 1
  \end{bmatrix}.
$$
As also proved in \cite{U}, $\mathrm{Gram}(\Delta)$ has signature $(3,1),$ hence $$U''_{\boldsymbol\alpha}(\xi^*) = -16\exp\bigg({\frac{\kappa_{\boldsymbol \alpha}(\xi^*)}{\pi \mathbf i}}\bigg)\sqrt{\det\mathrm{Gram}(\Delta)}\neq 0,$$ and $\xi^*$ is a non-degenerate critical point of $U_{\boldsymbol \alpha}.$
 \end{proof}

For a sufficiently small $\delta >0$, consider the  region $$D_{\delta}=\Big\{\xi\in \mathbb C\ \Big|\ \mathrm {Re}\xi\in[\max\{\tau_1,\tau_2,\tau_3,\tau_4\}+\delta, \min\{\eta_1,\eta_2,\eta_3,\eta_4\}-\delta]\Big\}.$$ Then we have the following Proposition \ref{bound} whose proof follows verbatim that of \cite[Proposition 3.10 and Proposition 3.11]{LMSWY}. 

 \begin{proposition}\label{bound} \begin{enumerate}[(1)]
 \item 
 For $\delta> 0$  sufficiently small, there exists a constant $K=K_{\delta}>0$ such that 
$$\Bigg|\frac{\partial \mathrm{Im}\kappa_{\boldsymbol\alpha}(\xi)}{\partial\mathrm{Im}\xi}\Bigg|<K$$
for all $\xi\in D_{\delta}.$
\item  
For $b>0$ and $\delta> 0$ both sufficiently small, there exists a constant $N=N_{\delta}>0$ independent of $b$ such that 
$$\mathrm{Im}\nu_{\boldsymbol\alpha,b}(\xi) \leqslant \big|\nu_{\boldsymbol\alpha,b}(\xi)\big|<N$$
for all $\xi\in D_{\delta}.$
\end{enumerate}
\end{proposition}

\begin{proof}[Proof of Theorem \ref{vol3}] Recall the domain $D=\Big\{ \xi\in\mathbb C\ \Big|\ \max\{\tau_1,\tau_2,\tau_3,\tau_4\}< \mathrm{Re}\xi < \min\{\eta_1,\eta_2,\eta_3,\eta_4\} \Big\}$. Let  $d>0$ be sufficiently small so that the region 
 $$B_{d}=\Big\{\xi\in \mathbb C \ \Big|\ |\mathrm{Re}\xi - \mathrm{Re}\xi^*| \leqslant d \text{ and } |\mathrm{Im}\xi-\mathrm{Im}\xi^*|\leqslant d \Big\}$$
 lies entirely in $D.$ By Propositions \ref{critical1}, \ref{decreasing} and \ref{limder}, there is an $\epsilon_1>0$ such that  
 \begin{equation}\label{ImU}
\mathrm{Im} U_{\boldsymbol\alpha}(\xi^*\pm \mathbf il )< -2\mathrm{Vol}(\Delta)-3\epsilon_1
\end{equation}
and for $l>d;$ and by Propositions \ref{limder},  there is an $L>d$ such that 
\begin{equation}\label{Dv1}
 \frac{\partial \mathrm{Im} U_{\boldsymbol\alpha}}{\partial\mathrm{Im}\xi} (\xi^*+  \mathbf i l)<-2\pi\quad\text{and}\quad  \frac{\partial \mathrm{Im} U_{\boldsymbol\alpha}}{\partial\mathrm{Im}\xi} (\xi^* - \mathbf i l)>2\pi
\end{equation}
for $l>L.$  Consider the contour 
$$\Gamma^*=\Big\{\xi\in D\ \Big|\ \mathrm{Re}\xi = \mathrm{Re}\xi^*  \Big\}.$$
Let
$$\Gamma^*_d=\Gamma^*\cap B_d=\Big\{ \xi \in \Gamma^*\ \Big|\ |\mathrm{Im}\xi-\mathrm{Im}\xi^*|<d\Big\}$$
and let 
$$\Gamma^*_L=\Big\{ \xi \in \Gamma^*\ \Big|\ |\mathrm{Im}\xi-\mathrm{Im}\xi^*|\leqslant L\Big\}.$$
We will show that, as $b\to 0,$
\begin{enumerate}[(I)]
\item 
$$\frac{1}{\pi b}\int_{\Gamma^*_d}\exp\bigg(\frac{U_{\boldsymbol \alpha}(\xi)+ \kappa_{\boldsymbol \alpha}(\xi)b^2+ \nu_{\boldsymbol \alpha,b}(\xi)b^4}{2\pi \mathbf i b^2} \bigg)d\xi =  \frac{e^{\frac{-\mathrm{Vol}(\Delta)}{\pi b^2}}}{\sqrt[4]{-\det\mathrm{Gram}(\Delta)}}\Big(1+O\big(b^2\big)\Big),
$$

\item $$\bigg|\frac{1}{\pi b}\int_{\Gamma^*_L\setminus \Gamma^*_d}\exp\bigg(\frac{U_{\boldsymbol \alpha}(\xi)+ \kappa_{\boldsymbol \alpha}(\xi)b^2+ \nu_{\boldsymbol \alpha,b}(\xi)b^4}{2\pi \mathbf i b^2} \bigg)d\xi\bigg|< O\Big(e^{\frac{-\mathrm{Vol}(\Delta)-\epsilon}{\pi b^2}}\Big),$$
and 

\item $$\bigg|\frac{1}{\pi b}\int_{\Gamma^*\setminus \Gamma^*_L}\exp\bigg(\frac{U_{\boldsymbol \alpha}(\xi)+ \kappa_{\boldsymbol \alpha}(\xi)b^2+ \nu_{\boldsymbol \alpha,b}(\xi)b^4}{2\pi \mathbf i b^2} \bigg)d\xi\bigg|<O\Big(e^{\frac{-\mathrm{Vol}(\Delta)-\epsilon}{\pi b^2}}\Big)$$
\end{enumerate}
for some $\epsilon>0,$ from which the result follows.
\medskip

For (I), we claim that  all the conditions of Proposition \ref{saddle} are satisfied by letting $\hbar=b^2,$ $D=B_d,$ $f=\frac{U_{\boldsymbol \alpha}}{2\pi \mathbf i},$ $g=\exp\big(\frac{\kappa_{\boldsymbol \alpha}}{2\pi \mathbf i }\big),$ $f_\hbar=\frac{U_{\boldsymbol \alpha}+\nu_{\boldsymbol \alpha,b}b^4}{2\pi \mathbf i},$ $\upsilon_h=\frac{\nu_{\boldsymbol \alpha,b}}{2\pi \mathbf i},$ $S=\Gamma^*_d$ and $c=\xi^*.$ 

Indeed, by Proposition \ref{critical1}, $\xi^*$ is a critical point of $f=\frac{U_{\boldsymbol \alpha}}{2\pi \mathbf i}$ in $B_{d},$ hence condition (i) is satisfied. 

By Propositions \ref{decreasing}, $\xi^*$ is the unique maximum point of $\mathrm{Re}f=\frac{\mathrm{Im}U_{\boldsymbol\alpha}}{2\pi}$ on $\Gamma^*_d,$ hence condition (ii) is satisfied; and by Proposition \ref{Hess}, condition (iii) is satisfied.

For condition (iv), since $\xi^*\neq \tau_i$ and $\xi^*\neq\eta_j$ for any $i$ and $j$ in $\{1,2,3,4\},$ $\kappa_{\boldsymbol \alpha}(\xi^*)$ is a finite value. As a consequence, $g(\xi^*)=\exp\big(\frac{\kappa_{\boldsymbol \alpha}(\xi^*)}{2\pi \mathbf i }\big)\neq 0,$ and condition (iv) is satisfied.

For condition (v), by Proposition \ref{bound} (2),  $|\upsilon_{\hbar}(\xi)|=\big|\frac{\nu_{\boldsymbol \alpha,b}(\xi)}{2\pi \mathbf i}\big|<\frac{N}{2\pi}$ on $B_{d}.$ 

For condition (vi), since $\Gamma^*$ is a straight line, it is smooth near $\xi^*.$

Finally, by Proposition \ref{saddle}, Proposition \ref{critical1}  and Proposition \ref{Hess}, we have as $b\to 0,$
\begin{equation*}
\begin{split}
\frac{1}{\pi b}\int_{\Gamma^*_d}\exp\bigg(\frac{U_{\boldsymbol \alpha,b}(\xi)}{2\pi \mathbf ib^2}\bigg) d\xi=& \frac{(2\pi b^2)^\frac{1}{2}}{\pi b} \frac{\exp\big(\frac{\kappa_{\boldsymbol \alpha}(\xi^*)}{2\pi \mathbf i}\big)}{\sqrt{-\frac{U_{\boldsymbol \alpha}''(\xi^*)}{2\pi \mathbf i}}}e^{\frac{U_{\boldsymbol\alpha}(\xi^*)}{2\pi \mathbf i b^2}}\Big(1+O\big(b^2\big)\Big)\\
=&\frac{e^{\frac{-\mathrm{Vol}(\Delta)}{\pi b^2}}}{\sqrt[4]{-\det\mathrm{Gram}(\Delta)}}\Big(1+O\big(b^2\big)\Big).
\end{split}
\end{equation*}
This completes the proof of (I). 
\medskip

For (II) and (III), we have
\begin{equation*}
\begin{split}
& \bigg|\frac{1}{\pi b}\int_{\Gamma^*\setminus \Gamma^*_L}\exp\bigg(\frac{U_{\boldsymbol \alpha}(\xi)+ \kappa_{\boldsymbol \alpha}(\xi)b^2+ \nu_{\boldsymbol \alpha,b}(\xi)b^4}{2\pi \mathbf i b^2} \bigg)d\xi\bigg|\\
\leqslant & \frac{1}{\pi b}\int_{\Gamma^*\setminus \Gamma^*_L}\exp\bigg(\frac{\mathrm{Im}U_{\boldsymbol \alpha}(\xi)+\mathrm{Im}\kappa_{\boldsymbol \alpha}(\xi)b^2+\mathrm{Im}\nu_{\boldsymbol \alpha,b}(\xi)b^4}{2\pi b^2} \bigg)|d\xi|.
\end{split}
\end{equation*}

For  (II), let $\epsilon_1$ be as in (\ref{ImU}). Then by Proposition \ref{bound},  there is a  $b_1>0$ such that 
\begin{equation}\label{last2}
\mathrm{Im}\nu_{\boldsymbol\alpha,b}(\xi)b^4<Nb^4<\epsilon_1 
\end{equation}
for all $b<b_1$ and for all $\xi\in\Gamma^*;$  and together with  (\ref{ImU}), we have
\begin{equation}\label{last3}
\mathrm{Im}U_{\boldsymbol \alpha}(\xi)+\mathrm{Im}  \nu_{\boldsymbol \alpha,b}(\xi)b^4<  -2\mathrm{Vol}(\Delta) - 2\epsilon_1
\end{equation}
for all $b<b_1$ and  $\xi\in\Gamma^*\setminus \Gamma^*_d.$ By the compactness of ${\Gamma^*_L}\setminus \Gamma^*_d,$ there exists an $M>0$ such that 
\begin{equation}\label{Imk}
\mathrm{Im}\kappa_{\boldsymbol\alpha}(\xi)<M
\end{equation}
for all $\xi\in\Gamma^*_L\setminus \Gamma^*_d.$ As a consequence of (\ref{last3}) and (\ref{Imk}), we have
\begin{equation*}
\begin{split}
&\frac{1}{\pi b}\int_{\Gamma^*\setminus \Gamma^*_L}\exp\bigg(\frac{\mathrm{Im}U_{\boldsymbol \alpha}(\xi)+\mathrm{Im}\kappa_{\boldsymbol \alpha}(\xi)b^2+\mathrm{Im}\nu_{\boldsymbol \alpha,b}(\xi)b^4}{2\pi b^2} \bigg)|d\xi|\\
<  & \frac{2(L-d) e^{\frac{M}{2\pi}}}{\pi b} \exp\bigg(\frac{-\mathrm{Vol}(\Delta)-\epsilon_1}{\pi b^2}   \bigg )< O\Big(e^{\frac{-\mathrm{Vol}(\Delta)-\epsilon}{\pi b^2}}\Big)
\end{split}
\end{equation*}
for any $\epsilon<\epsilon_1.$ This completes the proof of (II). 
\smallskip

For (III), let $b_1$ be as in the proof of (II) above and still let $\epsilon_1$ be as in (\ref{ImU}). Then there  is a $b_0\in (0, b_1)$ such that for all $b<b_0,$
\begin{equation}\label{ImK}
\mathrm{Im}\kappa_{\boldsymbol\alpha}(\xi^*\pm \mathbf i L) b^2<\epsilon_1,
\end{equation}
$Kb^2<\epsilon_1$ and $Nb^2<\epsilon_1,$ where $K$ and $N$ are respectively the constants in Proposition \ref{bound}. We claim that, for $\xi\in\Gamma^*\setminus \Gamma^*_L$ and $b<b_0,$ 
\begin{equation}\label{cl}
\mathrm{Im}U_{\boldsymbol \alpha}(\xi)+\mathrm{Im}\kappa_{\boldsymbol \alpha}(\xi)b^2+\mathrm{Im}\nu_{\boldsymbol \alpha,b}(\xi)b^4<-2\big(\mathrm{Vol}(\Delta)+\epsilon_1\big)-(2\pi-\epsilon_1)\big(|\xi-\xi^*|-L\big),
\end{equation}
as a consequence of which we have
\begin{equation}\label{CI}
\begin{split}
& \frac{1}{\pi b}\int_{\Gamma^*\setminus \Gamma^*_L}\exp\bigg(\frac{\mathrm{Im}U_{\boldsymbol \alpha}(\xi)+\mathrm{Im}\kappa_{\boldsymbol \alpha}(\xi)b^2+\mathrm{Im}\nu_{\boldsymbol \alpha,b}(\xi)b^4}{2\pi b^2} \bigg)|d\xi| \\
 < & \frac{1}{\pi b} \exp\bigg(\frac{-\mathrm{Vol}(\Delta)-\epsilon_1}{\pi b^2}\bigg)\int_{\Gamma^*\setminus \Gamma^*_L}\exp\bigg(\frac{-(2\pi-\epsilon_1)\big(|\xi-\xi^*|-L\big)}{2\pi}\bigg) |d\xi|\\
< & O\Big(e^{\frac{-\mathrm{Vol}(\Delta)-\epsilon}{\pi b^2}}\Big)
\end{split}
\end{equation}
for any $\epsilon<\epsilon_1.$ 

For the proof of the claim, by (\ref{Dv1}), Proposition \ref{bound} (1) and the choice of $b_0,$ for $l>L,$ we have 
$$\frac{\partial}{\partial \mathrm{Im}\xi} \Big(\mathrm{Im}U_{\boldsymbol\alpha}(\xi^*+\mathbf il)+\mathrm{Im}\kappa_{\boldsymbol\alpha}(\xi^*+\mathbf il)b^2\Big)<-2\pi+Kb^2<-2\pi+\epsilon_1,$$
and 
$$\frac{\partial}{\partial \mathrm{Im}\xi} \Big(\mathrm{Im}U_{\boldsymbol\alpha}(\xi^*-\mathbf il)+\mathrm{Im}\kappa_{\boldsymbol\alpha}(\xi^*-\mathbf il)b^2\Big)>2\pi-Kb^2>2\pi-\epsilon_1.$$
Together with the Mean Value Theorem, (\ref{ImU}) and (\ref{ImK}), we have 
\begin{equation}\label{Bou}
\begin{split}
\mathrm{Im}U_{\boldsymbol\alpha}(\xi)+ \mathrm{Im}\kappa_{\boldsymbol\alpha}(\xi)b^2 < & \mathrm{Im}U_{\boldsymbol\alpha}(\xi^*\pm \mathbf iL) + \mathrm{Im}
\kappa_{\boldsymbol\alpha}(\xi^*\pm \mathbf iL)b^2 - (2\pi-\epsilon_1) \big |  \xi - (\xi^*\pm \mathbf iL ) \big|\\
< & -2\mathrm{Vol}(\Delta)-2\epsilon_1  -(2\pi-\epsilon_1)\big(|\xi-\xi^*|-L \big)
\end{split}
\end{equation}
for all $\xi \in \Gamma^*\setminus \Gamma^*_L.$  Finally, putting (\ref{Bou}) and (\ref{last2}) together, we have  (\ref{cl}) and the first  inequality in (\ref{CI}); and since 
$$|\xi-\xi^*|-L\to+\infty$$
as $\xi \in \Gamma^*\setminus \Gamma^*_L$ approaches $\infty,$ we have the second inequality in (\ref{CI}). This completes the proof of (III).
\medskip

Putting (I), (II)  and (III) together, we have as $b\to 0,$ 
$$\bigg\{\begin{matrix} a_1 & a_2 & a_3 \\ a_4 & a_5 & a_6 \end{matrix} \bigg\}_b=\frac{e^{\frac{-\mathrm{Vol}(\Delta)}{\pi b^2}}}{\sqrt[4]{-\det\mathrm{Gram}(\Delta)}}\Big(1+O\big(b^2\big)\Big). $$ 
\end{proof}

As an immediate consequence of Theorem \ref{vol3} and Proposition \ref{reflection},  the reflection symmetry, we have 

\begin{theorem}\label{vol} Let $(\theta_1,\dots, \theta_6)$ be the dihedral angles of a truncated hyperideal hyperbolic tetrahedron $\Delta.$ Then as $b\to0,$
$$\bigg\{\begin{matrix} \frac{Q}{2} \pm \frac{\theta_1}{2\pi b} & \frac{Q}{2} \pm  \frac{\theta_2}{2\pi b} & \frac{Q}{2} \pm \frac{\theta_3}{2\pi b} \\ \frac{Q}{2} \pm  \frac{\theta_4}{2\pi b} & \frac{Q}{2} \pm  \frac{\theta_5}{2\pi b} & \frac{Q}{2} \pm  \frac{\theta_6}{2\pi b} \end{matrix} \bigg\}_b=\frac{e^{-\frac{\mathrm{Vol}(\Delta)}{\pi b^2}}}{\sqrt[4]{-\det\mathrm{Gram}(\Delta)}} \Big(1 +O\big(b^2\big)\Big),$$
where in each entry of the $b$-$6j$ symbol the sign $+$ and $-$ can be chosen arbitrarily,  $\mathrm{Vol}(\Delta)$ is the hyperbolic volume of $\Delta$ and $\mathrm{Gram}(\Delta)$ is the Gram matrix of $\Delta$ in the dihedral angles. 
\end{theorem}

\subsubsection{The anti-de Sitter case}

Let $(\theta_1,\dots, \theta_6)$ be the dihedral angles of a truncated hyperideal tetrahedron $\Delta$ in $\ads^3$. Without loss of generality, we assume that 
$\mathrm{Re}\theta_1=\mathrm{Re}\theta_4=\pi$
 and $\mathrm{Re}\theta_2=\mathrm{Re}\theta_3=\mathrm{Re}\theta_5=\mathrm{Re}\theta_6=0$ in the rest of this subsection.
For $k\in\{1,\dots,6\}$, let $a_k= \frac{Q}{2} + \frac{\theta_k}{2\pi b}$. 
Then we have 
$$\mathrm{Re}t_1=\mathrm{Re}t_2=\mathrm{Re}t_3=\mathrm{Re}t_4=2Q-\frac{b}{2},$$
$$\mathrm{Re}q_1=\mathrm{Re}q_2=3Q-b\quad\text{and}\quad\mathrm{Re}q_3=\mathrm{Re}q_4=2Q.$$
In particular, 
$$\mathrm{Re}q_j-\mathrm{Re}t_i\in\Big\{\frac{b}{2},Q-\frac{b}{2}\Big\}$$
for all $i,j\in\{1,2,3,4\},$
and the six-tuple $(a_1,\dots,a_6)$ is $b$-admissible. Also, by Theorem \ref{criterion2}, Remarks \ref{rm} and \ref{rm2}, we have any $i\in\{1,2,3,4\}$ and $j\in\{1,2\},$
$$\mathrm{Im}q_j<\mathrm{Im}t_i<\mathrm{Im}q_4=0<\mathrm{Im}q_3.$$
By Definition \ref{b6j}, the $b$-$6j$ symbol of $(a_1,\dots,a_6)$ is computed by
\begin{equation}\label{fk2}
\bigg\{\begin{matrix} a_1 & a_2 & a_3 \\ a_4 & a_5 & a_6 \end{matrix} \bigg\}_b=\Bigg(\frac{1}{\prod_{i=1}^4\prod_{j=1}^4S_b(q_j-t_i)}\Bigg)^{\frac{1}{2}}\int_\Gamma \prod_{i=1}^4S_b(u-t_i)\prod_{j=1}^4S_b(q_j-u)d u,
\end{equation}
where the contour $\Gamma$ is any vertical line passing the interval $(2Q-\frac{b}{2},2Q)$. See Figure \ref{Da3} (a).

\begin{figure}[htbp]
\centering
\includegraphics[scale=0.33]{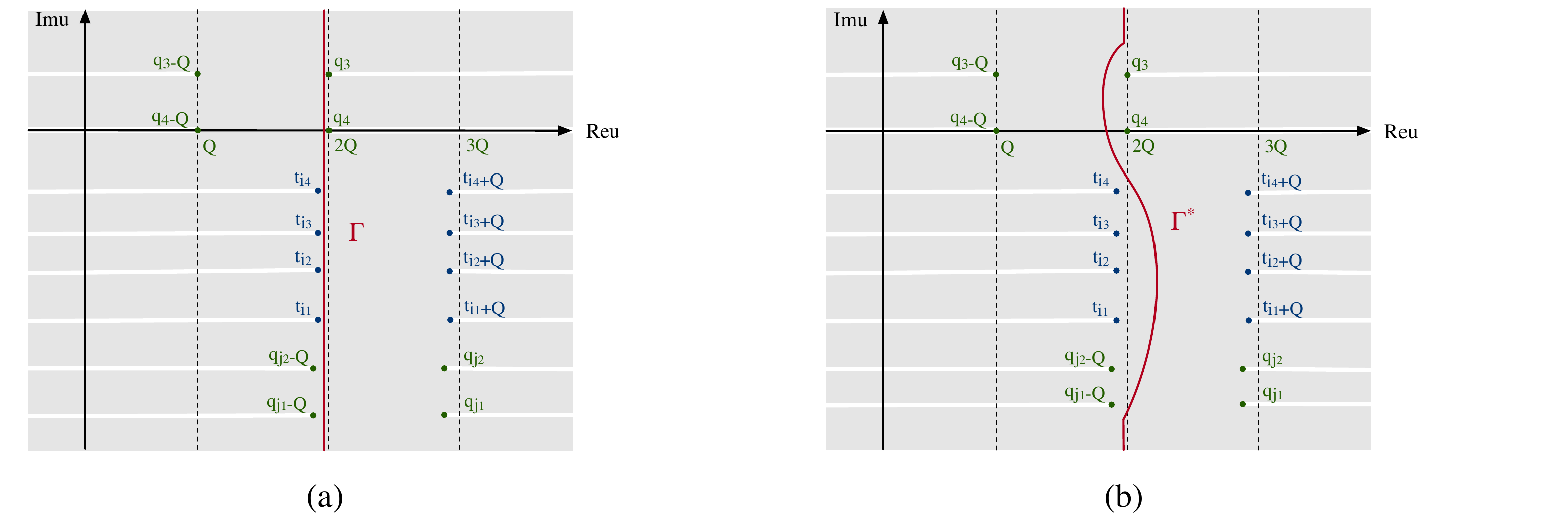}
\caption{Contours $\Gamma$ and $\Gamma^*$, and possible zeros and poles (located in the white rays) of the integrand in (\ref{fk2}), where $\{i_1,i_2,i_3,i_4\}=\{1,2,3,4\}$  with $\mathrm{Im}t_{i_1}\leqslant \mathrm{Im}t_{i_2}\leqslant \mathrm{Im}t_{i_3}\leqslant \mathrm{Im}t_{i_4},$ and $\{j_1,j_2\}=\{1,2\}$ with $\mathrm{Im}q_{j_1}\leqslant \mathrm{Im}q_{j_2}.$}
\label{Da3}
\end{figure}

By Proposition \ref{contour}, we can also choose the contour of integral as follows:
Given $L>0$ and $c\in [2Q-\frac{b}{2},2Q]$, let 
$\Gamma^*$ be a contour satisfying the following condition: for $u\in \Gamma^*$ with $|\mathrm {Im} u|\geqslant L$, $\mathrm {Re} u=c$; and for $u\in \Gamma^*$ with $|\mathrm {Im} u|\leqslant L,$ if $\mathrm{Im}u=\mathrm{Im}t_i$ for some $i\in\{1,2,3,4\},$ then 
$$\mathrm{Re}t_i<\mathrm{Re}u<\mathrm{Re}t_i+Q,$$
and if $\mathrm{Im}u=\mathrm{Im}q_j$ for some $j\in\{1,2,3,4\},$ then 
$$\mathrm{Re}q_j-Q<\mathrm{Re}u<\mathrm{Re}q_j.$$ See Figure \ref{Da3} (b).
Then for $L$ sufficiently large, we have 
\begin{equation*}
\bigg\{\begin{matrix} a_1 & a_2 & a_3 \\ a_4 & a_5 & a_6 \end{matrix} \bigg\}_b=\Bigg(\frac{1}{\prod_{i=1}^4\prod_{j=1}^4S_b(q_j-t_i)}\Bigg)^{\frac{1}{2}}\int_{\Gamma^*} \prod_{i=1}^4S_b(u-t_i)\prod_{j=1}^4S_b(q_j-u)du.
\end{equation*}

Under the change of variables in (\ref{alpha-a}) and (\ref{tau-t}), we have
\begin{equation}\label{retau}
\mathrm{Re}\tau_1=\mathrm{Re}\tau_2=\mathrm{Re}\tau_3=\mathrm{Re}\tau_4=2\pi,
\end{equation}
\begin{equation}\label{reeta}
\mathrm{Re}\eta_1=\mathrm{Re}\eta_2=3\pi\quad\text{and}\quad \mathrm{Re}\eta_3=\mathrm{Re}\eta_4=2\pi;
\end{equation}
and for any $i\in\{1,2,3,4\}$ and $j\in\{1,2\},$
we have
\begin{equation}\label{imtaueta}
\mathrm{Im}\eta_j<\mathrm{Im}\tau_i<\mathrm{Im}\eta_4=0<\mathrm{Im}\eta_3.
\end{equation}
Then by  (\ref{6jint}), we have 
\begin{equation}
\bigg\{\begin{matrix} a_1 & a_2 & a_3 \\ a_4 & a_5 & a_6 \end{matrix} \bigg\}_b=\frac{1}{\pi b}\int_{\Gamma}\exp\bigg(\frac{U_{\boldsymbol \alpha}(\xi)+\kappa_{\boldsymbol \alpha}(\xi)b^2+
\nu_{\boldsymbol \alpha,b}(\xi)b^4}{2\pi \mathbf ib^2} \bigg)d\xi,
\end{equation}
where $U_{\boldsymbol\alpha}$, $\kappa_{\boldsymbol\alpha}$ and $\nu_{\boldsymbol\alpha,b}$ are respectively as defined in (\ref{U}), (\ref{kappa}) and (\ref{nuU}), and the integral contour $\Gamma$ can be chosen as follows:  For  $\delta >0$ sufficiently small and $c \in (\delta, \pi-\delta),$ let $D_{\delta,c}=D_{\delta,c}^{\boldsymbol\alpha}$ be the shaded region depicted in Figure \ref{Ddc2} consisting of $\xi$ such that: (a) $2\pi-c<\mathrm{Re}\xi<2\pi+c$, (b) for $i\in\{1,2,3,4\}$, if $|\mathrm{Im}\xi-\mathrm{Im}\tau_i|<\delta,$ then $\mathrm{Re}\xi>2\pi+\delta,$ (c) for $j\in\{1,2\},$ if $|\mathrm{Im}\xi-\mathrm{Im}\eta_j|<\delta,$ then $\mathrm{Re}\xi>2\pi+\delta,$ and (d) for $j\in\{3,4\},$ if $|\mathrm{Im}\xi-\mathrm{Im}\eta_j|<\delta,$ then $\mathrm{Re}\xi<2\pi-\delta.$ Then $\Gamma$ is a contour in $D_{\delta,c}$ such that for $\xi\in\Gamma$ with $|\mathrm{Im}\xi|>L$
 for some $L$ sufficiently large, $\mathrm{Re}\xi=2\pi.$

\begin{figure}[htbp]
\centering
\includegraphics[scale=0.4]{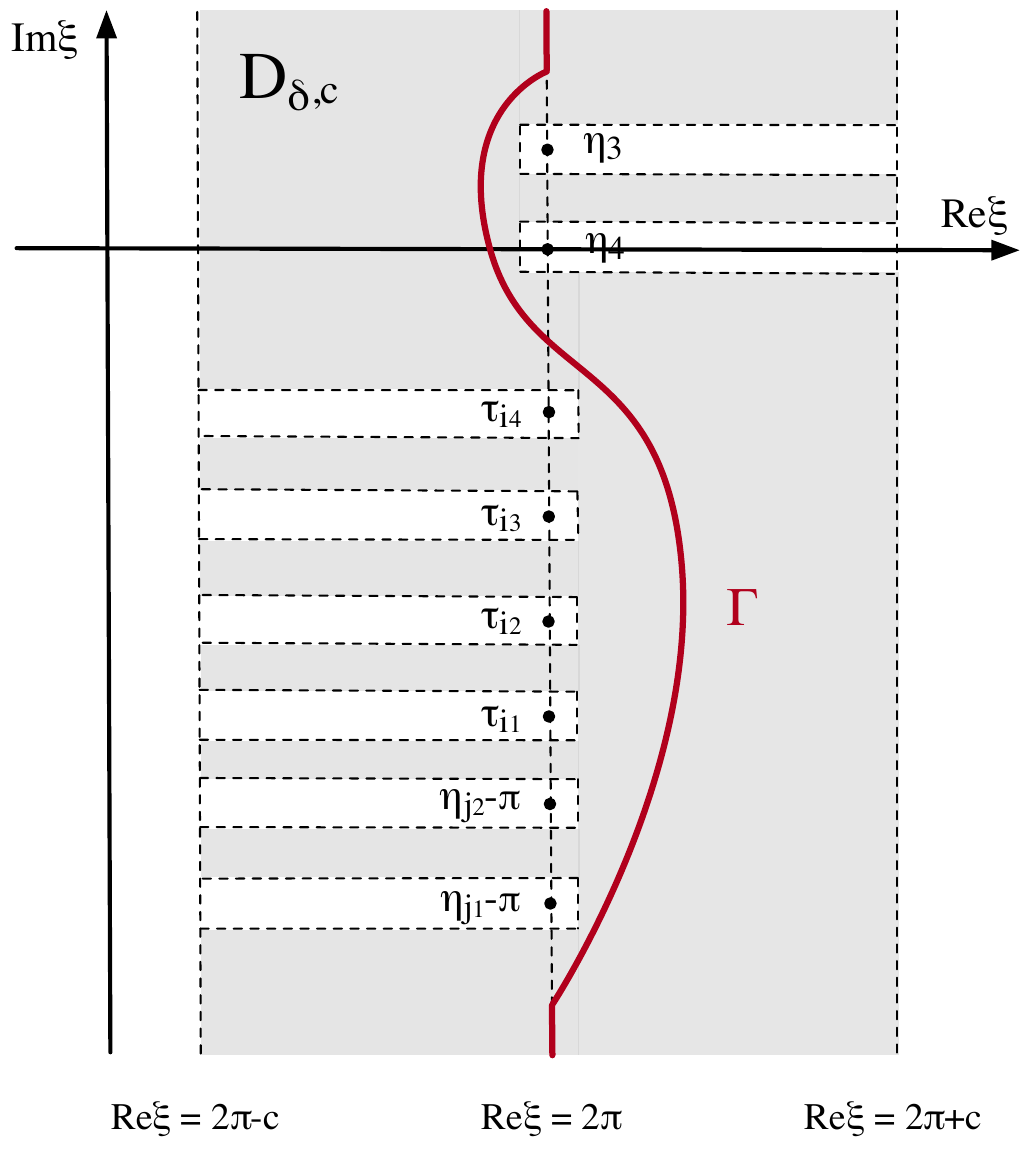}
\caption{Region $D_{\delta,c}$ and contour $\Gamma$, where $\{i_1,i_2,i_3,i_4\}=\{1,2,3,4\}$ with $\mathrm{Im}\tau_{i_1}\leqslant \mathrm{Im}\tau_{i_2}\leqslant \mathrm{Im}\tau_{i_3}\leqslant \mathrm{Im}\tau_{i_4},$ and $\{j_1,j_2\}=\{1,2\}$ with $\mathrm{Im}\eta_{j_1}\leqslant \mathrm{Im}\eta_{j_2}.$}
\label{Ddc2}
\end{figure}

\begin{proposition}\label{critical2} Let  $(\theta_1,\dots, \theta_6)$ be the dihedral angle of a truncated hyperideal tetrahedron $\Delta$ in $\ads^3$. Then the function $U_{\boldsymbol \alpha}(\xi)$ has a unique critical point $\xi^*$ in the domain $D_{\delta,c}$ with $$\mathrm{Re}\xi^*=2\pi,$$
and 
$$U_{\boldsymbol \alpha}(\xi^*)=2\mathrm{Vol}(\Delta).$$
\end{proposition}

As a direct consequence of Proposition \ref{critical2}, we have the following volume formula of a truncated hyperideal tetrahedron in $\ads^3$ in terms of the dihedral angles. 

\begin{theorem}\label{volume2} 
Let $\Delta$ be a truncated hyperideal tetrahedron in $\ads^3$ with dihedral angles $(\theta_1,\dots ,\theta_6),$ and let 
$$W(\theta_1,\dots,\theta_6)=U_{\boldsymbol\alpha}(\xi^*).$$ Then
$$\mathrm{Vol}(\Delta)=\frac{1}{2}\cdot W(\theta_1,\dots,\theta_6).$$
\end{theorem}

\begin{remark}\label{3.20}
By Remark \ref{ccov}, Proposition \ref{critical2} and Theorem \ref{volume2} can be re-stated in terms of the genuine volume as 
$$W(\theta_1,\dots,\theta_6)=-2\mathbf i\mathrm{Vol}_{\mathbb C}(\Delta).$$
Compare with \cite[Proposition 3.1 (1)]{LMSWY}, Remark \ref{3.4} and Proposition \ref{critical1}.
\end{remark}

\begin{proof}[Proof of Proposition \ref{critical2}]
The proof follows the arguments adapted from \cite{MY, U, BY, LMSWY}. 
For $k\in\{1,\dots,6\},$ let $u_k=e^{2\mathbf i\alpha_k},$ and let $z=e^{-2\mathbf i\xi}.$ Then by a direct computation or \cite[Eq.\ (3.9)]{BY} we have
\begin{equation}\label{mod4}
U_{\boldsymbol \alpha}'(\xi)=2\mathbf i\log\frac{(1-zu_1u_2u_3)(1-zu_1u_5u_6)(1-zu_2u_4u_6)(1-zu_3u_4u_5)}{(1-z)(1-zu_1u_2u_4u_5)(1-zu_1u_3u_4u_6)(1-zu_2u_3u_5u_6)}\quad\quad (\mathrm{mod}\ 4\pi).
\end{equation} Here $\mathrm{mod}\ 4\pi$ is due to the ``$2\mathbf i\log$'' in~\eqref{eq:Lx}. Then, as a necessary condition, $U'_{\boldsymbol \alpha}(\xi)=0$ implies
\begin{equation*}
\frac{(1-zu_1u_2u_3)(1-zu_1u_5u_6)(1-zu_2u_4u_6)(1-zu_3u_4u_5)}{(1-z)(1-zu_1u_2u_4u_5)(1-zu_1u_3u_4u_6)(1-zu_2u_3u_5u_6)}=1,
\end{equation*}
which is equivalent to the following quadratic equation  (see \cite[Eq. (3.11)]{BY})
$$Az^2+Bz+C=0$$
with 
\begin{equation*}
\begin{split}
A=&u_1u_4+u_2u_5+u_3u_6-u_1u_2u_6-u_1u_3u_5-u_2u_3u_4-u_4u_5u_6+u_1u_2u_3u_4u_5u_6,\\
B=&-\Big(u_1-\frac{1}{u_1}\Big)\Big(u_4-\frac{1}{u_4}\Big)-\Big(u_2-\frac{1}{u_2}\Big)\Big(u_5-\frac{1}{u_5}\Big)-\Big(u_3-\frac{1}{u_3}\Big)\Big(u_6-\frac{1}{u_6}\Big),\\
C=&\frac{1}{u_1u_4}+\frac{1}{u_2u_5}+\frac{1}{u_3u_6}-\frac{1}{u_1u_2u_6}-\frac{1}{u_1u_3u_5}-\frac{1}{u_2u_3u_4}-\frac{1}{u_4u_5u_6}+\frac{1}{u_1u_2u_3u_4u_5u_6}.\\
\end{split}
\end{equation*}
By  \cite[Page 15]{BY} and Theorem \ref{criterion2}, we have $B^2-4AC=16\det\mathrm{Gram}(\Delta)>0$.
Let 
\begin{equation*}\label{xi}
z^*=\frac{-B+\sqrt{B^2-4AC}}{2A}\quad\text{and}\quad {z^{**}}=\frac{-B-\sqrt{B^2-4AC}}{2A}
\end{equation*}
be the two roots of this quadratic equation. Then by Lemma \ref{MathCompetition} below, we have $A>0,$ $B<0$ and $C>0;$ and as a consequence, we have 
$z^*>0$ and $z^{**}>0,$ and hence there are $\xi^*$ and $\xi^{**}$ in $D_{\delta,c}$ with $\mathrm{Re}\xi^*=\mathrm{Re}\xi^{**}=2\pi$ such that $z^*=e^{-2\mathbf i\xi^*}$ and ${z^{**}}=e^{-2\mathbf i{\xi^{**}}}$. 

By~\eqref{mod4} we have 
$$U_{\boldsymbol \alpha}'(\xi^*)=4k\pi\quad\text{and}\quad U_{\boldsymbol \alpha}'(\xi^{**})=4k'\pi$$
for some integers $k$ and $k'.$ We claim that $k=0$ and $k'=2,$ as a consequence of which $\xi^*$ is the unique critical point of $U_{\boldsymbol\alpha}$ in $D_{\delta,c}.$  Indeed, as $k$ and $k'$ are integers depending continuously on $\{\theta_1,\dots,\theta_6\},$ it suffices to compute them at a special value. To do this, let $\theta_1=\theta_4=\pi-\mathbf i \ln 2,$ and let $\theta_2=\theta_3=\theta_5=\theta_6=\mathbf i \epsilon$ for some $\epsilon>0$ sufficiently small. A direct computation verifies that $(\theta_1,\dots,\theta_6)$ satisfies the conditions (a), (b) and (c) of Theorem \ref{criterion2} (2), hence are the dihedral angles of a truncated hyperideal tetrahedron in $\ads^3.$
Then $u_1=u_4=2$ and $u_2=u_3=u_5=u_6=e^{2\mathbf i\epsilon}$ which is sufficiently close to $1.$ As a consequence, the quadratic equation $Az^2+Bz+C=0$ is sufficiently close to $2z^2-\frac{9}{4}z+\frac{1}{2}=0$ with $z^*$ sufficiently close to $\frac{9+\sqrt{17}}{16}\approx 0.820$ and $z^{**}$ sufficiently close to $\frac{9-\sqrt{17}}{16}\approx 0.305,$ or equivalently, $\xi^*\approx 2\pi - 0.099\mathbf i$ and $\xi^{**}\approx 2\pi - 0.594 \mathbf i.$  On the other hand, we have $\mathrm{Im}\eta_1=\mathrm{Im}\eta_2=\ln 2 + 2\epsilon \approx -0.693$ and $\mathrm{Im}\tau_1=\mathrm{Im}\tau_2=\mathrm{Im}\tau_3=\mathrm{Im}\tau_4=\frac{\ln 2}{2}+2\epsilon \approx -0.347.$ This implies that $\mathrm{Im}\xi^*\in(\mathrm{Im}\tau_{i_4},\mathrm{Im}\eta_4)$ and $\mathrm{Im}\xi^{**}\in(\mathrm{Im}\eta_{j_2},\mathrm{Im}\tau_{i_1}).$ Then by Proposition \ref{PL} below, 
$$\frac{\partial \mathrm{Im}U_{\boldsymbol\alpha}(\xi)}{\partial \mathrm{Im}\xi}\Big|_{\xi=\xi^*}=0\quad\text{and}\quad \frac{\partial \mathrm{Im}U_{\boldsymbol\alpha}(\xi)}{\partial \mathrm{Im}\xi}\Big|_{\xi=\xi^{**}}=8\pi,$$
implying that $k=0$ and $k'=2.$

To see the critical value $U_{\boldsymbol\alpha}(\xi^*),$ we use the same argument as in the proof of Theorem \ref{ab}. Namely, let 
$$W(\theta_1,\dots,\theta_6)=U_{\boldsymbol\alpha}(\xi^*);$$
and we first verify that $\frac{W}{2}$ satisfies the Schlafli formula (\ref{sch}), or equivalently, that 
\begin{equation}\label{2sch}
\frac{\partial W(\theta_1,\dots,\theta_6)}{\partial\mathrm{Im}\theta_k}=-l_k.
\end{equation}
By the symmetry of $W$ in its variables, it suffices to prove it for $k=1.$ By the same algebraic computation as in  \cite[Eq. (3.14)]{LMSWY}, we have 
\begin{equation}\label{pW2}
\frac{\partial W(\theta_1,\dots,\theta_6)}{\partial\mathrm{Im}\theta_1}=\frac{\partial W(\theta_1,\dots,\theta_6)}{\partial 
(-\mathbf i\theta_1)}=-\frac{1}{2}\log\frac{G_{12}-\sqrt{G_{12}^2-G_{11}G_{22}}}{G_{12}+\sqrt{G_{12}^2-G_{11}G_{22}}}.
\end{equation}
To show that the right hand side of (\ref{pW2}) equals $-l_1,$ 
by (\ref{cosh}) and (\ref{inner}), we have
$$\cosh l_1=-\frac{G_{12}}{\sqrt{G_{11}G_{22}}}.$$
Then
$$\frac{G_{12}-\sqrt{G_{12}^2-G_{11}G_{22}}}{G_{12}+\sqrt{G_{12}^2-G_{11}G_{22}}}=\frac{-\frac{G_{12}}{\sqrt{G_{11}G_{22}}}+\sqrt{\frac{G_{12}^2}{G_{11}G_{22}}-1}}{-\frac{G_{12}}{\sqrt{G_{11}G_{22}}}-\sqrt{\frac{G_{12}^2}{G_{11}G_{22}}-1}}=\frac{\cosh l_1-\sqrt{\cosh^2l_1-1}}{\cosh l_1+\sqrt{\cosh^2l_1-1}}=e^{2l_1},$$
from (\ref{pW2}) and which, (\ref{2sch}) follows.
This implies that 
$$W(\theta_1,\dots,\theta_6)=2\mathrm{Vol}(\Delta)+K$$
for some constant $K.$ To determine this constant $K$, we check the values at the  dihedral angles  $\boldsymbol\theta(\epsilon)=(\theta_1(\epsilon), \dots, \theta_6(\epsilon))$ of a continuous family of hyperideal tetrahedra in $\ads^3$  that converge as $\epsilon$ tends to $0$ to a flat tetrahedron $\Delta_0$. Then we have 
$\lim_{\epsilon\to0}\boldsymbol\theta(\epsilon)=(\pi,0,0,\pi,0,0),$
$\lim_{\epsilon\to 0}\tau_i(\epsilon)=2\pi$ for all $i\in\{1,2,3,4\}$, $\lim_{\epsilon\to 0}\eta_j(\epsilon)=3\pi$ for all $j\in\{1,2\}$, and $\lim_{\epsilon\to 0}\eta_j(\epsilon)=2\pi$ for all $j\in\{3,4\}$. By Proposition \ref{PL} below, $\mathrm{Im}\xi^*\in (\mathrm{Im}\tau_{i_4},\mathrm{Im}\eta_4),$ which implies that 
$\lim_{\epsilon\to 0}\xi^*(\epsilon)=2\pi$, and 
\begin{equation*}
\lim_{\epsilon\to 0} W(\boldsymbol\theta(\epsilon)) =0=2\mathrm{Vol}(\Delta_0).
\end{equation*}
This, by Proposition \ref{cont}, implies that $K=0,$ completing the proof.
\end{proof}

\begin{lemma}\label{MathCompetition} 
\begin{enumerate}[(1)]
\item $A>0$,
\item $B<0,$ and 
\item $C>0.$
\end{enumerate}
\end{lemma}

\begin{proof}       
To see (1),  we first observe the following:
\begin{enumerate}[(a)]
\item $u_1>1,$ $u_4>1,$ $-1<u_2<0,$ $-1<u_3<0,$ $-1<u_5<0$, $-1<u_6<0,$ and 
\item for $i\in\{1,4\}$ and $j\in \{2,3,5,6\},$ $u_iu_j<-1.$
\end{enumerate}
Indeed, (a) follows directly from the definition of the dihedral angles $\theta_k$'s, and (b) follows directly from Theorem \ref{criterion2}, Remark \ref{rm} and  Remark \ref{rm2} (in particular, (3) of Remark \ref{rm}). 

Then we have 
\begin{equation*}
\begin{split}
A&=(u_1u_4-1)(u_2u_5-1)(u_3u_6-1)\\
&\ \ + \frac{1}{2}\Big( -u_2^2 -u_3^2 -u_5^2 -u_6^2  + 2u_2u_3u_5u_6+2  \\
&\ \ + (u_1u_6-u_2)(u_3u_4-u_2)+ (u_1u_5-u_3)(u_2u_4-u_3)\\
&\ \ + (u_1u_3-u_5)(u_4u_6-u_5)+ (u_1u_2-u_6)(u_4u_5-u_6) \Big)\\
&>(u_1u_4-1)(u_2u_5-1)(u_3u_6-1)\\
&\ \ +\frac{1}{2}\bigg( \Big((u_1u_6-u_2)(u_3u_4-u_2)- (1+u_2)^2\Big)\\
&\ \ + \Big((u_1u_5-u_3)(u_2u_4-u_3) -(1+u_3)^2 \Big) \\
&\ \ + \Big((u_1u_3-u_5)(u_4u_6-u_5) -(1+u_5)^2\Big)\\
&\ \ + \Big((u_1u_2-u_6)(u_4u_5-u_6)- (1+u_6)^2 \Big)\bigg)>0,\\
\end{split}
\end{equation*}
where the first inequality comes from (a) and that 
$u_2u_3u_5u_6+u_2+u_3+u_5+u_6+3=(u_2u_3-1)(u_5u_6-1)+(u_2+1)(u_3+1)+(u_5+1)(u_6+1)>0,$ hence $2u_2u_3u_5u_6+2>-2u_2-2u_3-2u_5-2u_6-4,$
and the last inequality comes from  that each summand is positive due to (a) and (b).
\medskip

To see (3), as a consequence of (a) and (b), we have 
\begin{enumerate}[(a$'$)]
\item 
$0<\frac{1}{u_1}<1,$ $0<\frac{1}{u_4}<1,$ $\frac{1}{u_2}<-1,$ $\frac{1}{u_3}<-1,$ $\frac{1}{u_5}<-1$, $\frac{1}{u_6}<-1$, and
\item for $i\in\{1,4\}$ and $j\in \{2,3,5,6\},$ $-1<\frac{1}{u_iu_j}<0.$
\end{enumerate}
Then the proof follows verbatim that of $A>0.$
\medskip

Finally, (2) follows directly from (a) and (a$'$).
\end{proof}

 \begin{proposition}\label{Hess2} At the critical point  $\xi^*$ of $U_{\boldsymbol\alpha}$, we have
 \begin{equation*}
\frac{-U''_{\boldsymbol \alpha}(\xi^*)}{\exp\big(\frac{\kappa_{\boldsymbol \alpha}(\xi^*)}{\pi \mathbf i}\big)}=16\sqrt{\det\mathrm{Gram}(\Delta)}.
\end{equation*}
As a consequence, $\xi^*$ is a non-degenerate critical point of $U_{\boldsymbol\alpha}.$
\end{proposition}

\begin{proof} We observe that for each $k\in\{1,\dots,6\},$
$$\cos\theta_k=-\frac{u_k+u_k^{-1}}{2}.$$ 
Hence 
$$\mathrm{Gram}(\Delta)=-\begin{bmatrix}
1 & \frac{u_4+u_4^{-1}}{2} & \frac{u_5+u_5^{-1}}{2}
    & \frac{u_3+u_3^{-1}}{2}\\
\frac{u_4+u_4^{-1}}{2} & 1 & \frac{u_6+u_6^{-1}}{2}
    & \frac{u_2+u_2^{-1}}{2}\\
   \frac{u_5+u_5^{-1}}{2} & \frac{u_6+u_6^{-1}}{2} & 1
    & \frac{u_1+u_1^{-1}}{2}\\
 \frac{u_3+u_3^{-1}}{2} & \frac{u_2+u_2^{-1}}{2} & \frac{u_1+u_1^{-1}}{2}
    & 1
  \end{bmatrix},$$
  and as algebraically computed in \cite{U}, 
  $$B^2-4AC=16\det\mathrm{Gram}(\Delta).$$
  Then the rest of the proof follows verbatim that of Proposition \ref{Hess}.
\end{proof}

\begin{proposition}\label{PL} 
As depicted in Figure \ref{Fig:PL},
 $\mathrm{Im}U_{\boldsymbol\alpha}(\xi)$ is a piecewise linear function on the straight line $\Gamma_{2\pi}=\{\xi\in \mathbb C\ |\ \mathrm{Re}\xi=2\pi\}$ with
\begin{equation*}
\frac{\partial \mathrm{Im}U_{\boldsymbol\alpha}(\xi)}{\partial \mathrm{Im}\xi}=
\left\{\begin{array}{ccl}
      4\pi & \text{if } & -\infty < \mathrm{Im} \xi< \mathrm{Im}\eta_{j_1},\\
       6\pi & \text{if } & \mathrm{Im}\eta_{j_1}< \mathrm{Im} \xi< \mathrm{Im}\eta_{j_2},\\
        8\pi & \text{if } & \mathrm{Im}\eta_{j_3}< \mathrm{Im} \xi< \mathrm{Im}\tau_{i_1},\\
         6\pi & \text{if } & \mathrm{Im}\tau_{i_1}< \mathrm{Im} \xi< \mathrm{Im}\tau_{i_2},\\
          4\pi & \text{if } & \mathrm{Im}\tau_{i_2}< \mathrm{Im} \xi< \mathrm{Im}\tau_{i_3},\\
           2\pi & \text{if } & \mathrm{Im}\tau_{i_3}< \mathrm{Im} \xi< \mathrm{Im}\tau_{i_4},\\
            0 & \text{if } & \mathrm{Im}\tau_{i_4}< \mathrm{Im} \xi< \mathrm{Im}\eta_{4},\\
             -2\pi & \text{if } & \mathrm{Im}\eta_{4}< \mathrm{Im} \xi< \mathrm{Im}\eta_{3},\\
              -4\pi & \text{if } & \mathrm{Im}\eta_{3}< \mathrm{Im} \xi< +\infty.\\
    \end{array} \right.
\end{equation*}
\end{proposition}

\begin{figure}[htbp]
\centering
\includegraphics[scale=0.3]{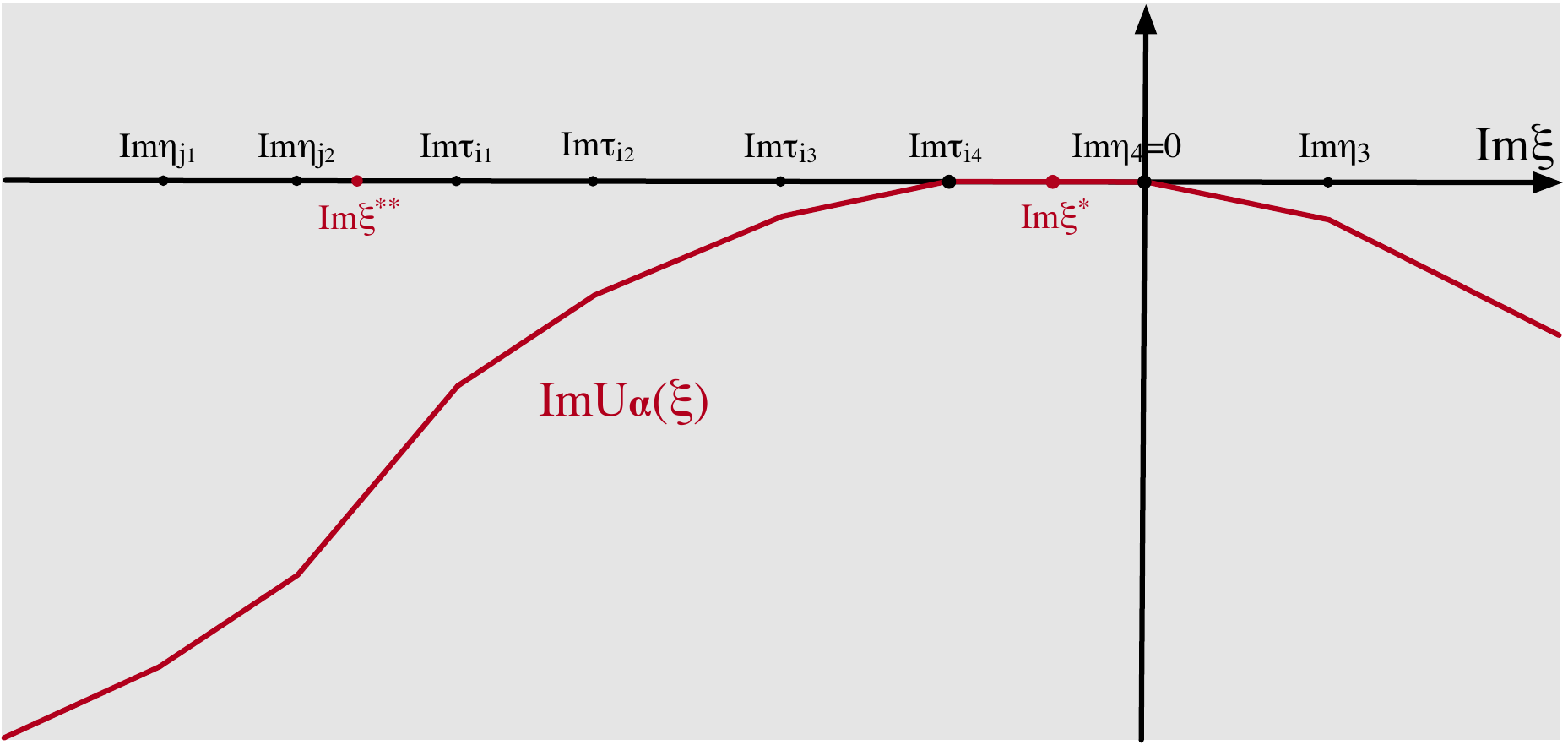}
\caption{Graph of $\mathrm{Im}U_{\boldsymbol\alpha}(\xi)$ in $\mathrm{Im}\xi$.}
\label{Fig:PL}
\end{figure}

\begin{proof} By (\ref{limp}), we have
\begin{equation}\label{a}
\begin{split}
\frac{\partial\mathrm{Im}L(\mathbf il)}{\partial l}=&\lim_{\alpha\to 0^+}\frac{\partial \mathrm{Im}L(\alpha+\mathbf il)}{\partial l}\\
=&\lim_{\alpha\to0^+} \Big(2\alpha-\pi-2\arg\big(1-e^{-2l+2\mathbf i\alpha}\big)\Big)=\left\{\begin{array}{ccl}
-\pi & \text{if} & l>0,\\
\pi & \text{if} & l<0,
 \end{array} \right.
 \end{split}
 \end{equation}
 and
\begin{equation}\label{b}
\begin{split}
\frac{\partial\mathrm{Im}L(\pi+\mathbf il)}{\partial l}=&\lim_{\alpha\to \pi^-}\frac{\partial \mathrm{Im}L(\alpha+\mathbf il)}{\partial l} \\= &\lim_{\alpha\to\pi^-} \Big(2\alpha-\pi-2\arg\big(1-e^{-2l+2\mathbf i\alpha}\big)\Big)=\left\{\begin{array}{ccl}
\pi & \text{if} & l>0,\\
-\pi & \text{if} & l<0.
 \end{array} \right.
  \end{split}
 \end{equation}
 Then the result follows directly from (\ref{a}) and (\ref{b}).
\end{proof}

Let $D$ be the region in $\mathbb C$ consisting of $\xi$ with: (a) for each $i\in\{1,2,3,4\},$ if $\mathrm{Im}\xi=\mathrm{Im}\tau_i,$ then $\mathrm{Re}\xi>2\pi,$ (b) for each $j\in\{1,2\},$ if $\mathrm{Im}\xi=\mathrm{Im}\eta_j,$ then $\mathrm{Re}\xi>2\pi,$ and (c) for each $j\in\{3,4\},$ if $\mathrm{Im}\xi=\mathrm{Im}\eta_j,$ then $\mathrm{Re}\xi<2\pi.$

\begin{proposition}
\label{cor4}
There exists an $\delta>0$ such that: \begin{enumerate}[(1)]
\item  For each $i\in\{1,2,3,4\},$ if $\xi$ lies in the $\delta$-square neighborhood of $\tau_i$ in $D,$  then 
$$\frac{\partial \mathrm{Im}U_{\boldsymbol \alpha}(\xi)}{\partial \mathrm{Re}\xi}<0.$$
\item For each $j\in\{1,2\},$ if $\xi$ lies in the $\delta$-square neighborhood of $\eta_j-\pi$ in $D,$ then
$$\frac{\partial \mathrm{Im}U_{\boldsymbol \alpha}(\xi)}{\partial \mathrm{Re}\xi}>0.$$
\item For each $j\in\{3,4\},$  if $\xi$ lies in the $\delta$-square neighborhood of $\eta_j$ in $D,$ then
$$\frac{\partial \mathrm{Im}U_{\boldsymbol \alpha}(\xi)}{\partial \mathrm{Re}\xi}>0.$$ 
\end{enumerate}
\end{proposition}

\begin{proof}
The result follows directly from the continuity of $\frac{\partial \mathrm{Im}U_{\boldsymbol \alpha}(\xi)}{\partial \mathrm{Re}\xi}$ and the following claims:
\begin{enumerate}[(a)]
    \item 
For all $i\in\{1,2,3,4\},$ 
$$\lim_{\xi\to\tau_i}\frac{\partial \mathrm{Im}U_{\boldsymbol \alpha}(\xi)}{\partial \mathrm{Re}\xi}= -\infty.$$
\item For all $j\in\{1,2\},$  
$$\lim_{\xi\to\eta_j-\pi}\frac{\partial \mathrm{Im}U_{\boldsymbol \alpha}(\xi)}{\partial \mathrm{Re}\xi}= +\infty.$$
\item For all $j\in\{3,4\},$ 
$$\lim_{\xi\to\eta_j}\frac{\partial \mathrm{Im}U_{\boldsymbol \alpha}(\xi)}{\partial \mathrm{Re}\xi}= +\infty.$$
\end{enumerate}
To prove the claims, by a direct computation, we have
$$ \frac{\partial \mathrm{Im}L(\alpha+\mathbf il)}{\partial \alpha} =2\log\Big|1-e^{2\mathbf i(\alpha+\mathbf i l)}\Big|+2l,$$
$$ \frac{\partial \mathrm{Im}L(-\alpha+\mathbf il)}{\partial \alpha} =-2\log\Big|1-e^{2\mathbf i(-\alpha+\mathbf i l)}\Big|-2l,$$
and $ \frac{\partial \mathrm{Im}L(\pm \alpha+\mathbf il)}{\partial \alpha} $ converges to a finite value as $\pm \alpha + \mathbf i l$ tends to a point that is not $k\pi$ for any integer $k.$ 
As a consequence, we have 
$$\lim_{\alpha +\mathbf i  l\to 0} \frac{\partial \mathrm{Im}L(\alpha+\mathbf il)}{\partial \alpha} =-\infty,$$
and 
$$\lim_{-\alpha+\mathbf i l\to \pi} \frac{\partial \mathrm{Im}L(-\alpha+\mathbf il)}{\partial \alpha} = + \infty\quad\text{and}\quad\lim_{-\alpha+\mathbf i l\to 0} \frac{\partial \mathrm{Im}L(-\alpha+\mathbf il)}{\partial \alpha} = + \infty.$$
Then as $$\frac{\partial \mathrm{Im}U_{\boldsymbol \alpha}(\xi)}{\partial \mathrm{Re}\xi} =\sum_{i=1}^4\frac{\partial L(\xi-\tau_i)}{\partial \mathrm{Re}\xi}+\sum_{j=1}^4\frac{\partial L(\eta_j-\xi)}{\partial \mathrm{Re}\xi},$$
we have
$$\lim_{\xi\to\tau_i}\frac{\partial \mathrm{Im}U_{\boldsymbol \alpha}(\xi)}{\partial \mathrm{Re}\xi}= -\infty$$
for $i\in\{1,2,3,4\},$
$$\lim_{\xi\to\eta_j-\pi}\frac{\partial \mathrm{Im}U_{\boldsymbol \alpha}(\xi)}{\partial \mathrm{Re}\xi}= +\infty$$
for $j\in\{1,2\},$ and 
$$\lim_{\xi\to\eta_j}\frac{\partial \mathrm{Im}U_{\boldsymbol \alpha}(\xi)}{\partial \mathrm{Re}\xi}= +\infty$$
for $j\in\{3,4\}.$
\end{proof}

\begin{proposition}\label{nonvanish}
For $\xi\in D$ with $\mathrm{Re}\xi=2\pi,$ if $\mathrm{Im}\xi\in (\mathrm{Im}\tau_{i_4},\mathrm{Im}\xi^*),$ then 
$$\frac{\partial \mathrm{Im}U_{\boldsymbol\alpha}(\xi)}{\partial \mathrm{Re}\xi}<0;$$
and if $\mathrm{Im}\xi\in (\mathrm{Im}\xi^*,\mathrm{Im}\eta_{4}),$ then 
$$\frac{\partial \mathrm{Im}U_{\boldsymbol\alpha}(\xi)}{\partial \mathrm{Re}\xi}>0.$$
\end{proposition}

\begin{proof} Let $\mathrm I$ be the straight line segment connecting $\tau_{i_4}$ and $\eta_4.$ By Proposition \ref{PL}, $\mathrm{Im}U_{\boldsymbol\alpha}$ is a constant on $\mathrm I,$ hence $\frac{\partial \mathrm{Im}U_{\boldsymbol\alpha}}{\partial \mathrm{Im}\xi} \equiv 0$ on $\mathrm I;$ and since $\xi^*$ is the only critical point of $\mathrm U_{\boldsymbol\alpha},$ we have $\frac{\partial \mathrm{Im}U_{\boldsymbol\alpha}}{\partial \mathrm{Re}\xi}\neq 0$ on $\mathrm I\setminus \{\xi^*\}.$ By Proposition \ref{cor4}, $\frac{\partial \mathrm{Im}U_{\boldsymbol\alpha}}{\partial \mathrm{Re}\xi} < 0$ near the end point $\tau_{i_4}$ of $\mathrm I,$ and $\frac{\partial \mathrm{Im}U_{\boldsymbol\alpha}}{\partial \mathrm{Re}\xi} > 0$ near the end point $\eta_4$ of $\mathrm I.$  Therefore, $\frac{\partial \mathrm{Im}U_{\boldsymbol\alpha}}{\partial \mathrm{Re}\xi} < 0$ on $\mathrm I\setminus \{\xi^*\}$ when $\mathrm{Im}\xi\in (\mathrm{Im}\tau_{i_4},\mathrm{Im}\xi^*),$  and $\frac{\partial \mathrm{Im}U_{\boldsymbol\alpha}}{\partial \mathrm{Re}\xi} > 0$ on $\mathrm I\setminus \{\xi^*\}$  when $\mathrm{Im}\xi\in (\mathrm{Im}\xi^*,\mathrm{Im}\eta_{4}).$ 
\end{proof}

Finally, we need the following Proposition \ref{bound2} whose proof follows verbatim that of \cite[Proposition 3.10 and Proposition 3.11]{LMSWY}. 

 \begin{proposition}\label{bound2} 
 \begin{enumerate}[(1)]
 \item 
 For $\delta> 0$  sufficiently small and $c\in(\delta,\pi-\delta)$, there exists a constant $K=K_{\delta,c}>0$ such that 
$$\Bigg|\frac{\partial \mathrm{Im}\kappa_{\boldsymbol\alpha}(\xi)}{\partial\mathrm{Im}\xi}\Bigg|<K$$
for all $\xi\in D_{\delta,c}.$
\item  
For $b>0$ and $\delta> 0$ both sufficiently small and $c\in(\delta,\pi-\delta)$, there exists a constant $N=N_{\delta,c}>0$ independent of $b$ such that 
$$\mathrm{Im}\nu_{\boldsymbol\alpha,b}(\xi) \leqslant \big|\nu_{\boldsymbol\alpha,b}(\xi)\big|<N$$
for all $\xi\in D_{\delta,c}.$
\end{enumerate}
\end{proposition}

\begin{proof}[Proof of Theorem \ref{vol2}]
We will carefully choose an integral contour $\Gamma^*$ in (\ref{6jint}), and use  Proposition \ref{saddle}, the Saddle Point Approximation. A seemingly natural choice of the contour is $\Gamma_{2\pi}=\{\xi\in\mathbb C\ |\ \mathrm{Re}\xi=2\pi\}$ that passes through the unique critical point $\xi^*$ of $U_{\boldsymbol\alpha}$ given by Proposition \ref{critical2}.  However, there are  issues with this choice as:
\begin{enumerate}[(a)]
\item  $\Gamma_{2\pi}$ pass through the points $\tau_1,\tau_2,\tau_3,\tau_4,\eta_1-\pi, \eta_2-\pi, \eta_3$ and $\eta_4$ where $U_{\boldsymbol \alpha}$ is only continuous but not holomorphic, and 
\item  the critical point $\xi^*$ is not the unique maximum of $\mathrm{Im}U_{\boldsymbol\alpha}$ on the contour, hence condition (ii) of Proposition \ref{saddle} is not satisfied.
\end{enumerate} 
We will resolve these issues as follows. Let $\delta>0$ be sufficiently small so that: 
\begin{enumerate}[(i)]
    \item $\mathrm{Im}U_{\boldsymbol\alpha}(\xi)<\mathrm{Im}U_{\boldsymbol\alpha}(\xi^*)$ for all $\xi$ in the right-half of the $\delta$-square neighborhood of $\tau_1,\tau_2,\tau_3$ and $\tau_4.$ This can be done due to Corollary \ref{cor4} (1). 
    \item $\mathrm{Im}U_{\boldsymbol\alpha}(\xi)<\mathrm{Im}U_{\boldsymbol\alpha}(\xi^*)$ for all $\xi$ in the left-half of the $\delta$-square neighborhood of $\eta_3$ and $\eta_4.$ This can be done due to Corollary \ref{cor4} (3).   
    \item $\mathrm{Im}U_{\boldsymbol\alpha}(\xi)<\mathrm{Im}U_{\boldsymbol\alpha}(\xi^*)$ for all $\xi$ in the  $\delta$-square neighborhood of $\eta_1-\pi$ and $\eta_2-\pi.$ This can be done because of the continuity of $\mathrm{Im}U_{\boldsymbol\alpha}$ and that $\mathrm{Im}U_{\boldsymbol\alpha}(\eta_j-\pi)<\mathrm{Im}U_{\boldsymbol\alpha}(\xi^*)$
 for $j\in\{1,2\}$ by Proposition \ref{PL}.
 \end{enumerate}
Choose a $c\in(\delta, \pi-\delta),$ and observe that $\Gamma_{2\pi}$ intersects the complement of $D_{\delta,c}$. Then we deform $\Gamma_{2\pi}$ to the new contour $\Gamma'$ by pushing the parts out of $D_{\delta,c}$ into the boundary of $D_{\delta,c}.$ See Figure \ref{contours} (a). 

Let $\mathrm I_\delta$ be the straight line segment connecting $\tau_{i_4}+\mathbf i\delta$ and $\eta_4-\mathbf i\delta.$ Then by Proposition \ref{PL} and the construction of $\Gamma'$, we have 
$$\mathrm{Im}U_{\boldsymbol\alpha}(\xi)<\mathrm{Im}U_{\boldsymbol\alpha}(\xi^*)$$
for all $\xi\in \Gamma'-\mathrm I_\delta,$
and 
$$\mathrm{Im}U_{\boldsymbol\alpha}(\xi)=\mathrm{Im}U_{\boldsymbol\alpha}(\xi^*)$$
for all $\xi\in \mathrm I_\delta.$ Next we let $\Gamma^*$ be the contour obtained from $\Gamma'$ by replacing $\mathrm I_\delta$ by a new arc obtained from $\mathrm I_\delta$ by following the flow of the vector field $\mathbf v=\big(-\frac{\partial \mathrm{Im}U_{\boldsymbol\alpha}}{\partial \mathrm{Re}\xi},0\big)$ for a short time. See Figure \ref{contours} (b). Then by Proposition \ref{nonvanish}, we have
\begin{equation}\label{<0}
\mathrm{Im}U_{\boldsymbol\alpha}(\xi)<\mathrm{Im}U_{\boldsymbol\alpha}(\xi^*)
\end{equation}
for all $\xi\in \Gamma^*\setminus\{\xi^*\}.$
\medskip

\begin{figure}[htbp]
\centering
\includegraphics[scale=0.35]{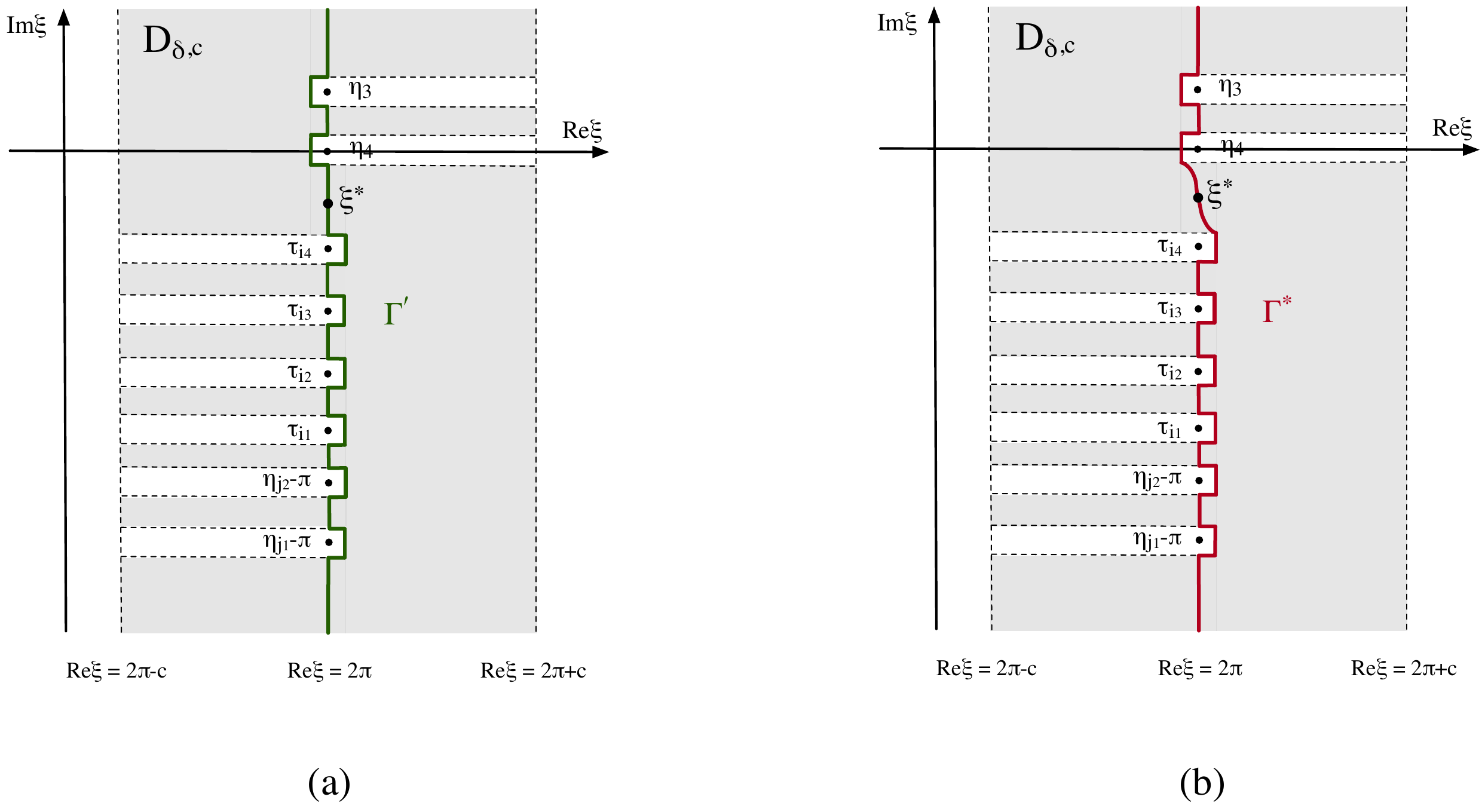}
\caption{Contour $\Gamma'$ and contour $\Gamma^*$.}
\label{contours}
\end{figure}

Now in (\ref{6jint}), we choose $\Gamma=\Gamma^*$ as the contour of integral.  Then we have
$$\bigg\{\begin{matrix} a_1 & a_2 & a_3 \\ a_4 & a_5 & a_6 \end{matrix} \bigg\}_b=\frac{1}{\pi b}\int_{\Gamma^*}\exp\bigg(\frac{U_{\boldsymbol \alpha}(\xi)+\kappa_{\boldsymbol \alpha}(\xi)b^2+\nu_{\boldsymbol \alpha,b}(\xi)b^4}{2\pi \mathbf i b^2} \bigg)d\xi.$$
Let $d>0$ be sufficiently small so that the region
$$B_d=\Big\{\xi\in \mathbb C\ \Big|\ |\mathrm{Re} \xi-\mathrm{Re} \xi^*|<d, |\mathrm{Im}\xi -\mathrm{Im} \xi^*|<d \Big\}$$
lies entirely in $D_{\delta,c},$ 
and 
$$\Gamma^*_d=\Gamma^*\cap B_d=\Big\{ \xi \in \Gamma^*\ \Big|\ |\mathrm{Im}\xi-\mathrm{Im}\xi^*|<d\Big\}.$$
Let $L>0$ be sufficiently large so that 
$$\mathrm{Im}\xi^*-L < \mathrm{Im}\eta_{j_1}-\delta \quad\text{and}\quad \mathrm{Im}\eta_{3}+\delta<\mathrm{Im}\xi^*+L,$$
and let
$$\Gamma^*_L=\Big\{ \xi \in \Gamma^*\ \Big|\ |\mathrm{Im}\xi-\mathrm{Im}\xi^*|\leqslant L\Big\}.$$

We will show that, as $b\to 0,$
\begin{enumerate}[(I)]
\item 
$$\frac{1}{\pi b}\int_{\Gamma^*_d}\exp\bigg(\frac{U_{\boldsymbol \alpha}(\xi)+ \kappa_{\boldsymbol \alpha}(\xi)b^2+ \nu_{\boldsymbol \alpha,b}(\xi)b^4}{2\pi \mathbf i b^2} \bigg)d\xi =  \frac{e^{\frac{-\mathrm{Vol}(\Delta)\cdot\mathbf i}{\pi b^2}}}{\sqrt[4]{-\det\mathrm{Gram}(\Delta)}}\Big(1+O\big(b^2\big)\Big),
$$

\item $$\bigg|\frac{1}{\pi b}\int_{\Gamma^*_L\setminus \Gamma^*_d}\exp\bigg(\frac{U_{\boldsymbol \alpha}(\xi)+ \kappa_{\boldsymbol \alpha}(\xi)b^2+ \nu_{\boldsymbol \alpha,b}(\xi)b^4}{2\pi \mathbf i b^2} \bigg)d\xi\bigg|< O\Big(e^{\frac{-\epsilon}{\pi b^2}}\Big),$$
and 

\item $$\bigg|\frac{1}{\pi b}\int_{\Gamma^*\setminus \Gamma^*_L}\exp\bigg(\frac{U_{\boldsymbol \alpha}(\xi)+ \kappa_{\boldsymbol \alpha}(\xi)b^2+ \nu_{\boldsymbol \alpha,b}(\xi)b^4}{2\pi \mathbf i b^2} \bigg)d\xi\bigg|<O\Big(e^{\frac{-\epsilon}{\pi b^2}}\Big)$$
\end{enumerate}
for some $\epsilon>0,$ from which the result follows.
\medskip

For (I), we claim that  all the conditions of Proposition \ref{saddle} are satisfied by letting $\hbar=b^2,$ $D=B_d,$ $f=\frac{U_{\boldsymbol \alpha}}{2\pi \mathbf i},$ $g=\exp\big(\frac{\kappa_{\boldsymbol \alpha}}{2\pi \mathbf i }\big),$ $f_\hbar=\frac{U_{\boldsymbol \alpha}+\nu_{\boldsymbol \alpha,b}b^4}{2\pi \mathbf i},$ $\upsilon_h=\frac{\nu_{\boldsymbol \alpha,b}}{2\pi \mathbf i},$ $S=\Gamma^*_d$ and $c=\xi^*.$ 

Indeed, by Proposition \ref{critical2}, $\xi^*$ is a critical point of $f=\frac{U_{\boldsymbol \alpha}}{2\pi \mathbf i}$ in $B_{d},$ hence condition (i) is satisfied. 

By the construction of $\Gamma^*$, $\xi^*$ is the unique maximum point of $\mathrm{Re}f=\frac{\mathrm{Im}U_{\boldsymbol\alpha}}{2\pi}$ on $\Gamma^*_d,$ hence condition (ii) is satisfied, and by Proposition \ref{Hess2}, condition (iii) is satisfied.

For condition (iv), since $\xi^*\neq \tau_i$ for any $i\in\{1,2,3,4\},$  $\xi^*\neq\eta_j-\pi$ for $j
\in\{1,2\}$, and $\xi^*\neq \eta_j$ for $j\in\{3,4\},$ $\kappa_{\boldsymbol \alpha}(\xi^*)$ is a finite value. As a consequence, $g(\xi^*)=\exp\big(\frac{\kappa_{\boldsymbol \alpha}(\xi^*)}{2\pi \mathbf i }\big)\neq 0,$ and condition (iv) is satisfied.

For condition (v), by Proposition \ref{bound2} (2),  $|\upsilon_{\hbar}(\xi)|=\big|\frac{\nu_{\boldsymbol \alpha,b}(\xi)}{2\pi \mathbf i}\big|<\frac{N}{2\pi}$ on $B_{d}.$ 

For condition (vi), since $\Gamma^*$ is obtained from  a straight line segment by following the flow of a smooth vector field, it is a smooth embedding near $\xi^*.$

Finally, by Proposition \ref{saddle}, Proposition \ref{critical2}  and Proposition \ref{Hess2}, we have as $b\to 0,$
\begin{equation*}
\begin{split}
\frac{1}{\pi b}\int_{\Gamma^*_d}\exp\bigg(\frac{U_{\boldsymbol \alpha,b}(\xi)}{2\pi \mathbf ib^2}\bigg) d\xi=& \frac{(2\pi b^2)^\frac{1}{2}}{\pi b} \frac{\exp\big(\frac{\kappa_{\boldsymbol \alpha}(\xi^*)}{2\pi \mathbf i}\big)}{\sqrt{-\frac{U_{\boldsymbol \alpha}''(\xi^*)}{2\pi \mathbf i}}}e^{\frac{U_{\boldsymbol\alpha}(\xi^*)}{2\pi \mathbf i b^2}}\Big(1+O\big(b^2\big)\Big)\\
=&\frac{e^{\frac{-\mathrm{Vol}(\Delta)\cdot\mathbf i}{\pi b^2}}}{\sqrt[4]{-\det\mathrm{Gram}(\Delta)}}\Big(1+O\big(b^2\big)\Big).
\end{split}
\end{equation*}
This completes the proof of (I). 
\medskip

For  (II),  by (\ref{<0}) and Proposition \ref{critical2}, $\mathrm{Im}U_{\boldsymbol\alpha}(\xi)<\mathrm{Im}U_{\boldsymbol\alpha}(\xi^*)=0$ for all $\xi\in\Gamma^*\setminus\{\xi^*\};$ and by the compactness of ${\Gamma^*_L}\setminus \Gamma^*_d,$ there exists an $\epsilon_1>0$ such that 
 \begin{equation}\label{ImU2}
\mathrm{Im} U_{\boldsymbol\alpha}(\xi)<-3\epsilon_1
\end{equation}
for all $\xi\in {\Gamma^*_L}\setminus \Gamma^*_d.$ Also by the compactness, there is an $M>0$ such that 
\begin{equation}\label{Imk2}
\mathrm{Im}\kappa_{\boldsymbol\alpha}(\xi)<M
\end{equation}
for all $\xi\in\Gamma^*_L\setminus \Gamma^*_d.$ By Proposition \ref{bound2} (2), there is a  $b_1>0$ such that 
\begin{equation}\label{last5}
\mathrm{Im}\nu_{\boldsymbol\alpha,b}(\xi)b^4<Nb^4<\epsilon_1
\end{equation}
for all $b<b_1$ and $\xi\in\Gamma^*;$ and together with  (\ref{ImU2}), we have
\begin{equation}\label{last4}
\mathrm{Im}U_{\boldsymbol \alpha}(\xi)+\mathrm{Im}  \nu_{\boldsymbol \alpha,b}(\xi)b^4<  - 2\epsilon_1
\end{equation}
for all $b<b_1$ and  $\xi\in\Gamma^*_L\setminus \Gamma^*_d.$  Putting (\ref{Imk2}) and  (\ref{last4}) together, we have
\begin{equation*}
\begin{split}
&\bigg|\frac{1}{\pi b}\int_{\Gamma^*_L\setminus \Gamma^*_d}\exp\bigg(\frac{U_{\boldsymbol \alpha}(\xi)+ \kappa_{\boldsymbol \alpha}(\xi)b^2+ \nu_{\boldsymbol \alpha,b}(\xi)b^4}{2\pi \mathbf i b^2} \bigg)d\xi\bigg|\\
 <  & \frac{2(L-d) e^{\frac{M}{2\pi}}}{\pi b} \exp\bigg(\frac{-\epsilon_1}{\pi b^2}   \bigg )\\
<&  O\Big(e^{\frac{-\epsilon}{\pi b^2}}\Big)
\end{split}
\end{equation*}
for any $\epsilon<\epsilon_1.$ This completes the proof of (II). 
\medskip

For (III), let $b_1$ and $\epsilon_1$ be as in the proof of (II) above. Then there  is a $b_0\in (0, b_1)$ such that for all $b<b_0,$
\begin{equation}\label{ImK2}
\mathrm{Im}\kappa_{\boldsymbol\alpha}(\xi^*\pm \mathbf i L) b^2<\epsilon_1,
\end{equation}
$Kb^2<\epsilon_1$ and $Nb^2<\epsilon_1,$ where $K$ and $N$ are respectively the constants in Proposition \ref{bound}. 

We claim that, for $\xi\in\Gamma^*\setminus \Gamma^*_L$ and $b<b_0,$ 
\begin{equation}\label{cl2}
\mathrm{Im}U_{\boldsymbol \alpha}(\xi)+\mathrm{Im}\kappa_{\boldsymbol \alpha}(\xi)b^2+\mathrm{Im}\nu_{\boldsymbol \alpha,b}(\xi)b^4<-2\epsilon_1-(4\pi-\epsilon_1)\big(|\xi-\xi^*|-L\big),
\end{equation}
as a consequence of which we have
\begin{equation}\label{CI2}
\begin{split}
& \frac{1}{\pi b}\int_{\Gamma^*\setminus \Gamma^*_L}\exp\bigg(\frac{\mathrm{Im}U_{\boldsymbol \alpha}(\xi)+\mathrm{Im}\kappa_{\boldsymbol \alpha}(\xi)b^2+\mathrm{Im}\nu_{\boldsymbol \alpha,b}(\xi)b^4}{2\pi b^2} \bigg)|d\xi| \\
 < & \frac{1}{\pi b} \exp\bigg(\frac{-\epsilon_2}{\pi b^2}\bigg)\int_{\Gamma^*\setminus \Gamma^*_L}\exp\bigg(\frac{-(4\pi-\epsilon_1)\big(|\xi-\xi^*|-L\big)}{2\pi}\bigg) |d\xi|\\
 < & O\Big(e^{\frac{-\epsilon}{\pi b^2}}\Big)
\end{split}
\end{equation}
for any $\epsilon<\epsilon_1.$ 

For the proof of the claim, by  Proposition \ref{PL} and (\ref{ImU2}),  for $l>L,$ i.e., $\xi^*\pm \mathbf il \in \Gamma^*\setminus\Gamma^*_L,$ we have
\begin{equation}\label{ImU2'}
\mathrm{Im}U_{\boldsymbol \alpha}(\xi^*\pm \mathbf il )\leqslant \mathrm{Im}U_{\boldsymbol\alpha}(\xi^*\pm \mathbf i L) < -3\epsilon_1; 
\end{equation} 
and by Proposition \ref{PL}  again and the choice of $L$ and $b_0,$ we have 
$$\frac{\partial}{\partial \mathrm{Im}\xi} \Big(\mathrm{Im}U_{\boldsymbol\alpha}(\xi^*+\mathbf il)+\mathrm{Im}\kappa_{\boldsymbol\alpha}(\xi^*+\mathbf il)b^2\Big)<-4\pi+Kb^2<-4\pi+\epsilon_1,$$
and 
$$\frac{\partial}{\partial \mathrm{Im}\xi} \Big(\mathrm{Im}U_{\boldsymbol\alpha}(\xi^*-\mathbf il)+\mathrm{Im}\kappa_{\boldsymbol\alpha}(\xi^*-\mathbf il)b^2\Big)>4\pi-Kb^2>4\pi-\epsilon_1.$$
Together with the Mean Value Theorem, (\ref{ImU2'}) and (\ref{ImK2}), we have 
\begin{equation}\label{Bou2}
\begin{split}
\mathrm{Im}U_{\boldsymbol\alpha}(\xi)+ \mathrm{Im}\kappa_{\boldsymbol\alpha}(\xi)b^2 < & \mathrm{Im}U_{\boldsymbol\alpha}(\xi^*\pm \mathbf iL) + \mathrm{Im}
\kappa_{\boldsymbol\alpha}(\xi^*\pm \mathbf iL)b^2 - (4\pi-\epsilon_1) \big |  \xi - (\xi^*\pm \mathbf iL ) \big|\\
< &-2\epsilon_1  -(4\pi-\epsilon_1)\big(|\xi-\xi^*|-L \big)
\end{split}
\end{equation}
for all $\xi \in \Gamma^*\setminus \Gamma^*_L.$  Finally, putting (\ref{Bou2}) and (\ref{last5}) together, we have (\ref{cl2}) and  the first  inequality in  (\ref{CI2}); and since 
$$|\xi-\xi^*|-L\to+\infty$$
as $\xi\in \Gamma^*\setminus \Gamma^*_L$ approaches $\infty,$ we have the second inequality in (\ref{CI2}). This completes the proof of (III).
\medskip

Putting (I), (II)  and (III) together, we have as $b\to 0,$ 
$$\bigg\{\begin{matrix} a_1 & a_2 & a_3 \\ a_4 & a_5 & a_6 \end{matrix} \bigg\}_b=\frac{e^{\frac{-\mathrm{Vol}(\Delta)\cdot\mathbf i}{\pi b^2}}}{\sqrt[4]{-\det\mathrm{Gram}(\Delta)}}\Big(1+O\big(b^2\big)\Big). $$ 
\end{proof}

As an immediate consequence of Theorem \ref{vol2} and Proposition \ref{reflection}, the reflection symmetry, we have

\begin{theorem}\label{vol2pm}
Let $(\theta_1,\dots, \theta_6)$ be the dihedral angles of  a truncated hyperideal anti-de Sitter tetrahedron $\Delta.$ Then  as $b\to 0,$
\begin{equation*}\label{AdS/CFT}
\bigg\{\begin{matrix}  \frac{Q}{2}\pm \frac{\theta_4}{2\pi b} & \frac{Q}{2}\pm \frac{\theta_5}{2\pi b} & \frac{Q}{2}\pm\frac{\theta_6}{2\pi b}\\\frac{Q}{2} \pm \frac{\theta_1}{2\pi b} & \frac{Q}{2}\pm\frac{\theta_2}{2\pi b} & \frac{Q}{2}\pm\frac{\theta_3}{2\pi b}  \end{matrix} \bigg\}_b=\frac{e^{\frac{-\mathrm{Vol}(\Delta)\cdot \mathbf i}{\pi b^2} }}{\sqrt[4]{-\det\mathrm{Gram}(\Delta)}} \Big(1 +O\big(b^2\big)\Big),
\end{equation*}
where in each entry of the $b$-$6j$ symbol the sign $+$ and $-$ can be chosen arbitrarily, $\mathrm{Vol}(\Delta)$ is the anti-de Sitter volume of $\Delta$, and $\mathrm{Gram}(\Delta)$ is the Gram matrix of $\Delta$ in the dihedral angles.
\end{theorem}

\bibliography{b-6j}
\bibliographystyle{abbrv}


\end{document}